\documentclass[a4paper,11pt]{article}
\pdfoutput=1 

\usepackage{jheppub} 
                     
\usepackage[T1]{fontenc} 
\graphicspath{{img/}} 

\usepackage{amsfonts,amssymb,amsthm,bbm,mathrsfs,enumitem}

\usepackage{amsmath}

\newcommand\utimes{\mathbin{\ooalign{$\cup$\cr%
   \hfil\raise0.42ex\hbox{$\scriptscriptstyle\times$}\hfil\cr}}}
\newcommand\bigutimes{\mathop{\ooalign{$\bigcup$\cr%
   \hfil\raise0.36ex\hbox{$\scriptscriptstyle\boldsymbol{\times}$}\hfil\cr}}}
   
\usepackage{tensor}

\allowdisplaybreaks[3]

\usepackage{hyperref}

\usepackage{subcaption}
\usepackage{color,psfrag}

\usepackage{tikz}
\usepackage{tikz-cd}
\usetikzlibrary{calc,matrix}
\usetikzlibrary{arrows,decorations.pathmorphing,decorations.pathreplacing,decorations.markings}
\usetikzlibrary{hobby,knots,celtic,shapes.geometric,calc}

\tikzset{
    edge/.style={draw, postaction={decorate},
        decoration={markings,mark=at position .55 with {\arrow{>}}}},
    linking-D/.style={draw, postaction={decorate},
        decoration={markings,mark=at position 1 with { rectangle, draw, inner sep=1pt, minimum size=2mm, fill=cyan }}},
}

\usepackage{mathtools}

\newcounter{footnotemarknum}

\newcommand{\C}{{\mathbb C}}
\newcommand{\N}{{\mathbb N}}

\newcommand{\cB}{{\mathcal B}}
\newcommand{\cE}{{\mathcal E}}
\newcommand{\cF}{{\mathcal F}}
\newcommand{\cG}{{\mathcal G}}

\newcommand{\cL}{{\mathcal L}}
\newcommand{\cH}{{\mathcal H}}

\newcommand{\cO}{{\mathcal O}}

\newcommand{\cV}{{\mathcal V}}

\newcommand{\cU}{{\mathcal U}}

\newcommand{\SU}{\mathrm{SU}}

\newcommand{\U}{\mathrm{U}}

\newcommand{\be}{\begin{equation}}
\newcommand{\ee}{\end{equation}}
\newcommand{\beq}{\begin{eqnarray}}
\newcommand{\eeq}{\end{eqnarray}}

\renewcommand{\u}{{\mathfrak{u}}}

\newcommand{\la}{\langle}
\newcommand{\ra}{\rangle}

\newcommand{\tr}{{\mathrm{Tr}}}
\newcommand{\f}{\frac}
\newcommand{\tl}{\widetilde}
\newcommand{\nn}{\nonumber}
\newcommand{\pp}{\partial}

\newcommand{\rd}{\mathrm{d}}

\newcommand{\eps}{\epsilon}

\newcommand{\id}{\mathbb{I}}

\newcommand{\act}{\triangleright}

\newcommand{\tk}{\tilde{k}}

\newcommand{\ri}{ { \rm{i} } }

\newcommand{\wh}{\widehat}

\newtheorem{theorem}{Theorem}[section]
\newtheorem{lemma}[theorem]{Lemma}
\newtheorem{res}[theorem]{Result}
\newtheorem{prop}[theorem]{Proposition}

\title{\boldmath Spin network entanglement: coarse-graining}

\author{Qian Chen}

\affiliation{ENSL, CNRS, Laboratoire de Physique, F-69342 Lyon, France}

\emailAdd{chenqian.phys@gmail.com}

\abstract{In loop quantum gravity, partitioning graph introduces boundaries and entanglement between spin sub-networks, reflecting non-local degrees of freedom and correlation amongst spatial regions. This gives rise to the view of coarse-graining, reducing the degrees of freedom that are unnecessary to be considered.
The present work sets coarse-graining in the framework of bulk-boundary relation. We investigates the spin network entanglement, showing that the entanglement is invariant under the coarse-graining at kinematical level. Moreover, we build the transformation between holonomy operators based on finer graph and coarser graph, and reveal the preservation of the coarse-graining method under the evolution generated by the holonomy operator. This leads to a holographical perspective for the entanglement issue in loop quantum gravity.
}

\begin{document} 
\maketitle
\flushbottom

\section{Introduction}
\label{sec:intro}
Loop Quantum Gravity (LQG) proposes a background independent framework for a theory of quantum general relativity (see \cite{Gaul:1999ys,Thiemann:2007zz,Rovelli:2014ssa,Bodendorfer:2016uat} for reviews), provides a local definition of quantum states of space geometry as spin networks and a canonical description for their dynamics through a Hamiltonian constraint. The geometry of 3d space slices is described by a pair of canonical fields, the (co-)triad and the Ashtekar-Barbero connection. Standard loop quantum gravity approach performs a canonical quantization of the holonomy-flux algebra, of observables smearing the Ashtekar-Barbero connection along 1d curves - holonomies, and the (co-)triad along 2d surfaces - fluxes, and defines quantum states of geometry as polymer structures or graph-like geometries with edges and vertices. Those spin network states represent the excitations of the Ashtekar-Barbero connection as Wilson lines in gauge field theory. On the other hand, geometric observables are raised to quantum operators acting on the Hilbert space spanned by spin networks, leading to the celebrated result of discrete spectra for areas (living along edges) and volumes (living at vertices) \cite{Rovelli:1994ge,Ashtekar:1996eg,Ashtekar:1997fb}.

So spin networks are the kinematical states of the theory and the game is to describe their dynamics, i.e. their evolution in time generated by the Hamiltonian constraints\footnotemark. 
\footnotetext{
Although the traditional canonical point of view is to attempt to discretize, regularize and quantize the Hamiltonian constraints \cite{Thiemann:1996aw,Thiemann:1996av}, this often leads to anomalies. The formalism naturally evolved towards a path integral formulation. The resulting spinfoam models, constructed from (extended) topological quantum field theories (TQFTs) with defects, define transition amplitudes for histories of spin networks \cite{Reisenberger:1996pu,Baez:1997zt,Barrett:1997gw,Freidel:1998pt} (see \cite{Livine:2010zx,Dupuis:2010kqh,Perez:2012wv} for reviews). The formalism then evolves in a third quantization, where so-called “group field theories” define non-perturbative sums over random spin network histories in a similar way than matrix model partition functions define sums over random 2d discrete surfaces \cite{DePietri:1999bx,Reisenberger:2000zc,Freidel:2005qe} (see \cite{Oriti:2006se,Carrozza:2013oiy,Oriti:2014yla} for reviews).
}
Evolving spin networks, formalized as spinfoams, describe the four-dimensional quantum space-time at the Planck scale. This quantum space-time is defined without reference to a background classical geometry, with quantum states of geometry are defined up to diffeomorphisms with no reference to any intrinsic coordinate system or background structure. Then concepts in classical geometry, such as distance, area, curvature, become emergent notions, in a continuum limit after suitably coarse-graining Planck scale quantum fluctuations. They can only be reconstructed from the interaction between subsystems, quantified by correlation and entanglement shared between subsystems (see e.g. \cite{Donnelly:2016auv,Feller:2017jqx}). This perspective sets the field of quantum information at the heart of research in quantum gravity, with essential roles to play for entanglement, decoherence and quantum localization in probing quantum states of geometries and thinking about the quantum-to-classical transition for the space-time geometry.

In the context of loop quantum gravity, the work investigating the entanglement carried by spin networks states have slowly built since the birth of the theory \cite{Livine:2005mw,Livine:2007sy,Donnelly:2008vx,Livine:2008iq,Donnelly:2011hn,Donnelly:2014gva,Feller:2015yta,Bianchi:2015fra,Feller:2016zuk,Bianchi:2016hmk,Delcamp:2016eya,Livine:2017fgq,Baytas:2018wjd}, but has definitely sped up in the past few years with the burst of interest in the bulk-to-boundary propagator and bulk-from-boundary reconstruction in the light of holography, see for instance \cite{Anza:2016fix,Chirco:2017xjb,Colafranceschi:2021acz,Chirco:2021chk,Colafranceschi:2022ual,Chen:2021vrc}.

The key function now played by the holographic principle as a guide for quantum gravity has put great emphasis of the role of boundaries. Although holography, inspired from black hole entropy and the AdS/CFT correspondence, can be initially thought as an asymptotic global property, recent researches on local area-entropy relations, holographic entanglement, holographic diamonds and the investigation of quasi-local holography and gravitational edge modes \cite{Donnelly:2016auv,Freidel:2016bxd,Freidel:2018pvm,Freidel:2019ees,Freidel:2020xyx,Freidel:2020svx,Freidel:2020ayo,Takayanagi:2019tvn} for finite boundaries necessarily pushes us to include (spatial) boundaries in the description of quantum geometries, not just as mere classical boundary conditions but as legitimate quantum boundary states. This translates a shift of perspective from a global description of space(-time) as a whole to a quasi-local description where any local bounded region of space(-time) is thought as an open quantum system.

\smallskip

The question --- where are the quantum gravity degrees of freedom, could be phrased to be a snapshot for the interlaced issues: the coarse-graining of quantum geometry states from the Planck scale to larger scales, the definition of quantum dynamics consistent with the holographic principle, and the implementation of (discretized) diffeomorphism at quantum level as the fundamental gauge symmetry of the theory (or, in other words, the implementation of a relativity principle for quantum geometry) \cite{Livine:2017xww}.

\smallskip

In this paper, we would like to discuss the relation between coarse-graining and and holography, and focus on spin network entanglement, i.e. the entanglement between spin sub-networks. Partitioning spin network inevitably introduces boundaries (e.g. Fig.\ref{fig:coarse-grain}), and every spin sub-network is a spin network with nonempty boundary evolving along time (e.g. Fig.\ref{fig:boundary}).
On the other hand, we can coarse-grain sub-network into simpler graph-structure, presenting the sub-network with a single vertex and boundary edges (e.g. Fig.\ref{fig:coarse-grain}). The coarse-grained spin network state can be defined based on the coarse-grained graph. We would like to answer the questions: (a) Can we study spin network entanglement from coarse-grained graph? (b) Can we study dynamics of spin network entanglement from the coarse-graining graph if we take evolution into account?
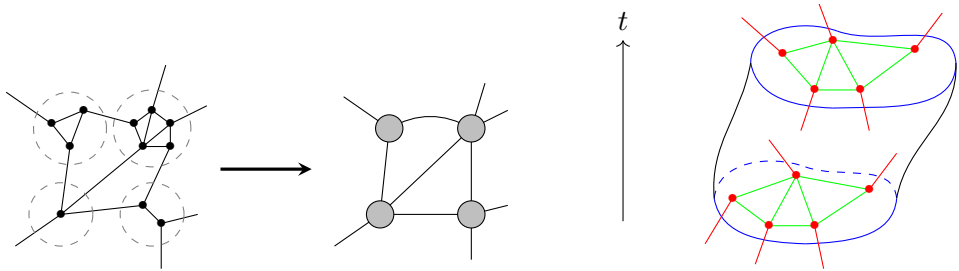
\begin{figure}[htb]
\begin{subfigure}[t]{0.45\linewidth}
\centering
\begin{tikzpicture}[scale=0.6]

\coordinate(a1) at (-0.2,0);
\coordinate(a2) at (0.5,0.3);
\coordinate(a3) at (0.2,-0.5);
\coordinate(b1) at (2,0.3);
\coordinate(b2) at (1.8,-0.5);
\coordinate(b3) at (2.4,0);
\coordinate(b4) at (2.4,-0.5);
\coordinate(b5) at (1.6,0);
\coordinate(c1) at (0,-2);
\coordinate(d1) at (1.8,-1.8);
\coordinate(d2) at (2.2,-2.2);

\draw (a1) -- (a2);
\draw (a1) -- (a3);
\draw (a3) -- (a2);
\draw (b1) -- (b2);
\draw (b1) -- (b3);
\draw (b1) -- (b5);
\draw (b2) -- (b4);
\draw (b2) -- (b5);
\draw (b3) -- (b2);
\draw (b3) -- (b4);
\draw (a2) -- (b5);
\draw (c1) -- (b2);
\draw (a3) -- (c1);
\draw (c1) -- (d1);
\draw (d1) -- (d2);
\draw (b4) -- (d1);
\draw (a1)--++(-1,.7);
\draw (c1)--++(-1,-.7);
\draw (b1)--++(.3,1);
\draw (b3)--++(.8,.4);
\draw (d2)--++(.8,.2);
\draw (d2)--++(0,-1);

\draw (a1) node[scale=0.7] {$\bullet$};
\draw (a2) node[scale=0.7] {$\bullet$};
\draw (a3) node[scale=0.7] {$\bullet$};
\draw (b1) node[scale=0.7] {$\bullet$};
\draw (b2) node[scale=0.7] {$\bullet$};
\draw (b3) node[scale=0.7] {$\bullet$};
\draw (b4) node[scale=0.7] {$\bullet$};
\draw (b5) node[scale=0.7] {$\bullet$};
\draw (c1) node[scale=0.7] {$\bullet$};
\draw (d1) node[scale=0.7] {$\bullet$};
\draw (d2) node[scale=0.7] {$\bullet$};

\draw[gray,dashed] (0.2,-0.1) circle(.8);
\draw[gray,dashed] (2,-0.1) circle(.85);
\draw[gray,dashed] (0,-2) circle(.7);
\draw[gray,dashed] (2,-2) circle(.7);

\draw[->,>=stealth,very thick] (3.5,-1) -- (5.5,-1);

\coordinate(a) at (7.2,-0.1);
\coordinate(b) at (9,-0.1);
\coordinate(c) at (7,-2);
\coordinate(d) at (9,-2);

\draw (a)--++(-1,.7);
\draw (c)--++(-1,-.7);
\draw (b)--++(.3,1);
\draw (b)--++(.8,.4);
\draw (d)--++(.8,.2);
\draw (d)--++(0,-1);

\draw (a) to[bend left] (b);
\draw (b) -- (c);
\draw (a) -- (c);
\draw (b) -- (d);
\draw (c) -- (d);

\draw[fill=lightgray] (a) circle(0.3);
\draw[fill=lightgray] (b) circle(0.3);
\draw[fill=lightgray] (c) circle(0.3);
\draw[fill=lightgray] (d) circle(0.3);

\end{tikzpicture}
\caption{An illustration of the partitioning of the graph in order to coarse-grain the spin network state.}
\label{fig:coarse-grain}
\end{subfigure}
\begin{subfigure}[t]{0.45\linewidth}
\centering
	\begin{tikzpicture} [scale=0.6]
	
	\draw[->] (-2,-0.5) -- (-2,3.5) node[above] {$t$};

\coordinate  (O1) at (0,0);
\coordinate  (O2) at (2,0.7);
\coordinate  (O3) at (4,0);
\coordinate  (O4) at (2,-1);

\coordinate  (P1) at (0.8,3);
\coordinate  (P2) at (2.7,3.7);
\coordinate  (P3) at (5.3,3);
\coordinate  (P4) at (2.6,2);

\draw[dashed,color=blue] (O1) to[out=85,in=160] (O2);
\draw[dashed,color=blue] (O2) to[out=-20,in=95] (O3);
\draw[color=blue]        (O3) to[out=-95,in=2] (O4);
\draw[color=blue]        (O4) to[out=178,in=-80] (O1);

\draw[color=blue] (P1) to[out=90,in=160] (P2);
\draw[color=blue] (P2) to[out=-20,in=95] (P3);
\draw[color=blue]        (P3) to[out=-90,in=2] (P4);
\draw[color=blue]        (P4) to[out=178,in=-80] (P1);

\draw        (O1) to[out=95,in=-100] (P1);
\draw        (O3) to[out=80,in=-87] (P3);

\coordinate  (A1) at (0.4,0);
\coordinate  (A2) at (1.8,0.5);
\coordinate  (A3) at (3.4,0.2);
\coordinate  (A4) at (2.2,-0.6);
\coordinate  (A5) at (1.2,-0.6);

\draw (A1) [color=red] -- ++ (-0.6,-1);
\draw (A2) [color=red] -- ++ (-0.6,0.8);
\draw (A3) [color=red] -- ++ (0.6,0.8);
\draw (A4) [color=red] -- ++ (0.2,-0.95);
\draw (A5) [color=red] -- ++ (-0.3,-0.85);

\draw (A1) [color=green] -- (A2);
\draw (A2) [color=green] -- (A3);
\draw (A3) [color=green] -- (A4);
\draw (A4) [color=green] -- (A5);
\draw (A5) [color=green] -- (A1);
\draw (A2) [color=green] -- (A4);
\draw (A2) [color=green] -- (A5);

\node[scale=0.7,color=red] at (A1) {$\bullet$};
\node[scale=0.7,color=red] at (A2) {$\bullet$};
\node[scale=0.7,color=red] at (A3) {$\bullet$};
\node[scale=0.7,color=red] at (A4) {$\bullet$};
\node[scale=0.7,color=red] at (A5) {$\bullet$};

\coordinate  (B1) at (1.5,3.2);
\coordinate  (B2) at (2.6,3.5);
\coordinate  (B3) at (4.4,3.3);
\coordinate  (B4) at (3.2,2.4);
\coordinate  (B5) at (2.2,2.4);

\draw (B1) [color=red] -- ++ (-0.9,0.8);
\draw (B2) [color=red] -- ++ (-0.3,0.8);
\draw (B3) [color=red] -- ++ (0.6,0.8);
\draw (B4) [color=red] -- ++ (0.2,-0.9);
\draw (B5) [color=red] -- ++ (-0.3,-0.9);

\draw (B1) [color=green] -- (B2);
\draw (B2) [color=green] -- (B3);
\draw (B3) [color=green] -- (B4);
\draw (B4) [color=green] -- (B5);
\draw (B5) [color=green] -- (B1);
\draw (B2) [color=green] -- (B4);
\draw (B2) [color=green] -- (B5);

\node[scale=0.7,color=red] at (B1) {$\bullet$};
\node[scale=0.7,color=red] at (B2) {$\bullet$};
\node[scale=0.7,color=red] at (B3) {$\bullet$};
\node[scale=0.7,color=red] at (B4) {$\bullet$};
\node[scale=0.7,color=red] at (B5) {$\bullet$};

	\end{tikzpicture}
	\caption{Spin sub-network of right-upper Fig.\ref{fig:coarse-grain} evolves along time.
	}
	\label{fig:boundary}
\end{subfigure}
\caption{Coarse-graining and spin network with non-empty boundary.
	}
\end{figure}

The answers are positive, at least for the cases that evolution is generated by loop holonomy operator, which can excite spin network entanglement \cite{Chen:2022rty}. We will see how to arrive the answers with the method of bulk-boundary maps \cite{Chen:2021vrc}. We wish this work to take a step forward in the exploration for coarse-graining and holography in quantum gravity.

\smallskip

Section \ref{Section:SpinNetwork-BoundaryMap} sets up the mathematical definitions of spin network Hilbert spaces, boundary Hilbert spaces and dual boundary Hilbert spaces. The coarse-graining method is defined via gauge transformation and gauge-fixing. The key property is that the coarse-graining preserves the scalar product for dual boundary Hilbert space. Based on it, we formulate the unitarity for dual boundary Hilbert space.
Section \ref{Section:SPNWsEntanglement} is dedicated to the analysis of spin network entanglement, i.e. entanglement between spin sub-networks. Our most important result of this part is that spin network entanglement can be coarse-grained, and the spin network entanglement can be exactly described by the entanglement between loopy intertwiners that presents spin sub-network state.
Section \ref{Section:DynamicsCoarseGraining} extends the power of the coarse-graining method for spin network entanglement: we investigate the situation that allows to implement dynamics with loop holonomy operator. The most important results in this part are two-fold: (a) The exact transformation is built between the loop holonomy operators on graph and the corresponding coarse-grained graph. (b) The coarse-graining for spin network entanglement still holds under the dynamics generated by the operator.
In a word, not only one can study the spin network entanglement from a simpler graph, but can also study the dynamics of the spin network entanglement from the simpler graph.

Section \ref{Section:Examples} applies these general results of coarse-graining for explicit examples, studying entanglement excitation generated by holonomy operator for the sake of looking into the dynamics of spin entanglement. We present examples with one-loop triangle graph and square graph, and a simple two-loop graph.


\subsection{Conventions}
Here we summarize the conventions used throughout the paper.
\begin{itemize}
\item[$-$] $\Gamma$ generic graph. $\Gamma^{o}$ the bulk of $\Gamma$. $\pp\Gamma$ the boundary of $\Gamma$. For the cases that two spin networks $\Gamma$ and $\tl{\Gamma}$ have identical boundary, we usually denote the boundary by $\cB$, i.e. $\cB=\pp\Gamma=\pp\tl{\Gamma}$.

\item[$-$] $\Upsilon$ loopy graph (one vertex plus boundary edges and likely self-loops). 

\item[$-$] $\cH_{\Gamma}$ spin network Hilbert space based on $\Gamma$. Relevant notations:
\begin{itemize}
\item[$\bullet$] $| \psi_{\Gamma} \ra$ spin network state. $\psi_{\Gamma}( \{ g_{e} \})$ spin network wave-function (depending on the holonomies) that corresponds to the $| \psi_{\Gamma} \ra$.

\item[$\bullet$] $| \Psi_{\Gamma,\{j_{e},I_{v}\}} \ra$ spin network basis state. $\Psi_{\Gamma,\{j_{e},I_{v}\}}(\{g_{e}\})$ spin network wave-function for the spin network basis state. For alleviation, we usually use $| \Psi_{\Gamma,\{I_{v}\}} \ra \equiv | \Psi_{\Gamma,\{j_{e},I_{v}\}} \ra$ since spin-labels are included in intertwiner-labels.

\item[$\bullet$] Scalar product for $\cH_{\Gamma}$, e.g. $\la \phi_{\Gamma} | \psi_{\Gamma} \ra$.
\end{itemize}

\item[$-$] $\cH_{\pp\Gamma}$ boundary Hilbert space for $\pp\Gamma$.
Relevant notations:
\begin{itemize}
\item[$\bullet$] Sometimes we talk about $\cH_{\cB}$, especially when we do not want to emphasize the bulk $\Gamma$, since there are situations that $\cB=\pp\Gamma=\pp\tl{\Gamma}$ but $\Gamma \neq \tl{\Gamma}$. Note that $\cH_{\cB} \cong \cH_{\pp\Gamma} \cong \cH_{\pp\tl{\Gamma}}$.

\item[$\bullet$] Scalar product for $\cH_{\cB}$, e.g. $\underline{ \la \phi_{\cB} | \psi_{\cB} \ra }$, for $| \phi_{\cB} \ra\,, | \psi_{\cB} \ra \in \cH_{\cB}$.

\end{itemize}
\item[$-$] $(\cH_{\pp\Gamma} )^{*}$ dual boundary Hilbert space, to $\cH_{\pp\Gamma}$.
Relevant notations:
\begin{itemize}
\item[$\bullet$] $|\psi_{\pp\Gamma}(\{g_{e}\}_{e\in\Gamma^{o}} )\ra \in \cH_{\Gamma}$ boundary state that is determined by spin network wave-function $\psi_{\Gamma}(\{g_{e}\}_{e\in\Gamma^{o}} )$. We also usually refer the $|\psi_{\pp\Gamma}(\{g_{e}\}_{e\in\Gamma^{o}} )\ra$ to `bulk-boundary map' as a member of dual boundary Hilbert space.

\item[$\bullet$] Scalar product for $( \cH_{\pp\Gamma} )^{*}$, e.g. $\la \phi_{\pp\Gamma} | \psi_{\pp\Gamma} \ra$ for bulk-boundary maps based on spin network states $| \phi_{\Gamma} \ra$ and $| \psi_{\Gamma} \ra$. Note that the holonomies have been integrated.

\item[$\bullet$] The dual boundary Hilbert space $(\cH_{\cB})^{*}$ is the direct sum Hilbert space over all $\Gamma$ with $\pp\Gamma=\cB$, see eq.(\ref{eq:DirectSum:DualBoundaryHilbert}).

\end{itemize}

\end{itemize}


\section{Dual boundary Hilbert space --- bulk-boundary maps} \label{Section:SpinNetwork-BoundaryMap}
We begin with a review for some relevant aspects of the techniques of bulk-boundary maps, focusing particularly on properties of unitarity for dual boundary Hilbert space.

\subsection{Spin networks with non-empty boundary} \label{subsection:SP-SNWF}
Loop quantum gravity proceeds to a canonical quantization of general relativity (see \cite{Thiemann:2007zz} for detailed lectures, or \cite{Ashtekar:2021kfp} for a recent overview), describing the evolution of a 3d (space-like) slice in time, thereby generating the four-dimensional space-time. It defines quantum states of geometry and describes their constrained evolution in time.
A state of geometry is defined as a wave-function $\psi$ of the Ashtekar-Barbero connection on the canonical hypersurface. Loop quantum gravity choose cylindrical wave-functions, that depend on the holonomies of that connection along the edges of a graph $\Gamma$. These holonomies are $\SU(2)$ group elements,  $g_{e}\in\SU(2)$ for each edge $e$ of the graph. And the wave-functions are required to be gauge-invariance under local $\SU(2)$ transformations, which act at every vertices of the graph.

For closed 3d spatial slices, without boundary, we consider closed graphs, i.e. without open links.
A wave-function $\psi$ on a closed oriented graph $\Gamma$ is a function of one $\SU(2)$ group element $g_{e}$ for each edge $e$, and is assumed to be invariant under the $\SU(2)$-action at each vertex $v$ of the graph:
\be
\label{eq:GaugeTransformation}
\begin{array}{rlcl}
\psi_{\Gamma}: 
&\SU(2)^{\times E}
&\longrightarrow
&\C \\
&\{g_{e}\}_{e\in\Gamma}
&\longmapsto
&\psi(\{g_{e}\}_{e\in\Gamma})
=
\psi
(\{h_{t(e)}g_{e}h_{s(e)}^{-1}\}_{e\in\Gamma})\,, \quad
\forall h_{v}\in\SU(2)\,
\end{array}
%
\ee
where $t(e)$ and $s(e)$ respectively refer to the target and source vertices of the edge $e$. We write $E$ and $V$ respectively for the number of edges and vertices of the considered graph $\Gamma$.
The scalar product between such wave-functions is given by the Haar measure on the Lie group $\SU(2)$:
\be \label{eq:ScalarProduct-SpinNetwork-Integration}
\la \psi_{\Gamma}
| \tl{\psi}_{{\Gamma} } \ra
=\int_{\SU(2)^{{\times E}}}\prod_{e}\rd g_{e}\,
\overline{\psi
(\{g_{e}\}_{e\in\Gamma}) }\,
\tl{\psi}
(\{g_{e}\}_{e\in{\Gamma} })
\,.
\ee
The Hilbert space of quantum states with support on the closed graph $\Gamma$ is thus realized as a space of square-integrable functions, $\cH_{\Gamma}=L^{2}(\SU(2)^{{\times E}}/\SU(2)^{{\times V}})$.

\smallskip

For a 3d slice with boundary, we consider graphs with open links puncturing the slice boundary. Those open links are connected to one vertex of the graph, while their other extremity is left loose. Without loss of generality, we can assume that all those open links are outward oriented, i.e. that their source vertex belongs to the graph, while their target vertex does not. Those open links are referred to as the boundary links.
We can use the same definition as above for a closed graph, considering wave-functions of one group element per link, including both the standard links in the interior and the boundary links. The difference is that gauge transformations will only act at the graph vertices and will not act on the open ends.

To proceed, we consider introducing spin networks on a bounded spatial slice and view spin networks as bulk-boundary maps. As advocated in \cite{Chen:2021vrc}, one explicitly partitions the graph set of edges into interior links and boundary links,
\be
\Gamma= \Gamma^{o}\sqcup \pp\Gamma\,.
\ee
The $\Gamma^{o}=\Gamma\setminus\pp\Gamma$ is referred to as the bulk or interior of the graph $\Gamma$, and $\pp\Gamma$ is referred to as the boundary or exterior of the graph $\Gamma$.
Each boundary edge $e\in\pp\Gamma$ carries a spin $j_{e}$ and a vector in the corresponding representation $| j_{e}, m_{e} \ra \in \cV_{j_{e}}$ (for inward oriented $| j_{e}, m_{e} ] \in \cV^{*}_{j_{e}}$). As explained in  \cite{Chen:2021vrc}, the boundary Hilbert space is the tensor product of spin states living on the open links
\begin{alignat}{2}
\cH_{\cB}
=
\cH_{\pp\Gamma}
=&\bigotimes_{e\in\pp\Gamma}\cH_{e}
\qquad
&\textrm{with}\quad
&\cH_{e}=\bigoplus_{j_{e}\in\f\N2}\cV_{j_{e}}\,,
\\
=&\bigoplus_{\{j_{e}\}}\cH^{\{j_{e}\}_{e\in\pp\Gamma}}_{\Gamma}
\qquad
&\textrm{with}\quad
&\cH^{\{j_{e}\}_{e\in\pp\Gamma}}_{\pp\Gamma}=\bigotimes_{e\in\pp\Gamma}\cV_{j_{e}}\,.
\end{alignat}
A spin network wave-function on the graph $\Gamma$ with boundary is still a function of group elements living on bulk edges $e\in\Gamma^{o}$, but is not anymore valued in the field $\C$ but into the boundary Hilbert space $\cH_{\pp\Gamma}$:
\be
\label{eq:Bulk-BoundaryMaps}
\begin{array}{rlcl}
\psi_{\Gamma}: 
&\SU(2)^{\{e\in \Gamma^{o}\}}
&\longrightarrow
&\cH_{\pp\Gamma}\,, \\
&\{g_{e}\}_{e\in\Gamma^{o}}
&\longmapsto
&| \psi_{\pp\Gamma} (\{g_{e}\}_{e\in\Gamma^{o}}) \ra \,\in\,\cH_{\pp\Gamma}\,.
\end{array}
\ee
That is, a spin network wave-function provides a bulk-boundary map to $\cH_{\pp\Gamma}$. We refer the interested reader to \cite{Chen:2021vrc}, and to the recent work \cite{Anza:2017dkd,Colafranceschi:2021acz,Colafranceschi:2022ual}, for more details on wave-functions for bounded regions and interesting work on typical bulk states.

\smallskip

Let us now briefly review spin network states and notations.
A basis of this Hilbert space can be constructed using the spin decomposition of $L^{2}$ functions on the Lie group $\SU(2)$ according to the Peter-Weyl theorem. A {\it spin} $j\in\f\N2$ defines an irreducible unitary representation of $\SU(2)$, with the action of $\SU(2)$ group elements realized on a $(2j+1)$-dimensional Hilbert space $\cV_{j}$. We use the standard orthonormal basis $|j,m\ra$, labeled by the spin $j$ and the magnetic index $m$ running by integer steps from $-j$ to $+j$, which diagonalizes the $\SU(2)$ Casimir $\vec{J}^{2}$ and the $\u(1)$ generator $J_{z}$. Group elements $g$ are then represented by the $(2j+1)\times (2j+1)$ Wigner matrices $D^{j}(g)$:
\be
D^{j}_{mm'}(g)=\la j,m|g|j,m'\ra\,,\qquad
\overline{D^{j}_{mm'}(g)}
=
D^{j}_{m'm}(g^{-1})
\,.
\ee
These Wigner matrices form an orthogonal basis of $L^{2}(\SU(2))$:
\be
\int_{\SU(2)}\rd g\,\overline{D^{j}_{ab}(g)}\,{D^{k}_{cd}(g)}
=
\f{\delta_{jk}}{2j+1}\delta_{ac}\delta_{bd}
\,,\qquad
\delta(g)
=\sum_{j\in\f\N2}(2j+1)\chi_{j}(g)
\,, \label{eq:Peter-Weyl}
\ee
where $\chi_{j}$ is the spin-$j$ character defined as the trace of the Wigner matrix, $\chi_{j}(g)=\tr D^{j}(g)=\sum_{m=-j}^{j}\la j,m|g|j,m\ra$.
Applying this to gauge-invariant wave-functions allows to build the {\it spin network} basis states of $\cH_{\Gamma}$, which depend on one spin $j_{e}$ on each edge and one intertwiner $I_{v}$ at each vertex:
\be
\Psi_{\Gamma,\{j_{e},I_{v}\}}(\{g_{e}\}_{e\in\Gamma})
=
\sum_{m_{e}^{t,s}}
\prod_{e}\sqrt{2j_{e}+1}\,\la j_{e}m_{e}^{t}|g_{e}|j_{e}m_{e}^{s}\ra
\,\prod_{v} \la \bigotimes_{e|\,v=s(e)} j_{e}m_{e}^{s}|\,I_{v}\,|\bigotimes_{e|\,v=t(e)}j_{e}m_{e}^{t}\ra
\,. \label{eq:SpinNetwork}
\ee
As illustrated on fig.\ref{fig:intertwiner}, an {\it intertwiner} is a $\SU(2)$-invariant  state -or singlet- living in the tensor product of the incoming and outgoing spins at the vertex $v$:
\be \label{eq:IntertwinerHilbertSpace}
I_{v}\in\textrm{Inv}_{\SU(2)}\Big{[}
\bigotimes_{e|\,v=s(e)} \cV_{j_{e}}
\otimes
\bigotimes_{e|\,v=t(e)} \cV_{j_{e}}^{*}
\Big{]}
\,.
\ee
The intertwiners are singlet states, recoupling $\SU(2)$ irreducible representations. The simplest cases are bivalent and trivalent intertwiner which are uniquely determined when spins are given, since there is an unique way to recouple two or three spins into a singlet state - null angular momentum. For higher valent vertex, to label a recoupling amongst spins, one needs to choose a channel to recouple them. This amounts to split a higher valent vertex into trivalent vertices and link them via a tree graph (no forming loop), e.g., fig.\ref{fig:Unfolding-5Valent}.
Once such channel, or tree, is chosen, an intertwiner basis state is defined by the assignment of a spin to each internal link, or intermediate spins. Two different intertwiner basis states that have different intermediate spins (i.e. at least one different intermediate spins) are mutually orthogonal. An generic intertwtiner state will then be a arbitrary superposition of those basis states.
There are various possible ways to choose a channel to label intertwiner basis states. The unitary map between basis associated to different channels is given by the $\{3nj\}$ spin recoupling symbols.

\smallskip

We will alleviate the notation for intertwiner.
An intertwiner $\vert \{ j_e\}_{e\ni v} \,, I_{v}^{ ( \{ j_e\}_{ e\ni v } ) } \ra$ is labeled by $\{ j_e\}_{e\ni v}$, the spins attached, and $I_{v}^{ ( \{ j_e\}_{ e\ni v } ) }$ the internal indices when attached spins $\{ j_e\}_{e\ni v}$ are given. For instance, a trivalent intertwiner only needs attached spins $\{ j_1,j_2,j_3 \}$ for labeling, while a four-valent intertwiner needs $\{ j_1,j_2,j_3,j_4\}$ for attached spins plus an internal index $j_{12}$ for recoupled spin of $\cV_{j_1} \otimes \cV_{j_2}$, likewise, for five-valent vertex unfolding in fig.\ref{fig:Unfolding-5Valent}, attached spins $\{ j_1,\cdots, j_5 \}$ and internal indices $\{ j_{12}, j_{45} \}$ are needed for labeling.
For each vertex $v$, the attached spins are implicitly expressed in internal indices $I_{v}^{ ( \{ j_e\}_{ e\ni v } ) }$, we don't need to explicitly specify the attached spins, and adopt $I_{v}$ simply instead of $I_{v}^{ ( \{ j_e\}_{ e\ni v } ) }$, unless in some necessary cases. Hence from now on, $\vert \{ j_e\}_{e\ni v} \,, I_{v}^{ ( \{ j_e\}_{ e\ni v } ) } \ra \equiv | I_{v} \ra$.
Under the alleviation, the scalar product between two spin network basis states on the same graph $\Gamma$ is then given by the product of the scalar products between their intertwiners:
\be \label{eq:ScarlarProduct-SpinNetwork-BasisStates}
\la \Psi_{ \Gamma, \{j_{e},I_{v}\} } | \Psi_{ \Gamma, \{\tilde{j}_{e},\tilde{I}_{v}\}} \ra
=\prod_{e}
\delta_{j_{e},\tilde{j}_{e}} \,
\prod_{v}
\la I_{v}|\tilde{I}_{v}\ra
\equiv
\prod_{v}
\la I_{v}|\tilde{I}_{v}\ra
\,.
\ee

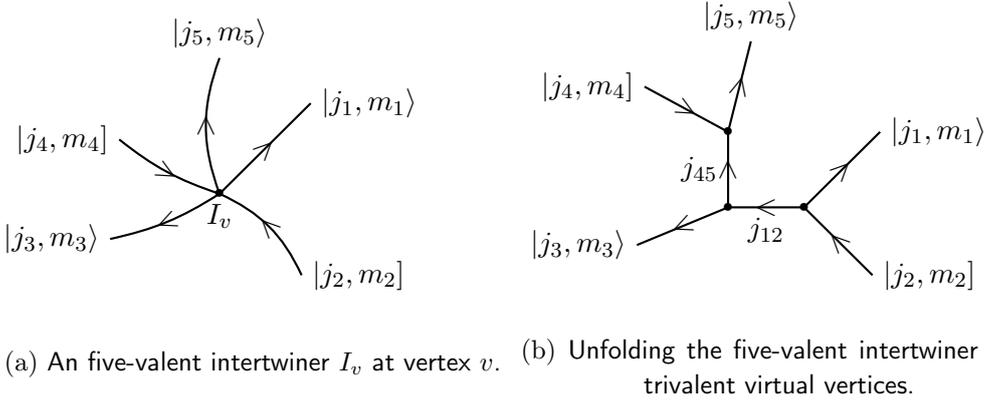
\begin{figure}[htb]
	\centering
\begin{subfigure}[t]{0.45\linewidth}
	\begin{tikzpicture} [scale=1.2]
\coordinate  (O) at (0,0);

\node[scale=0.7] at (O) {$\bullet$} node[below] {$I_v$};

\draw[thick] (O)  --  node[midway,sloped]{$>$} ++ (1,1) node[right] {$| j_1, m_1 \ra$};

\draw[thick] (O)  to[bend left=20]  node[midway,sloped]{$<$} ++ (0.9,-0.9) node[right] {$| j_2, m_2 ]$};

\draw[thick] (O)  to[bend left=20] node[midway,sloped]{$>$} ++ (0,1.5) node[above] {$| j_5, m_5 \ra$};

\draw[thick] (O)  to[bend left=10]  node[midway,sloped]{$<$} ++ (-1.2,-0.5) node[left] {$| j_3, m_3 \ra$};

\draw[thick] (O)  to[bend left=10]  node[midway,sloped]{$>$} ++ (-1.1,0.6) node[left] {$| j_4, m_4 ]$};

\end{tikzpicture}
\caption{An five-valent intertwiner $I_v$ at vertex $v$.}
\label{fig:intertwiner}
\end{subfigure}
\begin{subfigure}[t]{0.45\linewidth}
\begin{tikzpicture}[scale=1]

\coordinate  (O) at (0,0);
\coordinate  (A) at (-1,0);
\coordinate  (B) at (-1,1);

\node[scale=0.7] at (O) {$\bullet$};
\node[scale=0.7] at (A) {$\bullet$};
\node[scale=0.7] at (B) {$\bullet$};

\draw[thick] (O) --node[midway,sloped]{$<$} node[midway,below] {$j_{12}$} (A) --node[midway,sloped]{$>$} node[midway,left] {$j_{45}$} (B);

\draw[thick] (O)  --  node[midway,sloped]{$>$} ++ (1,1) node[right] {$| j_1, m_1 \ra$};

\draw[thick] (O)  --  node[midway,sloped]{$<$} ++ (0.9,-0.9) node[right] {$| j_2, m_2 ]$};

\draw[thick] (B)  -- node[midway,sloped]{$>$} ++ (0.3,1.2) node[above] {$| j_5, m_5 \ra$};

\draw[thick] (A)  -- node[midway,sloped]{$<$} ++ (-1.2,-0.5) node[left] {$| j_3, m_3 \ra$};

\draw[thick]  (B)  -- node[midway,sloped]{$>$} ++ (-1.1,0.6) node[left] {$| j_4, m_4 ]$};

\end{tikzpicture}

\caption{Unfolding the five-valent intertwiner with trivalent virtual vertices.
}
\label{fig:Unfolding-5Valent}
\end{subfigure}
\caption{The notation for a higher valent intertwiner $I_{v}$ in terms of virtual spins.
}
\label{fig:Indices-5Valent}
\end{figure}

In the next subsection, we build Hilbert space of bulk-boundary maps with scalar product of spin network Hilbert space.




\subsection{Scalar product for dual boundary Hilbert space} \label{subsection:SP-BMs}
In principle, it is free to view bulk-boundary maps as spin network wave-functions. This subsection is meant to clarify some details for establishing the Hilbert space of bulk-boundary maps. Especially, we deal with the scalar products of boundary Hilbert space and of bulk Hilbert space separately.

We first clarify the scalar product for boundary Hilbert space $\cH_{\pp\Gamma}$ or $\cH_{\cB}$. Now $\cH_{\pp\Gamma}$ is the tensor product of spin states living on the boundary edges. Suppose $\{ | j_{e} m_{e} \ra \}$ basis states for a $\cV_{j_{e}}$, then $\underline{ \la j_{e} m_{e} | j'_{e} m'_{e} \ra }=\delta_{j_{e} j'_{e} } \delta_{m_{e} m'_{e} }$. We adopt `$\underline{\la \, \ | \, \ \ra}$' to denote the scalar product for the $\cH_{\cB}$.

\smallskip

Bulk-boundary maps are considered as linear form living in dual Hilbert space $(\cH_{\pp\Gamma} )^{*}$. To see this, consider any boundary state $| \Phi_{\cB} \ra$. The boundary Hilbert space's scalar product for $| \Phi_{\cB} \ra$ and $| \psi_{\pp\Gamma} (\{g_{e}\}_{e\in\Gamma^{o}}) \ra$ defines a distribution on boundary states depending on the group elements, or holonomies,
\be
\forall\, | \Phi_{\cB} \ra \in \cH_{\cB}
\,, \qquad
\underline{
\la \Phi_{\cB} \vert \psi_{\pp\Gamma}
(\{g_{e}\}_{e\in\Gamma^{o}}) \ra
}
\in \C
\,.
\ee
One does not need to integrate any bulk holonomies. The complex number is evaluated by assigning every $g_{e}$. To illustrate, there are three situations:
\begin{itemize}
\item
Assume $| \Phi_{\cB} \ra$ does not depend on any holonomies at all. That is, $| \Phi_{\cB} \ra$ allows for a linear combination in terms of the basis vectors $\cH_{\cB}$. In that case, the coefficients of the linear combination are complex numbers without dependency of holonomies, thus
\be
| \Phi_{\cB} \ra=\bigotimes_{e\in\pp\Gamma} | j_e m_e \ra\,, \qquad
\underline{
\la \Phi_{\cB} | \psi_{\pp\Gamma} (\{g_{e}\}_{e\in\Gamma^{o}}) \ra
} \in \cF\big{[}\SU(2)^{\{e\in\Gamma^{o}\}}\big{]}\,,
\ee
i.e., the dependency of holonomies are all due to $| \psi_{\pp\Gamma} (\{g_{e}\}_{e\in\Gamma^{o}}) \ra$. 

\item
Assume $| \Phi_{\cB} \ra=| \phi_{\pp\Gamma} (\{g_{e}\}_{e\in\Gamma^{o}}) \ra$ is a boundary state that corresponds to another spin network wave-function $\phi_{\Gamma}$ based on the same graph. Again, the $| \Phi_{\cB} \ra$ allows for a linear combination in terms of the basis vectors $\cH_{\cB}$. In that case, the linear combination coefficients depend on holonomies, thus
\be \label{eq:FunctionsHolonomies-sameGraph}
| \Phi_{\cB} \ra=| \phi_{\pp\Gamma} (\{g_{e}\}_{e\in\Gamma^{o}}) \ra 
\,, \qquad
\underline{
\la \phi_{\pp\Gamma}(\{g_{e}\}_{e\in\Gamma^{o}}) | \psi_{\pp\Gamma} (\{g_{e}\}_{e\in\Gamma^{o}}) \ra
} \in \cF\big{[}\SU(2)^{\{e\in\cL\}}\big{]}\,,
\ee
where the $\cL$ is the set of loop edges, with less number of holonomies than the number of bulk edges. We will explain this in the next subsection.

\item
Assume $| \Phi_{\cB} \ra=| \phi_{\pp\tl{\Gamma}} (\{g_{ \tl{e} }\}_{\tl{e}\in\tl{\Gamma}^{o}}) \ra$ is a boundary state that corresponds to spin network wave-function $\phi_{\tl{\Gamma}}$ based on a different graph $\tl{\Gamma}$, i.e., $\tl{\Gamma} \neq \Gamma$. The number of holonomies is more than the number of $\Gamma^{o}$ or $\tl{\Gamma}^{o}$, i.e.,
\be
| \Phi_{\cB} \ra=| \phi_{\pp\tl{\Gamma}} (\{g_{ \tl{e} }\}_{\tl{e}\in\tl{\Gamma}^{o}}) \ra\,, \qquad
\underline{\la \phi_{\pp\tl{\Gamma}} (\{g_{ \tl{e} }\}_{\tl{e}\in\tl{\Gamma}^{o}}) | \psi_{\pp\Gamma} (\{g_{e}\}_{e\in\Gamma^{o}}) \ra
} \in \cF\big{[}\SU(2)^{\{e\in\Gamma^{o} \cup \tl{\Gamma}^{o} \}}\big{]}\,.
\ee
\end{itemize}
The $| \Phi_{\cB} \ra$ in the first situation is considered as `no bulk degree of freedom at all', while the $| \Phi_{\cB} \ra$ in the second and third situations are considered as carrying bulk degrees of freedom.

\smallskip

Overall, the scalar product for $\cH_{\cB}$ gives a distribution depending on holonomies. We are able to define the scalar product thus to establish the dual boundary Hilbert space $(\cH_{\cB})^{*}$ via integration over holonomies. 
We denote $| \psi_{\pp\Gamma} \ra$ the member of $(\cH_{\cB})^{*}$. A $| \psi_{\pp\Gamma} \ra$ corresponds to a boundary state $| \psi_{\pp\Gamma} (\{g_{e}\}_{e\in\Gamma^{o} }) \ra \in \cH_{\pp\Gamma}$ depending on holonomies, or equivalently, corresponds to spin network wave-function $\psi_{\Gamma} (\{g_{e}\}_{e\in\Gamma^{o} })$. For any two $| \phi_{\pp\tl{\Gamma} } \ra\,, | \psi_{\pp\Gamma} \ra \in (\cH_{\cB})^{*}$, the scalar product for the dual boundary Hilbert space $(\cH_{\cB})^{*}$ is:
\be \label{eq:Definition-ScalarProduct-BoundaryMaps}
\left\la \phi_{\pp\tl{\Gamma}} \middle\vert \psi_{\pp\Gamma} \right\ra
=
\left\{
\begin{array}{rlcl}
&
\displaystyle{
\int
\prod_{e\in\Gamma^{o}}\rd g_{e}\,
\underline{
\la \phi_{\pp\Gamma} (\{g_{e}\}_{e\in\Gamma^{o}}) \vert \psi_{\pp\Gamma} (\{g_{e}\}_{e\in\Gamma^{o} }) \ra
}
}
\,,
&
\Gamma=\tl{\Gamma}\,, 
\\
&
0
\,,
&
\pp\Gamma=\pp\tl{\Gamma}\,, \quad
\Gamma^{o} \neq \tl{\Gamma}^{o}\,.
\end{array}
\right.
\ee
The first equation of eq.(\ref{eq:Definition-ScalarProduct-BoundaryMaps}) is equivalent to the scalar product for $\cH_{\Gamma}$ defined by eq.(\ref{eq:ScalarProduct-SpinNetwork-Integration})
\beq
\left\la \phi_{\pp\Gamma} \middle\vert \psi_{\pp\Gamma} \right\ra
&=&
\int \prod_{e\in\Gamma^{o}}\rd g_{e}\,
\underline{
\la \phi_{\pp\Gamma} (\{g_{e}\}_{e\in\Gamma^{o}}) \vert \psi_{\pp\Gamma} (\{g_{e}\}_{e\in\Gamma^{o} }) \ra
}
\nn \\
&=&\int \prod_{e\in\Gamma}\rd g_{e}\,
\overline{\phi_{ \Gamma } (\{g_{e}\}_{e\in\Gamma}) }\,
\psi_{\Gamma} (\{g_{e}\}_{e\in{\Gamma} })
=
\left\la \phi_{\Gamma} \middle\vert \psi_{\Gamma} \right\ra
\,.
\eeq
Note that the integrals for boundary holonomies in the second line are equivalent to the scalar product for $\cH_{\pp\Gamma}$ in the first line. In fact, it does not matter to add boundary holonomies for the first line since they are boundary unitaries and $\underline{\la jm | g^{\dagger} g | j'm'\ra}=\underline{\la jm | j'm'\ra}=\delta_{jj'}\delta_{mm'}$.
The second equation of eq.(\ref{eq:Definition-ScalarProduct-BoundaryMaps}) is due to the graph orthogonality in LQG. Therefore, we define dual boundary Hilbert space $(\cH_{\cB})^{*}$ via decomposition
\be \label{eq:DirectSum:DualBoundaryHilbert}
(\cH_{\cB})^{*}
=
\bigoplus_{\Gamma |\, \pp\Gamma=\cB} (\cH_{\pp\Gamma})^{*}
\,.
\ee
The normalization of $| \psi_{\pp\Gamma} \ra$ is expressed as
\be
\left\la \psi_{\pp\Gamma} \middle\vert {\psi}_{\pp\Gamma} \right\ra
=
\int
\prod_{e\in\Gamma^{o}}\rd g_{e}\,
\underline{ \la \psi_{\Gamma} (\{g_{e}\}_{e\in\Gamma^{o}}) \vert \psi_{\Gamma} (\{g_{e}\}_{e\in\Gamma^{o}}) \ra}
=1
\,.
\ee

\smallskip

Up to now, we have defined the scalar product for dual boundary Hilbert space $(\cH_{\cB})^{*}$. Let us close this subsection by summing up the three scalar products:
\begin{itemize}
\item
The scalar product `$\underline{\la \, \ | \, \ \ra}$' for boundary Hilbert space $\cH_{\pp\Gamma}$ does not integrate bulk holonomies, so insensitive to the bulk's graph structure.

\item
The scalar product `$\left\la \phi_{\tl{\Gamma}} \middle\vert \psi_{\Gamma} \right\ra$' for spin network Hilbert space and scalar product `$\left\la \phi_{\pp\tl{\Gamma}} \middle\vert \psi_{\pp\Gamma} \right\ra$' for dual boundary Hilbert space are always equal, i.e.,
\be \label{eq:ScalarProducts-SPN-DBH}
\left\la \phi_{\tl{\Gamma}} \middle\vert \psi_{\Gamma} \right\ra=\left\la \phi_{\pp\tl{\Gamma}} \middle\vert \psi_{\pp\Gamma} \right\ra\,.
\ee
\end{itemize}

In principle, the scalar product for $\cH_{\pp\Gamma}$ can be also defined by the scalar product for $\cH_{\Gamma}$ by considering boundary holonomies \cite{Livine:2021sbf}. The advantage with using bulk-boundary maps method is that we analyze the boundary holonomies and bulk holonomies separately.
For instance, evaluating $\left\la \phi_{\pp\Gamma} \middle\vert \psi_{\pp\Gamma} \right\ra$ requires less number of integrals to holonomies, leading to coarse-graining. Let us explain in more details.


\subsection{Coarse-graining via gauge-fixing}
This subsection shows how the gauge invariance takes into account the boundary: the gauge invariance of wave-function with respect to bulk gauge transformations leads to covariant gauge transformations on the boundary. This leads to the definition of coarse-graining via gauge-fixing.

Consider local gauge transformations at vertices, which leads to gauge transformation along bulk edges in this way: $g_{e} \mapsto h_{t(e)}g_{e}h_{s(e)}^{-1}$. The boundary state determined by wave-function transformed covariantly via dressing the boundary edges
with boundary holonomies \cite{Chen:2021vrc}:
\be \label{eq:GaugeInvariance-BoundaryMaps}
|\psi_{\pp\Gamma}
(\{h_{t(e)}g_{e}h_{s(e)}^{-1}\}_{e\in\Gamma^{o}} )\ra
=
\left(\bigotimes_{e\in\pp\Gamma} h_{v(e)}^{\eps_{e}^{v}}\right)
|\psi_{\pp\Gamma}
(\{g_{e}\}_{e\in\Gamma^{o}} )\ra
\,,
\ee
where $v(e)$ for $e\in\pp\Gamma$ denotes the vertex to which the boundary edge is attached, and $\eps_{e}^{v}=-1$ is for the outward boundary edge $v(e)=s(e)$ while $\eps_{e}^{v}=1$ is for the inward boundary edge $v(e)=t(e)$.

\smallskip

The gauge-fixing is based on the covariance. Following the earlier work on spin networks \cite{Freidel:2002xb} and subsequent works \cite{Livine:2006xk,Livine:2007sy,Livine:2008iq,Chen:2021vrc,Livine:2013gna, Charles:2016xwc, Anza:2016fix}, we choose an arbitrary root vertex $u\in\Gamma$ and a maximal tree in the bulk graph $T\subset\Gamma^{o}$. A tree is a set of edges that never form any cycle (or loop). A maximal tree $T$ has $(V-1)$ edges. Furthermore, for any vertex $v\in\Gamma$, it defines a unique path of edges $P[u\rightarrow v]\subset T$ along the tree linking the root vertex $u$ to the vertex $v$. This allows to gauge-fix all the group elements along tree edges to the identity,  $g_{e\in T}\mapsto\id$, by choosing gauge transformations $h_{v}$ at every vertex but the root vertex as:
\be \label{eq:GaugeFixing}
h_{v}=\left(\overleftarrow{\prod_{\ell\in P[u\rightarrow v]}} g_{\ell}\right)^{-1}\,,
\ee
where the product of group elements is taken from right to left over $g_{\ell}$ if the edge $\ell$ is oriented in the same direction than the path $P[u\rightarrow v]$ and over its inverse $g_{\ell}^{-1}$ otherwise.
This maps all the group elements on tree edges to the identity, $h_{t(e)}g_{e}h_{s(e)}^{-1}=\id$ for $e\in T$. The remaining edges, which do not belong the tree actually correspond to a minimal generating set of loops (or cycles) on the bulk graph $\Gamma^{o}$. Indeed, each non-tree edge defines a loop from the root vertex to the edge and back,
\be
\cL=\{ e | e \in \Gamma^{o} \setminus T \} \,, \qquad
e \in \cL : u \underset{T}{\rightarrow}s(e)\underset{e}{\rightarrow}t(e)\underset{T}{\rightarrow}v_{0}
\,.\nn
\ee
The number of loops is $L=E^{o}-V+1$ where $E^{o}$ is the number of bulk edges. One can show that every cycle on the bulk graph $\Gamma^{o}$ can generate from those cycles.
For $e\in\cL$, the gauge transformation built above does not map the group element $g_{e}$ to the identity anymore but maps it to the holonomy around the corresponding loop,
\be
\forall e\in\cL\,,\qquad
h_{t(e)}g_{e}h_{s(e)}^{-1}
=
\overleftarrow{\prod_{\ell\in \cL_{e}}} g_{\ell}
\equiv
G_{e}
\,.\nn
\ee
As a consequence, we obtain the gauge-fixed boundary state in line with eq.(\ref{eq:GaugeInvariance-BoundaryMaps})
\be \label{eq:GaugeFixing-BoundaryStates}
| \psi_{\pp\Gamma} (\{g_{e}\}_{e\in\Gamma^{o}}) \ra
=
\left(\bigotimes_{e\in\pp\Gamma} h_{v(e)}^{\eps_{e}^{v}}\right)
| \psi_{\pp\Gamma}
(\{G_{e}\}_{e\in\cL},\{\id\}_{e\in T}) \ra
\,.
\ee
Actually, we can glue boundary edges along the maximal tree $T$ such that the spin network transforms to a loopy spin network $\Upsilon$ to which $L$ of loops are attached, i.e.,
\be
\Upsilon=\pp\Gamma \sqcup \cL \sqcup u\,, \qquad \pp\Upsilon=\pp\Gamma\,, \qquad \Upsilon^{o}=\cL \sqcup u
\,.
\ee
This is one of definitions for coarse-graining \cite{Livine:2006xk}. We can think of the $| \psi_{\pp\Gamma} (\{G_{e}\}_{e\in\cL},\{\id\}_{e\in T}) \ra$ in eq.(\ref{eq:GaugeFixing-BoundaryStates}) providing an identical bulk-boundary map as a loopy spin network based on $\Upsilon$.

Moreover, the gauge-fixing allows to implement channel transformation as the examples in Fig.\ref{fig:ChannelTransform-Loop}.
\begin{figure}
\begin{subfigure}[t]{1.0\linewidth}
\centering
	\begin{tikzpicture} [scale=0.7]
\coordinate  (O1) at (0,0);
\coordinate  (O2) at (1.5,0);

\draw[thick] (O1)  -- node[midway,sloped] {$>$} node[midway,above=2] {$j_{5}[g_5]$} (O2);

\draw[thick] (O1)  -- node[midway,sloped] {$>$} ++ (-0.75,0.75) node[left] {$j_1[g_1]$};
\draw[thick] (O1)  -- node[midway,sloped] {$>$} ++ (-0.75,-0.75) node[left] {$j_2[g_2]$};

\draw[thick] (O2)  -- node[midway,sloped] {$>$} ++ (0.75,0.75) node[right] {$j_4[g_4]$};
\draw[thick] (O2)  -- node[midway,sloped] {$>$} ++ (0.75,-0.75) node[right] {$j_3[g_3]$};

\draw (O2) ++ (4,0) node {$= \ \displaystyle{ \sum_{j_{6} } } \, \{6j\} $};

\coordinate  (D3) at (8.5,0.5);
\coordinate  (D4) at (8.5,-0.5);

\draw[thick] (D3)  -- node[midway,sloped] {$<$} node[midway,right=2] {$j_{6}[\id]$} (D4);

\draw[thick] (D3)  -- node[midway,sloped] {$>$} ++ (-0.75,0.75) ++ (135:0.65) node {$j_1[g_1]$};
\draw[thick] (D3)  -- node[midway,sloped] {$>$} ++ (0.75,0.75) ++(45:0.65) node {$j_4[g_4 g_5]$};

\draw[thick] (D4)  -- node[midway,sloped] {$>$} ++ (-0.75,-0.75) ++ (225:0.65) node {$j_2[g_2]$};
\draw[thick] (D4)  -- node[midway,sloped] {$>$} ++ (0.75,-0.65) ++ (-45:0.75) node {$j_3[g_3 g_5]$};

\end{tikzpicture}
\caption{The $6j$-symbol represents the unitary for the channel transformation.}
\label{fig:ChannelTransform-Tree}
\end{subfigure}
\begin{subfigure}[t]{1.0\linewidth}
\centering
\begin{tikzpicture} [scale=0.7]
\coordinate  (O1) at (-1,0);
\coordinate  (O2) at (1,0);

\draw[thick] (O1)  to [bend left=45] node[midway,sloped] {$>$} node[midway,above=2] {$j_{5}[g_5]$} (O2);
\draw[thick] (O1)  to [bend right=45] node[midway,sloped] {$>$} node[midway,below=2] {$j_{6}[g_6]$} (O2);

\draw[thick] (O1)  -- node[midway,sloped] {$>$} ++ (-0.75,0.75) node[left] {$j_1[g_1]$};
\draw[thick] (O1)  -- node[midway,sloped] {$>$} ++ (-0.75,-0.75) node[left] {$j_2[g_2]$};

\draw[thick] (O2)  -- node[midway,sloped] {$>$} ++ (0.75,0.75) node[right] {$j_4[g_4]$};
\draw[thick] (O2)  -- node[midway,sloped] {$>$} ++ (0.75,-0.75) node[right] {$j_3[g_3]$};

\draw (O2) ++ (3.5,0) node {$= \ \displaystyle{ \sum } \, \{6j\}$};

\coordinate  (A1) at (8,0);
\coordinate  (A2) at (10,0);
\coordinate  (A3) at (9,0);

\draw[thick] (A1)  -- node[midway,sloped] {$>$} node[midway,above=2] {$\id$} (A3) -- node[midway,sloped] {$>$} node[midway,above=2] {$\id$} (A2);

\draw[thick] (A1)  -- node[midway,sloped] {$>$} ++ (-0.75,0.75) node[left] {$j_1[g_1]$};
\draw[thick] (A1)  -- node[midway,sloped] {$>$} ++ (-0.75,-0.75) node[left] {$j_2[g_2]$};

\draw[thick] (A2)  -- node[midway,sloped] {$>$} ++ (0.75,0.75) node[right] {$j_4[g_4 g_5]$};
\draw[thick] (A2)  -- node[midway,sloped] {$>$} ++ (0.75,-0.75) node[right] {$j_3[g_3 g_5]$};

\draw[thick] (A3) -- node[midway,left] {$\id$} ++ (0,-0.75) coordinate (A4) ;

\draw[thick,in=+30,out=-30,scale=3,rotate=-90] (A4)  to [loop] node[very near start,sloped,rotate=-90] {$>$} node[midway,below=2,rotate=0] {$j_6[ g_5^{-1} g_6]$} (A4);

\end{tikzpicture}
\caption{The $6j$-symbol represents the unitary for the channel transformation. The right hand side is equivalent to a loopy spin network in the sense of bulk-boundary map.}
\label{fig:ChannelTransform-Loop}
\end{subfigure}
\caption{Channel transformations for spin networks with tree and loop. 
}
\label{fig:ChannelTransform}
\end{figure}

\smallskip

With the gauge-fixing technique, we can prove below proposition:
\begin{prop} \label{Prop:InnerProduct-BoundaryHilbertSpace}
Let $\psi_{\Gamma}$ and $\phi_{\Gamma}$ be any two spin networks on same graph $\Gamma$. The corresponding boundary states are $| \psi_{\pp\Gamma}(\{g_{e}\}_{e\in\Gamma^{o}}) \ra$ and $| \phi_{\pp\Gamma}(\{g_{e}\}_{e\in\Gamma^{o}}) \ra$, respectively.
Then the scalar product for boundary Hilbert space is a function depending only on $L$ of group elements, i.e.,
\be
\underline{
\la \psi_{\pp\Gamma}
(\{g_{e}\}_{e\in\Gamma^{o}})
\vert
\phi_{\pp\Gamma}
(\{g_{e}\}_{e\in\Gamma^{o}}) \ra
}
=
F_{\psi\phi}(G_{1},\cdots,G_{L})
\,.
\ee
Moreover, the function is invariant under conjugation:
\be \label{eq:ConjugateInvariance-BulkHolonomies}
F_{\psi\phi}(G_{1},..,G_{L})=
F_{\psi\phi}(h \, G_1 \, h^{-1},\cdots,h \, G_L \, h^{-1})
\,, \quad \forall \, h \in \SU(2)
\,.
\ee
\end{prop}
\begin{proof}
Since spin networks $\psi_{\Gamma} $ and $\phi_{\Gamma} $ based on same graph $\Gamma$, we repeat the gauge-fixing implementation for $\psi_{\Gamma}$ to $\phi_{\Gamma}$.
The gauge-fixing leads to identical boundary holonomies with respect to eq.(\ref{eq:GaugeFixing-BoundaryStates}).
As a consequence, the scalar product for $\cH_{\pp\Gamma}$ depends only on these $L$ of group elements,
\be \label{eq:InnerProduct-BoundaryHilbertSpace-GaugeFixed}
\begin{aligned}
\underline{
\la \psi_{\pp\Gamma} (\{g_{e}\}_{e\in\Gamma^{o}}) | \phi_{\pp\Gamma} (\{g_{e}\}_{e\in\Gamma^{o}}) \ra
}
=
&
\underline{
\la \psi_{\pp\Gamma} (\{G_{e}\}_{e\in\cL},\{\id\}_{e\in T}) \vert \phi_{\pp\Gamma} (\{G_{e}\}_{e\in\cL},\{\id\}_{e\in T}) \ra
}
\\
\equiv
&
F_{\psi\phi}(G_{1},\cdots,G_{L})
\,,
\end{aligned}
\ee
since for every boundary edges, the boundary holonomies are undone by the scalar product for $\cH_{\pp\Gamma}$: $\underline{\la j_{e} m_{e} | h^{\dagger} h | j'_{e} m'_{e} \ra}=\underline{\la j_{e} m_{e} | j'_{e} m'_{e} \ra}=\delta_{j_{e} j'_{e} } \delta_{m_{e} m'_{e} }$.
\end{proof}
For spin networks $\psi_{\Gamma} $, $\phi_{\tl{\Gamma} } $ based on different graphs, the boundary holonomies cannot be undone by scalar product $\underline{\la \, \ | \, \ \ra}$. But thanks for graph orthogonality, the scalar product for $(\cH_{\cB})^{*}$ simply vanishes after integrating holonomies.


\subsection{Unitary for dual boundary Hilbert space} \label{subsection:Unitary-DualHilbertSpace}
This subsection is to read the boundary holonomies and channel transformations as unitary maps for $(\cH_{\cB})^{*}$. Thus the coarse-graining is viewed as a consequence of the unitarity of $(\cH_{\cB})^{*}$. It will be useful when comes to entanglement topic.

\smallskip

Following the analysis below eq.(\ref{eq:GaugeFixing-BoundaryStates}), there exists a loopy spin network state mapping identically as the gauge-fixed spin network state, i.e.
\be \label{eq:Correspondence-GaugeFixed,LoopySPN}
| \psi_{\pp\Gamma} (\{G_{e}\}_{e\in\cL},\{\id\}_{e\in T}) \ra
=
| \psi_{\pp\Upsilon} (\{G_{e}\}_{e\in\cL}) \ra
\,.
\ee
So eq.(\ref{eq:InnerProduct-BoundaryHilbertSpace-GaugeFixed}) can be written as
\be
\underline{
\la \psi_{\pp\Gamma}
(\{g_{e}\}_{e\in\Gamma^{o}})
\vert
\phi_{\pp\Gamma}
(\{g_{e}\}_{e\in\Gamma^{o}}) \ra
}
=
\underline{
\la \psi_{\pp\Upsilon}
(\{G_{e}\}_{e\in\cL})
\vert
\phi_{\pp\Upsilon}
(\{G_{e}\}_{e\in\cL}) \ra
}
\,,
\ee
and the scalar product for $(\cH_{\pp\Gamma})^{*} \subset (\cH_{\cB})^{*}$,
\be \label{eq:DualBoundary-SNWLoopy}
\begin{aligned}
\la \psi_{\pp\Gamma} \vert \phi_{\pp\Gamma} \ra
=&
\int
\prod_{e\in\Gamma^{o}}\rd g_{e}\,
\underline{
\la \psi_{\pp\Gamma}
(\{g_{e}\}_{e\in\Gamma^{o}})
\vert
\phi_{\pp\Gamma}
(\{g_{e}\}_{e\in\Gamma^{o}}) \ra
}
\\
=&
\int
\prod_{e\in\Upsilon^{o}}\rd G_{e}\,
\underline{
\la \psi_{\pp\Upsilon}
(\{G_{e}\}_{e\in\cL})
\,\vert\,
\phi_{\pp\Upsilon}
(\{G_{e}\}_{e\in\cL}) \ra
}
=
\la \psi_{\pp\Upsilon} \vert \phi_{\pp\Upsilon} \ra
\,.
\end{aligned}
\ee
From the second equality to the third, only $L$ number of integrals will be taken account, since the rest of integrals are trivial $\int_{\SU(2)} \rd g_e=1$.

\smallskip

Notably, the correspondence between $| \psi_{\pp\Gamma} (\{g_{e}\}_{e\in\Gamma^{o}}) \ra$ and $| \psi_{\pp\Upsilon} (\{G_{e}\}_{e\in\cL}) \ra$ is one-to-many, since there are many doable implementations for gauge-fixing. But the wanted relation $\la \psi_{\pp\Gamma} \vert \phi_{\pp\Gamma} \ra=\la \psi_{\pp\Upsilon} \vert \phi_{\pp\Upsilon} \ra$ is not affected by how to gauge-fix.

\smallskip

Eq.(\ref{eq:DualBoundary-SNWLoopy}) sheds light on the unitary maps for $(\cH_{\cB})^{*}$. Given a spin network state, we implement gauge-fix to obtain a loopy spin network state such that the two states provides identical bulk-boundary maps up to boundary holonomies, i.e., from $| \psi_{\pp\Gamma} \ra$ to $| \psi_{\pp\Upsilon} \ra$,
\be \label{eq:Unitary-GaugeFixed}
| \psi_{\pp\Gamma} (\{g_{e}\}_{e\in\Gamma^{o}}) \ra
=
\left(\bigotimes_{e\in\pp\Gamma} h_{v(e)}^{\eps_{e}^{v}}\right)
| \psi_{\pp\Upsilon}
(\{G_{e}\}_{e\in\Upsilon^{o}}) \ra
\equiv
\cG_{\cB} \act
| \psi_{\pp\Upsilon}
(\{G_{e}\}_{e\in\Upsilon^{o}}) \ra
\,,
\ee
where $\cG_{\cB}$ denotes the operation of boundary holonomies.
The reverse problem, i.e. how to recover $| \psi_{\pp\Gamma} \ra$ from $| \psi_{\pp\Upsilon} \ra$ is worth to study, which may appear elsewhere. At present stage, what one should bear in mind is that the reverse procedure is doable and is related to manipulating boundary holonomies.

\smallskip

In the present work, we also consider another type of unitary. The $\cH_{\Upsilon}$ is the intertwiner space for loopy vertex, and is defined according to (\ref{eq:IntertwinerHilbertSpace}),
\be \label{eq:HilbertSpace-Loopy}
\cH_{\Upsilon}
:=
\textrm{Inv}_{\SU(2)}\left[
\bigotimes_{e\in\pp \Upsilon |\,u=s(e)} \cV_{j_{e}}
\otimes
\bigotimes_{e\in\pp\Upsilon |\,u=t(e)} \cV_{j_{e}}^{*}
\otimes \bigotimes_{e'\in\cL} \big( \cV_{j_{e'}} \otimes \cV_{j_{e'} }^{*} \big)
\right]
\,.
\ee
Then we denote $\cU_{\Upsilon}$ the channel transformation for the loopy intertwiner. $\cU_{\Upsilon}$ is just for re-definition of internal labels for the intertwiner (e.g. Fig.\ref{fig:ChannelTransform}).

\smallskip

For the purpose of studying coarse-graining, we list the two classes of unitaries for the dual boundary Hilbert space $(\cH_{\cB})^{*}$:
\begin{enumerate}[label=(\alph*)]
\item
Boundary holonomies for boundary unitary $\cG_{\cB} : \cH_{\cB} \longrightarrow \cH_{\cB}$.

\label{Unitary:BoundaryHolonomies}

\item
Channel transformation for loopy intertwiner unitary $\cU_{\Upsilon} : \cH_{\Upsilon} \longrightarrow \cH_{\Upsilon}$.

\label{Unitary:LoopyIntertwiner}
\end{enumerate}
Both particular $\cG_{\cB}$ and $\cU_{\Upsilon}$ never change the boundary representations, i.e. change the boundary spins.
Moreover, they are mutually commutative, i.e., $\cG_{\cB} \circ \cU_{\Upsilon}=\cU_{\Upsilon} \circ \cG_{\cB}$.

\smallskip

Let us close the subsection by summing up above discussions:
\begin{prop} \label{Prop:Unitary:Graph-LoopyVertex}
Given boundary $\cB$, a unitary map for the dual boundary Hilbert space $(\cH_{\cB})^{*}$ preserves the scalar product
\be \label{eq:Unitary-ScalarProduct-BoundaryMaps}
\left\la \phi_{\pp\tl{\Gamma}} \middle\vert \psi_{\pp\Gamma} \right\ra
=
\left\{
\begin{array}{rlcl}
&
\left\la \phi_{\pp\Upsilon} \middle\vert \psi_{\pp\Upsilon} \right\ra
\,, \qquad
&
\Gamma=\tl{\Gamma}\,,
\\
&
0
\,, \qquad
&
\Gamma \neq \tl{\Gamma}\,.
\end{array}
\right.
\ee
We can compose boundary holonomies $\cG_{\cB} : \cH_{\cB} \longrightarrow \cH_{\cB}$ and channel transformation $\cU_{\Upsilon}: \cH_{\Upsilon} \longrightarrow \cH_{\Upsilon}$ to define the unitary map $\cU_{\cB} : (\cH_{\cB})^{*} \longrightarrow (\cH_{\cB} )^{*}$ via $\cU_{\cB}=\cG_{\cB} \circ \cU_{\Upsilon}$.

In particular, suppose $| \psi_{\pp\Gamma} (\{g_{e}\}_{e\in\Gamma^{o}}) \ra \in (\cH_{\pp\Gamma} )^{*}$ and $| \psi_{\pp\Upsilon} ( \{ G_{e} \}_{e\in \Upsilon^{o} } ) \ra \in (\cH_{\pp\Upsilon} )^{*}$. Let $\big\{ | \Psi_{ \Gamma, \{ I_{v} \} } \ra \big\}$ be orthonormal basis for $\cH_{\Gamma}$ and $\big\{ | \Psi_{\Upsilon, I_{u} } \ra \big\}$ be orthonormal basis for $\cH_{\Upsilon}$. Then according to eq.(\ref{eq:Correspondence-GaugeFixed,LoopySPN}), they are related by:
\begin{align}
| \psi_{\pp\Gamma} (\{g_{e}\}_{e\in\Gamma^{o}}) \ra
=&
\Big[ \cU_{\cB} \Big]^{\dagger}
\act
| \psi_{\pp\Upsilon } ( \{ G_{e} \}_{e\in \Upsilon^{o} } ) \ra
\,, \label{eq:Unitary-SchrodingerPicture}
\\
| \Psi_{\pp\Gamma, \{ I_{v} \} }(\{g_{e}\}_{e\in\Gamma^{o}}) \ra
=&
\sum_{ I_{u } }
\tensor{   {\Big[ \cU_{\cB} \Big] }   } {^{ I_{ u }  }  _{  \{ I_{v} \}_{v\in\Gamma}    }    }
\act
| \Psi_{\pp\Upsilon, I_{u} } ( \{ G_{e} \}_{e\in \Upsilon^{o} } ) \ra
\,. \label{eq:Unitary-HeisenbergPicture}
\end{align}
The first line expresses the Schr\"odinger picture for the unitary, and the second line expresses the Heisenberg picture for the unitary.
\end{prop}
We have clarified the unitary for dual boundary Hilbert space. The coarse-graining in this paper is now understood as a particular unitary map for the dual boundary Hilbert space.


\section{Spin network entanglement} \label{Section:SPNWsEntanglement}
This section aims to define the spin network entanglement and to present how to study the entanglement under the coarse-graining. The goal is to show that the coarse-graining exactly preserves the spin network entanglement at kinematical level. 

\subsection{Reduced density matrices on spin network} \label{Subsection:ReducedDensityMatrices}
This subsection is meant to define the spin network entanglement. A graph $\Gamma$ is partitioned by starting from partitioning the set of vertices. The set of vertices $\mathscr{V}=\{ v_1\,, v_2\,, \cdots \}$ is partitioned into subsets
\be
\mathscr{V}=\bigsqcup_{i=1}^{n} \mathscr{V}_{i}
\ee
such that the vertices of every subset $\mathscr{V}_{i}$ can be connected by a path. 
Then every $\mathscr{V}_{i}$ defines a subgraph $\Gamma_{i}$ by this way: (i) For the bulk of $\Gamma_{i}$, i.e. $\Gamma_{i}^{o}$, the set of vertices is $\mathscr{V}_{i}$, and the set of bulk edges consists of the edges in $\Gamma$ whose two-end vertices are both in $\mathscr{V}_{i}$. (ii) For those edges whose two-end vertices are in different subsets, they are split into two piecewise. For instance, suppose an edge $e$ whose $s(e) \in \mathscr{V}_{i}, \ t(e) \in \mathscr{V}_{j}$ and $i \neq j$, then $e$ is split into two piecewise $e_{i}, \ e_{j}$ such that $e=e_{i} \sqcup e_{j}$. (iii) All one-end edges whose source vertex or target vertex belongs to $\mathscr{V}_{i}$, define the boundary of $\Gamma_{i}$, i.e. $\pp\Gamma_{i}$.
Therefore, the graph $\Gamma$ is partitioned by
\be
\Gamma=\bigsqcup_{i=1}^{n} \Gamma_{i}\,, \qquad
\Gamma_{i}=\Gamma_{i}^{o} \sqcup \pp\Gamma_{i}\,, \qquad
\pp\Gamma_{i} \equiv \cB_{i}
\,.
\ee
Based on the partition, the spin network Hilbert space satisfies the set-relation
\be \label{eq:Hilbert-SPN-SPNSubN}
\cH_{\Gamma}
\subset \bigotimes_{i=1}^{n} \cH_{\Gamma_{i}}
\,.
\ee
Every $\cH_{\Gamma_{i}}$ is the spin sub-network Hilbert space based on the corresponding $\Gamma_{i}$.
Here the $\subset$ sign is due to the spin-matching constraint imposed amongst $\cH_{\pp\Gamma_{i}}$. 

Above definition can be viewed as a generalization for the particular situation that vertices are partitioned into sole vertex, i.e. every sole vertex and its attached edges make up a sub-graph with sole vertex,
\beq
\cH_{ \Gamma }
=
\bigoplus_{\{j_{e}\}_{e\in\Gamma}}
\cH_{v}^{\{j_{e}\}_{e\ni v}}
\subset \bigotimes_{v\in\Gamma } \, \cH_{v}
\,,
\eeq
where the vertex Hilbert spaces are defined as
\be
\cH_{v}^{\{j_{e}\}_{e\ni v}}
=
\textrm{Inv}_{\SU(2)}\Big{[}
\bigotimes_{e|\,v=s(e)} \cV_{j_{e}}
\otimes
\bigotimes_{e|\,v=t(e)} \cV_{j_{e}}^{*}
\Big{]}
\quad\textrm{and}\quad
\cH_{v}
=
\bigoplus_{ \{j_{e} \}_{e\ni v} }
\cH_{v}^{\{j_{e}\}_{e\ni v}}
\,.
\ee
Here the $\subset$ sign is due to the spin-matching constraint imposed amongst every bulk edge.

Spin-matching constraint introduces entanglement between the vertices or sub-networks to which it connects. The entanglement is introduced by spin-superposition. As done in \cite{Livine:2017fgq,Chen:2022rty}, a strategy is to consider spin networks as states in the larger Hilbert space $\bigotimes_{v\in\Gamma } \, \cH_{v}$ of tensor products of intertwiners without imposing the spin matching constraints along the bulk edges $e\in \Gamma^{o}$. The advantage with this starting point is that we are directly looking at correlations and entanglement between $\SU(2)$-gauge invariant excitations -the intertwiners- and that we do not have to worry about gauge breaking and correlations between non-gauge invariant observables (see e.g. \cite{Donnelly:2008vx,Donnelly:2016auv,Livine:2017fgq} for a discussion on this issue).

The present work will follow the strategy, and generalize the partition from $\bigotimes_{v\in\Gamma } \, \cH_{v}$ to $\bigotimes_{i=1}^{n} \, \cH_{\Gamma_{i} }$, i.e. from vertices to sub-networks.

\smallskip

To define the entanglement between sub-networks, we start from entanglement between vertices. A generic spin network state can be decomposed as a superposition over spin network basis states:
\beq  \label{eq:SPN-decomposition}
| \psi_{\Gamma} \ra
=
\sum_{\{I_v\}}
C_{ \Gamma } ( \{ I_{v} \} )
 \,
\bigotimes_{v\in\Gamma}
| \Psi_{ v, I_v } \ra
\,, \qquad \text{where} \quad
| \Psi_{ \Gamma, \{ I_v \} } \ra
=
\bigotimes_{v\in\Gamma}
| \Psi_{ v, I_v } \ra
\,.
\eeq
Here the intertwiner basis state $| \Psi_{ v, I_v } \ra \in \cH_{v}$ have definite spins and intertwiner, with spins and internal intertwiner indices packaged in the labels $I_{v}$. Then the coefficients $C_{ \Gamma } ( \{ I_{v} \} )$ for a general state allows for superpositions of both spins and intertwiners, thus leading to correlation between intertwiner states located at different vertices.

Since spin network basis states can be factorized as the tensor product of intertwiner basis state $| \Psi_{ v, I_v } \ra \in \cH_{v}$, we can group up the intertwiner basis states within every sub-network, which defines a factorization for spin network basis state $| \Psi_{ \Gamma, \{ I_v \} } \ra \in \cH_{\Gamma}$ with respect to spin sub-networks basis states $| \Psi_{ \Gamma_{i}, \{ I_{v} \}_{v\in\Gamma_{i}} } \ra \in \cH_{\Gamma_{i} }$,
\beq  \label{eq:SPSubN-decomposition-BasisState}
| \Psi_{ \Gamma, \{ I_v \} } \ra
&=&
\bigotimes_{i=1}^{n}
| \Psi_{ \Gamma_{i}, \{ I_v \}_{v\in\Gamma_{i}} } \ra
\,, \qquad \text{where} \quad
| \Psi_{ \Gamma_{i}, \{ I_v \}_{v\in\Gamma_{i}} } \ra
=
\bigotimes_{v\in\Gamma_{i} }
| \Psi_{ v, I_v } \ra
\,.
\eeq
It allows to re-group the coefficients $C_{ \Gamma } ( \{ I_{v} \} )=C_{ \Gamma } ( [ \{ I_{v} \}_{v\in\Gamma_{i}} ]_{i} )$ where every square bracket $[,]_{i}$ is adopted to cluster the vertices belonging to respect $\Gamma_{i}$.
In this way, the eq.(\ref{eq:SPN-decomposition}) can be also written as
\beq  \label{eq:SPSubN-decomposition}
| \psi_{\Gamma} \ra
=
\sum_{\{I_v\}_{v\in\Gamma_{i}} }
C_{ \Gamma } ( [\{ I_v \}_{v\in\Gamma_{i}}]_{i} )
 \,
\bigotimes_{i=1}^{n}
| \Psi_{ \Gamma_{i}, \{ I_v \}_{v\in\Gamma_{i}} } \ra
\,.
\eeq
The entanglement of $| \psi_{\Gamma} \ra$ between sub-graphs is encoded into the unfactorizability with respect to $\bigotimes_{i=1}^{n} \cH_{\Gamma_{i}}$, which can be studied via the formalism of density matrix.

\smallskip

Given any pure spin network state $| \psi_{\Gamma} \ra$, it corresponds to pure density matrix $\rho_{\Gamma}[\psi]=| \psi_{\Gamma} \ra\la \psi_{\Gamma} |$. It is straightforward to generalize the following procedure to the cases of mixed density matrix $\rho_{\Gamma}$, since any mixed density matrix allows decomposition $\rho=\sum_{k} W_{k} | \psi_{k} \ra\la \psi_{k} |$.
The reduced density matrix for sub-network $\Gamma_{i}$ is defined via partial trace over the complementary:
\be \label{eq:SPN-PartialTrace}
\rho_{\Gamma_{i}} [ \psi ]
=
\tr_{ \cH_{ \Gamma \setminus \Gamma_{i} } }\rho_{\Gamma} [ \psi ]
\quad \in \textrm{End} \Big[ \cH_{\Gamma_{i}} \Big]
\,, \qquad
\Gamma \setminus \Gamma_{i}=\bigsqcup_{j\neq i}^{n} \, \Gamma_{j}
\,.
\ee
The partial trace is implemented by choosing an orthonormal basis for $\cH_{\Gamma \setminus \Gamma_{i}}$,
\beq
| \Psi_{ \Gamma \setminus \Gamma_{i} }, \{ I_v \}_{v \notin \Gamma_{i} } \ra
=
\bigotimes_{j \neq i}^{n}
| \Psi_{ \Gamma_{j}, \{ I_v \}_{v\in\Gamma_{j}} } \ra
\,, \qquad \text{where} \quad
| \Psi_{ \Gamma_{j}, \{ I_v \}_{v\in\Gamma_{j}} } \ra
=
\bigotimes_{v\in\Gamma_{j} }
| \Psi_{ v, I_v } \ra
\,,
\eeq
therefore, eq.(\ref{eq:SPN-PartialTrace}) can be expressed by
\be \label{eq:SPN-PartialTrace-onBasis}
\rho_{\Gamma_{i}} [ \psi ]
=
\sum_{ \{ I_v \}_{v \notin \Gamma_{i} } }
\la \Psi_{ \Gamma \setminus \Gamma_{i} }, \{ I_v \}_{v \notin \Gamma_{i} } |
\, \rho_{\Gamma} [ \psi ] \,
| \Psi_{ \Gamma \setminus \Gamma_{i} }, \{ I_v \}_{v \notin \Gamma_{i} } \ra
\quad \in \textrm{End} \Big[ \cH_{\Gamma_{i}} \Big]
\,.
\ee
According to eq.(\ref{eq:SPSubN-decomposition}), the scalar product $\la \Psi_{ \Gamma \setminus \Gamma_{i} }, \{ I_v \}_{v \notin \Gamma_{i} } | \psi_{\Gamma} \ra$ is presented by
\begin{align}
&
| \psi_{\Gamma} \ra
=
\sum_{ \{ I_{v} \}_{v \in \Gamma_{i} } }
\sum_{ \{ I_{v'} \}_{v' \notin \Gamma_{i} } }
C_{ \Gamma } ( \{ I_v \}_{v\in\Gamma_{i}}, \{ I_{v'} \}_{v'\notin\Gamma_{i}} ) | \Psi_{ \Gamma_{i}, \{ I_v \}_{v\in\Gamma_{i}} } \ra
\otimes
| \Psi_{ \Gamma \setminus \Gamma_{i} }, \{ I_{v'} \}_{v' \notin \Gamma_{i} } \ra
\,,
\\
&
\la \Psi_{ \Gamma \setminus \Gamma_{i} }, \{ I_{v'} \}_{v' \notin \Gamma_{i} } | \psi_{\Gamma} \ra
=
\sum_{ \{ I_{v} \}_{v \in \Gamma_{i} } }
C_{ \Gamma } ( \{ I_v \}_{v\in\Gamma_{i}}, \{ I_{v'} \}_{v'\notin\Gamma_{i}} ) | \Psi_{ \Gamma_{i}, \{ I_v \}_{v\in\Gamma_{i}} } \ra
\,, \label{eq:SPN-PartialTrace-ScalarProduct}
\end{align}
thus the reduced density $\rho_{\Gamma_{i}} [ \psi ]$ in eq.(\ref{eq:SPN-PartialTrace-onBasis}) is expressed by
\begin{align} \label{eq:SPSubN-PartialTrace-onBasis}
\rho_{\Gamma_{i}} [ \psi ]
=
\sum_{ \{ I_{v'} \}_{v' \notin \Gamma_{i} } }
\sum_{ \{ I_{v} \}_{v \in \Gamma_{i} } }
\sum_{ \{ \tl{I}_{v} \}_{v \in \Gamma_{i} } }
&
C_{ \Gamma } ( \{ I_v \}_{v\in\Gamma_{i}}, \{ I_{v'} \}_{v'\notin\Gamma_{i}} )
\overline{ C_{ \Gamma } ( \{ \tl{I}_v \}_{v\in\Gamma_{i}}, \{ I_{v'} \}_{v'\notin\Gamma_{i}} ) }
\\
&
| \Psi_{ \Gamma_{i}, \{ I_v \}_{v\in\Gamma_{i}} } \ra\la \Psi_{ \Gamma_{i}, \{ \tl{I}_v \}_{v\in\Gamma_{i}} } |
\,. \nn
\end{align}
The reduced density matrix $\rho_{\Gamma_{i}} [ \psi ]$ encodes the information about intertwiners located in $\Gamma_{i}$, so the spin network entanglement is related to the intertwiner entanglement \cite{Livine:2017fgq}. In other words, the entanglement between spin sub-networks amounts to being entanglement between these `cluster-intertwiner'.

\smallskip

So far we have discussed the entanglement structure of spin networks, and the definition of reduced density matrices for spin sub-networks. The next subsection is to apply coarse-graining for the spin network entanglement.


\subsection{Entanglement coarse-graining} \label{Section:EntanglementCoarseGraining}
In this part, we investigate the coarse-graining for spin network entanglement. Suppose that $\Gamma$ is partitioned into $\Gamma=\bigsqcup_{i=1}^{n}\Gamma_{i}$. We will show that the entanglement between these sub-graphs $\{ \Gamma_{i} \}_{i=1}^{n}$, can be reflected in coarse-grained graph $\Gamma_{(R)}$ made up by loopy graphs $\{ \Upsilon_{i} \}_{i=1}^{n}$.


Let us apply the viewpoint of dual boundary Hilbert space for this goal. Any partition, needless to say, introduces boundaries. The dual boundary Hilbert spaces inherit the entanglement structure eq.(\ref{eq:Hilbert-SPN-SPNSubN})
\be \label{eq:DualBoundaryHilbert-SPN-SPNSubN}
( \cH_{\pp\Gamma} )^{*}
\subset
\bigotimes_{i=1}^{n} \,
(\cH_{\pp\Gamma_{i} } )^{*}
\,.
\ee
Here every $(\cH_{\pp\Gamma_{i} } )^{*}$ is the dual boundary Hilbert space associative with spin sub-network $\Gamma_{i}$. As the correspondence between $| \psi_{\pp\Gamma} \ra \in (\cH_{\pp\Gamma })^{*}$ and $| \psi_{\Gamma} \ra \in \cH_{\Gamma}$, the density matrix $\rho_{\pp\Gamma^{*}} \in \mathrm{End} \Big[ (\cH_{\pp\Gamma })^{*} \Big]$ corresponds to $\rho_{\Gamma} \in \mathrm{End} \Big[ \cH_{\Gamma} \Big]$ in same way. Due to eq.(\ref{eq:ScalarProducts-SPN-DBH}), the information they contain is identical. The next step is to check that the scalar product for $\cH_{\Gamma_{i}}$ and the scalar product for $( \cH_{\pp\Gamma_{i}} )^{*}$ are equivalent with respect to partial trace. This is again expected due to eq.(\ref{eq:ScalarProducts-SPN-DBH}), but notice the subtlety
\be \label{eq:ScalarProducts-SPN-DBH-PartialTrace}
\cH_{\Gamma_{i} } \ni
\la \Psi_{ \Gamma \setminus \Gamma_{i}, \{ I_v \}_{v \notin \Gamma_{i} } } | \psi_{\Gamma} \ra
\overset{!}{=}
\cG_{ \cB_{i} } \act
\la \Psi_{ \pp(\Gamma \setminus \Gamma_{i}), \{ I_v \}_{v \notin \Gamma_{i} } } | \psi_{\pp\Gamma} \ra
\in (\cH_{\pp\Gamma_{i} } )^{*}
\,.
\ee
The left hand side is computed from the spin network Hilbert space that does not care the boundary holonomies, and right hand side is computed from the dual boundary Hilbert space keeping the $\cG_{\cB_{i}}$. Here sign $\overset{!}{=}$ is adopted to indicate the equivalence but also to notice the slight subtlety.

Of course, the subtlety can be also understood from spin network wave-function: the holonomy is split in terms of group multiplication when the edge is cut, then partial trace removes one piecewise holonomy with integration. So $\cH_{\Gamma}$ and $(\cH_{\pp\Gamma})^{*}$ perspectives are entirely equivalent.

\smallskip

The boundary holonomies do not affect the scalar product for $(\cH_{\pp\Gamma_{i} } )^{*}$, thus they do not change the Schmidt eigenvalues of the reduced density matrix. Indeed, from the viewpoint of subsection \ref{subsection:Unitary-DualHilbertSpace}, the $\cG_{\cB_{i}}$ should be understood as a unitary for $(\cH_{\pp\Gamma_{i} } )^{*}$ for whole sub-network $\Gamma_{i}$, i.e. $\cG_{\cB_{i}}$ is a local unitary with respect to $\Gamma_{i}$. Therefore, $\cG_{\cB_{i}}$ can not affect the entanglement between spin sub-networks.

\begin{prop} \label{Prop:DensityMatrix-SPN-DBH}
Let $\rho_{\Gamma} \in  \mathrm{End} \Big[ \cH_{\Gamma} \Big]$ be the density matrix for spin network Hilbert space $\cH_{\Gamma}$, which allows decomposition $\rho_{\Gamma}=\sum_{k} W_{k} | \psi_{\Gamma}^{(k)} \ra\la \psi_{\Gamma}^{(k)} |$. Since every $| \psi_{\Gamma}^{(k)} \ra$ offers a bulk-boundary map $| \psi_{\pp\Gamma}^{(k)} \ra$, then in the dual boundary Hilbert space $( \cH_{\pp\Gamma ) } )^{*}$, the corresponding density matrix is $\rho_{\pp\Gamma^{*}} = \sum_{k} W_{k} | \psi_{\pp\Gamma}^{(k)} \ra\la \psi_{\pp\Gamma}^{(k)} | \in \mathrm{End} \Big[ (\cH_{\pp\Gamma})^{*} \Big]$.
The density matrices $\rho_{\Gamma}$ and $\rho_{\pp\Gamma^{*}}$ encode identical entanglement between spin sub-networks. 
\end{prop}

\medskip

In fact, any bulk-boundary map for $\pp\Gamma$ can be expressed by `gluing' $\bigutimes$ bulk-boundary maps associative with $\Gamma_{i}$,
\beq  \label{eq:DBHSub-decomposition}
| \psi_{\pp\Gamma} (\{ g_{e} \}_{e\in\Gamma^{o}})\ra
=
\sum_{\{I_v\}_{v\in\Gamma_{i}} }
C_{ \Gamma } ( [\{ I_v \}_{v\in\Gamma_{i}}]_{i} )
 \,
\bigutimes_{i=1}^{n}
\cG_{\pp\Gamma_{i} } \act
| \Psi_{ \pp\Gamma_{i}, \{ I_v \}_{v\in\Gamma_{i}} } (\{ g_{e} \}_{e\in\Gamma_{i}^{o}}) \ra
\,.
\eeq
Here notation $\bigutimes$ means gluing the sub-networks with boundary holonomies $\cG_{\pp\Gamma_{i}}$ along the interfacing edges. This is how we acquire a boundary state from a spin network state: every vertex and its edges make up a simplest open spin network, then bulk holonomies glue these vertices together, the rest of open edges make up the boundary Hilbert space. In this sense, spin network wave-function is referred to coarse-grainer.
We refer the interested reader to \cite{Livine:2021sbf} for the gluing operation for time boundaries.

\smallskip

The advantage with the viewpoint of dual boundary Hilbert is that we are able to coarse-grain sub-networks. To see this, let us revisit partial trace eq.(\ref{eq:SPN-PartialTrace-onBasis}) which requires to compute eq.(\ref{eq:SPN-PartialTrace-ScalarProduct}), but this time we compute it based on eq.(\ref{eq:DBHSub-decomposition}). Now, for any spin sub-network, the scalar product can be computed via eq.(\ref{eq:Unitary-ScalarProduct-BoundaryMaps}), i.e.,
\be
\la \phi_{\Gamma_{i} } | \psi_{\Gamma_{i} } \ra
=
\la \phi_{\pp\Gamma_{i} } | \psi_{\pp\Gamma_{i} } \ra
=
\la \phi_{\pp\Upsilon_{i} } | \psi_{\pp\Upsilon_{i} } \ra
\,,
\ee
where the left side is computed from spin sub-network Hilbert space $\cH_{\pp\Gamma_{i}}$, middle from dual boundary Hilbert space $(\cH_{\pp\Gamma_{i}})^{*}$, and right side from loopy dual boundary Hilbert space $(\cH_{\pp\Upsilon_{i}})^{*}$. 

\smallskip

Following the analysis, we coarse-grain spin sub-networks into loopy spin networks via gauge-fixing. Given graph $\Gamma$ and partition $\Gamma=\bigsqcup_{i=1}^{n} \Gamma_{i}$, the coarse-grained graph $\Gamma_{(R)}$ is obtained via gauge-fixing every $\Gamma_{i}$ to make every $\Gamma_{i}$ a loopy spin network $\Upsilon_{i}$, 
then glue $\Upsilon_{i}$ back and acquire the coarse-grained graph (e.g. Fig.\ref{fig:ReducingGraph})
\be
\Gamma_{(R)}
=
\bigsqcup_{i=1}^{n} \Upsilon_{i}
\,.
\ee

\smallskip

The coarse-grained graph preserves the entanglement between spin sub-networks. Indeed, one can understand the entanglement preservation from the perspective of Proposition \ref{Prop:Unitary:Graph-LoopyVertex}: gauge-fixing $\Gamma_{i}$ amounts to implementing a local unitary transformation with respect to $(\cH_{\cB_{i} })^{*}$. In general, in Heisenberg picture,
\begin{align}
| \Psi_{ \pp\Gamma, \{ I_{v} \} }(\{g_{e}\}_{e\in\Gamma^{o}}) \ra
=&
\sum_{ \{ I_{u_{i} } \}  }
\bigutimes_{i=1}^{n}
\left(
\tensor{   {\Big[ \cU_{\Upsilon_{i} } \Big] }   } {^{ I_{ u_{i} }  }  _{  \{ I_{v} \}_{v\in\Gamma_{i}}    }    }
\circ
\cG_{\cB_{i} } \act
| \Psi_{\pp\Upsilon_{i}, I_{u_{i} } } ( \{ G_{e} \}_{e\in \Upsilon_{i}^{o} } ) \ra
\right)
\nn \\
=&
\sum_{ \{ I_{u_{i} } \}  }
\prod_{i=1}^{n}
\tensor{   {\Big[ \cU_{\Upsilon_{i} } \Big] }   } {^{ I_{ u_{i} }  }  _{  \{ I_{v} \}_{v\in\Gamma_{i}}    }    }
\circ
\left(
\bigutimes_{i=1}^{n}
\cG_{\cB_{i} } \act
| \Psi_{\pp\Upsilon_{i}, I_{u_{i} } } ( \{ G_{e} \}_{e\in \Upsilon_{i}^{o} } ) \ra
\right)
\,. \label{eq:Unitary-SPN,CG-Gluing}
\end{align}
Here every $\cU_{\Upsilon_{i} }$ represents channel transformation for loopy intertwiner in $\cH_{\Upsilon}$, and every $\cG_{\cB_{i}}$ represents boundary holonomies for $\cH_{\pp\Gamma_{i} }$. They do not change the boundary spins, thus all spin-matching constraints hold. Every $\cU_{\Upsilon_{i} }$ and $\cG_{\cB_{i}}$ are local unitary for respect $(\cH_{\cB_{i} })^{*}$ so we factorize $\cU_{\Upsilon_{i} }$ and keep $\cG_{\cB_{i}}$ in gluing operation.
Notably, boundary holonomies on $\cB_{i}$ never change the scalar products for $\cH_{\cB_{i}}$ and $(\cH_{\cB_{i}})^{*}$. 
We then have transformation between density matrices:
\beq \label{eq:DensityMatrix-Transformation}
\rho_{\pp\Gamma^{*}}
=
\cU_{\cB_{1} }^{\dagger} \cdots \cU_{\cB_{n} }^{\dagger}
\, \rho_{\pp\Gamma_{(R)}^{*}} \, 
\cU_{\cB_{1} } \cdots \cU_{\cB_{n} }
\,.
\eeq
Again, every $\cU_{\cB_{i} }$ is to be interpreted as a local unitary transformations with respect to $(\cH_{\cB_{i}})^{*}$. Therefore, the $\rho_{\pp\Gamma^{*}}$ and $\rho_{\pp\Gamma_{(R)}^{*}}$ carry equal spin network entanglement with respect to the partition, since any entanglement measure is required to be invariant under local unitary transformations \cite{PhysRevA.68.042307,GUHNE20091}.
\begin{res}
Let $\cE$ be any entanglement measure.
Given a partition $\Gamma=\bigsqcup_{i=1}^{n}\Gamma_{i}$. Let $\Gamma_{(R)}=\bigsqcup_{i=1}^{n} \Upsilon_{i}$ be the coarse-grained graph for $\Gamma$ where $\Upsilon_{i}$ are loopy graphs for respect $\Gamma_{i}$. For any spin network state based on the $\Gamma$, the corresponding coarse-grained state based on the $\Gamma_{(R)}$ is defined via gauge-fixing. Then for density matrix $\rho_{\pp\Gamma^{*}}$ and the corresponding density matrix $\rho_{\pp\Gamma^{*}_{(R)} }$ for the coarse-grained state, the reduced density matrices $\rho_{ \pp\Gamma^{*}_{i} }$ and $\rho_{\pp\Upsilon^{*}_{i} }$ are related by
\beq \label{eq:ReducedDensityMatrix-GroupIntegralPartialTrace-ReducedGraph}
\rho_{ \pp\Gamma^{*}_{i} }
=
\cU_{\cB_{i} }^{\dagger} \,
\rho_{\pp\Upsilon^{*}_{i} }
\, \cU_{\cB_{i} }
\,.
\eeq
Here $\cU_{\cB_{i} }$ is unitary for $(\cH_{\cB_{i}})^{*}$. Moreover, $\rho_{ \pp\Gamma^{*} }$ and $\rho_{\pp\Gamma^{*}_{(R)} }$ carry equal spin network entanglement,
\be
\cE[ \rho_{\pp\Gamma^{*}} ]=\cE[ \rho_{\pp\Gamma^{*}_{(R)} } ]
\,.
\ee
Due to the Proposition \ref{Prop:DensityMatrix-SPN-DBH}, $\cE[ \rho_{ \Gamma } ]=\cE[ \rho_{ \Gamma_{(R)} } ]$.
\end{res}

\smallskip

At the end of the day, we have shown that spin network entanglement allows to be coarse-grained, thus one can study the entanglement by studying (loopy) intertwiner entanglement on coarser graph.

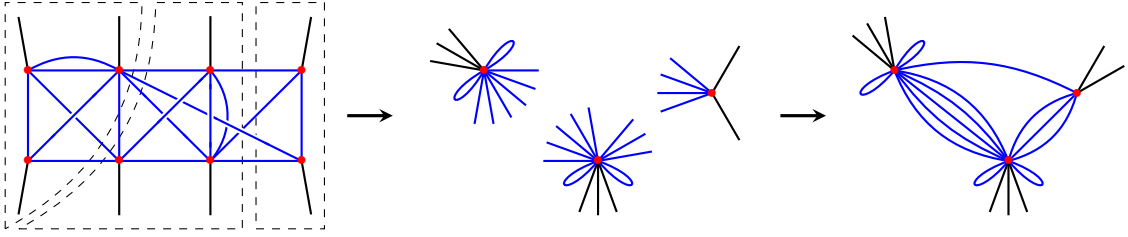
\begin{figure}[hbt!]
\centering
\begin{tikzpicture}[scale=0.6]

\coordinate (O1) at (0,-1);
\coordinate (O2) at (2,-1);
\coordinate (O3) at (4,-1);
\coordinate (O4) at (6,-1);
\coordinate (O5) at (0,1);
\coordinate (O6) at (2,1);
\coordinate (O7) at (4,1);
\coordinate (O8) at (6,1);

\draw[thick] (O1) --++(260:1.2);
\draw[thick] (O2) --++(-90:1.2);
\draw[thick] (O3) --++(-90:1.2);
\draw[thick] (O4) --++(-80:1.2);
\draw[thick] (O5) --++(100:1.2);
\draw[thick] (O6) --++(90:1.2);
\draw[thick] (O7) --++(90:1.2);
\draw[thick] (O8) --++(80:1.2);

\coordinate (G11) at (-0.5,-2.5);
\coordinate (G12) at (-0.5,2.5);
\coordinate (G13) at (2.5,2.5);

\draw[dashed] (G11) -- (G12) -- (G13) to [bend left=30] (G11);

\path (G11) ++ (0.3,0) coordinate (G21);
\coordinate (G22) at (4.7,-2.5);
\coordinate (G24) at (4.7,2.5);
\path (G13) ++ (0.3,0) coordinate (G23);

\draw[dashed] (G21) -- (G22) -- (G24) -- (G23) to [bend left=30] (G21);

\coordinate (G31) at (5,-2.5);
\coordinate (G32) at (6.5,-2.5);
\coordinate (G33) at (6.5,2.5);
\coordinate (G34) at (5,2.5);

\draw[dashed] (G31) -- (G32) -- (G33) -- (G34) -- (G31);

\draw[thick,blue] (O1) to [bend right=0] (O2)  to [bend right=0] (O3)  to [bend right=0]  (O4)  to [bend right=0]  (O8)  to [bend right=0]  (O7)  to [bend right=0]  (O6)  to [bend right=0]  (O5)  to [bend right=0] (O1);

\draw[thick,blue] (O5) to[bend left=30](O6);

\begin{knot}[
  consider self intersections=true,
  flip crossing=3,
]
\strand[thick,blue]
    (O1) -- (O6) --  (O3) -- (O8);
\strand[thick,blue]
(O5) -- (O2)--(O7) -- (O3);
\strand[thick,blue]
(O2)--(O6) -- (O4);
\strand[thick,blue]
(O3) to [bend right=40] (O7);
\end{knot}

\draw (O1) node[scale=0.7,red] {$\bullet$};
\draw (O2) node[scale=0.7,red] {$\bullet$};
\draw (O3) node[scale=0.7,red] {$\bullet$};
\draw (O4) node[scale=0.7,red] {$\bullet$};
\draw (O5) node[scale=0.7,red] {$\bullet$};
\draw (O6) node[scale=0.7,red] {$\bullet$};
\draw (O7) node[scale=0.7,red] {$\bullet$};
\draw (O8) node[scale=0.7,red] {$\bullet$};
\draw[->,>=stealth,very thick] (7,0) -- (8,0);
\coordinate (A1) at (10,1);
\coordinate (A2) at (12.5,-1);
\coordinate (A3) at (15,0.5);

\draw[thick] (A1) --++(130:1.2);
\draw[thick] (A1) --++(150:1.2);
\draw[thick] (A1) --++(170:1.2);
\draw[thick,blue] (A1) --++(0:1.2);
\draw[thick,blue] (A1) --++(-20:1.2);
\draw[thick,blue] (A1) --++(-40:1.2);
\draw[thick,blue] (A1) --++(-60:1.2);
\draw[thick,blue] (A1) --++(-80:1.2);
\draw[thick,blue] (A1) --++(-100:1.2);
\draw[blue,thick,in=60,out=30,scale=2.5,rotate=0] (A1)  to[loop] (A1);
\draw[blue,thick,in=240,out=210,scale=2.5,rotate=0] (A1)  to[loop] (A1);

\draw[thick] (A2) --++(250:1.2);
\draw[thick] (A2) --++(270:1.2);
\draw[thick] (A2) --++(290:1.2);
\draw[thick,blue] (A2) --++(140:1.2);
\draw[thick,blue] (A2) --++(160:1.2);
\draw[thick,blue] (A2) --++(180:1.2);
\draw[thick,blue] (A2) --++(120:1.2);
\draw[thick,blue] (A2) --++(100:1.2);
\draw[thick,blue] (A2) --++(10:1.2);
\draw[thick,blue] (A2) --++(50:1.2);
\draw[thick,blue] (A2) --++(30:1.2);
\draw[blue,thick,in=-50,out=-20,scale=2.5,rotate=0] (A2)  to[loop] (A2);
\draw[blue,thick,in=230,out=200,scale=2.5,rotate=0] (A2)  to[loop] (A2);

\draw[thick] (A3) --++(60:1.2);
\draw[thick] (A3) --++(-60:1.2);
\draw[thick,blue] (A3) --++(140:1.2);
\draw[thick,blue] (A3) --++(160:1.2);
\draw[thick,blue] (A3) --++(180:1.2);
\draw[thick,blue] (A3) --++(200:1.2);

\draw (A1) node[scale=0.7,red] {$\bullet$};
\draw (A2) node[scale=0.7,red] {$\bullet$};
\draw (A3) node[scale=0.7,red] {$\bullet$};
\draw[->,>=stealth,very thick] (16.5,0) -- (17.5,0);
\coordinate (B1) at (19,1);
\coordinate (B2) at (21.5,-1);
\coordinate (B3) at (23,0.5);

\draw[thick] (B1) --++(100:1.2);
\draw[thick] (B1) --++(140:1.2);
\draw[thick] (B1) --++(120:1.2);
\draw[thick] (B2) --++(270:1.2);
\draw[thick] (B2) --++(250:1.2);
\draw[thick] (B2) --++(290:1.2);
\draw[thick] (B3) --++(60:1.2);
\draw[thick] (B3) --++(30:1.2);

\draw[thick,blue] (B1) to [bend left=20] (B3);
\draw[thick,blue] (B1) to [bend left=30] (B2);
\draw[thick,blue] (B1) to [bend left=15] (B2);
\draw[thick,blue] (B1) to [bend left=0] (B2);
\draw[thick,blue] (B1) to [bend right=30] (B2);
\draw[thick,blue] (B1) to [bend right=15] (B2);

\draw[thick,blue] (B3) to [bend left=0] (B2);
\draw[thick,blue] (B3) to [bend right=30] (B2);
\draw[thick,blue] (B3) to [bend left=30] (B2);

\draw[blue,thick,in=60,out=30,scale=2.5,rotate=0] (B1)  to[loop] (B1);
\draw[blue,thick,in=230,out=200,scale=2.5,rotate=0] (B1)  to[loop] (B1);

\draw[blue,thick,in=-50,out=-20,scale=2.5,rotate=0] (B2)  to[loop] (B2);
\draw[blue,thick,in=230,out=200,scale=2.5,rotate=0] (B2)  to[loop] (B2);

\draw (B1) node[scale=0.7,red] {$\bullet$};
\draw (B2) node[scale=0.7,red] {$\bullet$};
\draw (B3) node[scale=0.7,red] {$\bullet$};

\end{tikzpicture}
\caption{Illustrations for coarse-graining: from fine-graph to coarse-grained graph. The coarse-graining preserves the spin network entanglement.
}
\label{fig:ReducingGraph}
\end{figure}

\section{Dynamics of entanglement coarse-graining: loop holonomy operators} \label{Section:DynamicsCoarseGraining}
So far we have shown the coarse-graining of spin network entanglement at kinematical level. In this section, we investigate how the coarse-graining method can be extended to dynamical level. Suppose $\rho_{\Gamma}$'s evolution and the dynamics of spin network entanglement, how to study the entanglement dynamics from coarse-grained graph $\Gamma_{(R)}$? We consider evolution generated by loop holonomy operator, and we show that the dynamics of spin network entanglement based on the $\Gamma$ is exactly reflected in the dynamics of spin network entanglement based on the coarse-grained graph $\Gamma_{(R)}$.

\subsection{Loop holonomy operator} \label{Section:LoopHolonomyOperator}
The dynamics of spin network states implement the flow generated by the Hamiltonian constraints on the embedded geometry of the canonical hypersurface. At the quantum level, the Hamiltonian constraint operators involve the holonomy operator and geometric observables, such as areas and volumes. The holonomy operator is analogous to the Wilson-loop in QCD. It corresponds to the quantization of the curvature in the polymer quantization scheme used in loop quantum gravity, where one does not access to point-like excitations, but only to gauge-invariant observables smeared along 1d structures. It is a non-local operator that excites non-local correlation and entanglement \cite{Chen:2022rty}. 

\smallskip

Let us analyze the action of the holonomy operator on spin network basis states, along the lines of \cite{Bonzom:2009zd,Borja:2010gn}, or \cite{Chen:2022rty} for a recent work. Let us look at the holonomy operator with spin-$\ell$ acting on a single edge $e$. This operator takes the tensor product of the spin-$\ell$ with the spin-$k_{e}$ carried by the edge, and its action can be expressed in terms of Clebsch-Gordan coefficients decomposing this tensor product $\ell\otimes k_{e}$ into irreducible representations, i.e. $\cV_{\ell} \otimes \cV_{k_{e} } = \bigoplus_{K_{e}=| k_{e}-\ell | }^{ k_{e}+\ell } \cV_{K_{e} }$. Indeed the Wigner matrices for the $\SU(2)$ group element $g_{e}$ carrying the holonomy along the edge $e$ satisfies the following algebraic relations:
\beq
\widehat{ D^{\ell}_{a_{e} b_{e} } }
\act
\Big[ D^{k_{e} }_{m_{e} n_{e}}(g_{e}) \Big]
&=&
D^{\ell}_{a_{e} b_{e} } (g_{e} ) \,
D^{k_{e} }_{m_{e} n_{e}}(g_{e} )
\,, \label{eq:HolonomyOperator-TensorProduct}
\\
&=&
\sum_{K_{e}=| k_{e}-\ell | }^{k_{e} + \ell } \,
\sum_{ M_{e}, N_{e}=-K_{e} }^{K_{e} } \,
(-1)^{M_{e}-N_{e} }
(2K_{e}+1)
\nn
\\
&&
\times
\begin{pmatrix}
   \ell & k_{e} & K_{e} \\
   a_{e} & m_{e} & -M_{e}
  \end{pmatrix}
\overline{ \begin{pmatrix}
   \ell & k_{e} & K_{e} \\
   b_{e} & n_{e} & -N_{e}
  \end{pmatrix} } \,
D^{K_{e} }_{M_{e} N_{e} }(g_{e})
  \,, \label{eq:HolonomyOperator-3j}
\eeq
where the recoupled spin-$K_{e}$ is bounded by the triangular inequalities $| k_{e}-\ell | \leq K_{e} \leq k_{e} +\ell$.

\smallskip

The holonomy operator along a single edge spoils the gauge invariance. In order to produce a gauge-invariant
holonomy operator, one must consider a closed loop on the graph $\Gamma$ underlying the spin network state. Consider a loop $W \subset \Gamma$ with $n$ edges, and assume the simplifying condition that it does not go through a vertex more than that once. The oriented loop $W$ can be described as the path $W[v_1 \overset{ e_{1} }{\to} \cdots \overset{ e_{p-1} }{\to} v_p \overset{ e_{p} }{\to} v_1]$ such that the edge $e_{\alpha}$ links the vertices $v_{\alpha}$ to $v_{\alpha+1}$, with $\alpha=1,\cdots, p$ and the implicit convention $p+1\equiv 1$. The loop holonomy operator is defined as a multiplicative operator on the wave-functions:
\beq \label{eq:LoopHolonomyOperator-BoundaryMap}
\left( \wh{\chi_{\ell} } \, \act_{W} \psi_{\Gamma} \right) ( \{g_e\}_{e\in\Gamma}  )
=
\chi_{\ell} (G_{W} ) \psi_{\Gamma} ( \{g_e\}_{e\in\Gamma}  )
\,, \qquad \text{with} \quad 
G_{W}=\overleftarrow{\prod_{ e_{\alpha} \in W } g_{ e_{\alpha} }   }
\,,
\eeq
where $\chi_{\ell}(g)=\tr D^{\ell}(g)$ is the character of the spin-$\ell$ representation. We take the inverse of a group element if the edge is oriented in the opposite direction than the loop. Since the factor $\chi_{\ell}(G_{W})$ is gauge invariant function, the resulting wave-function is still gauge-invariant. Thus the map $\left( \wh{\chi_{\ell} } \, \act_{W} \right)$ acts legitimately on the Hilbert space $\cH_{\Gamma}$ and we can write its action on the spin network basis:
\beq \label{eq:LoopHolonomyOperator-SpinNetwork}
\wh{\chi_{\ell} } \, \act_{W} \,
| \Psi_{ \Gamma, \{I_{v}\} } \ra
=
\sum_{ \{I'_v\} }
\tensor{   {\Big[ \tensor{Z(\Gamma) }{ _{\tensor{\chi}{_{\ell} }\, \act_{W}    }   } \Big] }   } {^{ \{ I'_{v} \} } _{ \{ I_{v} \} } }
\,
| \Psi_{ \Gamma, \{I'_{v}\} } \ra
\,,
\eeq
where the matrix elements $\tensor{Z(\Gamma) }{ _{\tensor{\chi}{_{\ell} }\, \act_{W}    }    }$ are given by the following integrals:
\beq
\tensor{   {\Big[ \tensor{Z(\Gamma) }{ _{\tensor{\chi}{_{\ell} }\, \act_{W}    }    } \Big] }   } {^{ \{ I'_{v} \} } _{ \{ I_{v} \} } }
=
\int
\prod_{e\in\Gamma}\rd g_{e} \,
\overline{ \Psi_{\Gamma, \{I'_{v}\} } (\{ g_{e} \}_{e\in\Gamma}) }
\chi_{\ell} (G_{W})
\Psi_{\Gamma, \{I_{v}\} } ( \{ g_{e} \}_{e\in\Gamma}  )
\,. \label{eq:TransitionMatrix}
\eeq
This matrix $\tensor{Z(\Gamma) }{ _{\tensor{\chi}{_{\ell} }\, \act_{W}    }    }$ satisfies a composition rule:
\beq
\sum_{ \{ I'_{v} \} }
\tensor{   {\Big[ \tensor{Z(\Gamma) }{   _{\tensor{\chi}{_{\ell_{1} } }\, \act_{W}    }   } \Big] }   } {^{ \{ I''_{v} \} } _{ \{ I'_{v} \} } }
\,
\tensor{   {\Big[ \tensor{Z(\Gamma) }{ _{\tensor{\chi}{_{ \ell_{2} } }\, \act_{W}    }  } \Big] }   } {^{ \{ I'_{v} \} } _{ \{ I_{v} \} } }
&=&
\sum_{ s=| {\ell}_{1} - {\ell}_{2} | }^{ {\ell}_{1} + {\ell}_{2} }
\tensor{   {\Big[ \tensor{Z(\Gamma) }{    _{\tensor{\chi}{_s}\, \act_{W}    }    } \Big] }   } {^{ \{ I''_{v} \} } _{ \{ I_{v} \} } }
\nn
\\
&=&
\tensor{   {\Big[ \tensor{Z(\Gamma) }{    _{ ( \tensor{\chi}{_{ {\ell}_{1} } } \cdot \tensor{\chi}{_{ {\ell}_{2} } } ) \, \act_{W}    }    } \Big] }   } {^{ \{ I''_{v} \} } _{ \{ I_{v} \} } }
\,, \label{eq:Composition-Amplitudes}
\eeq
which is inherited from the character recoupling formula $\chi_{ {\ell}_{1} } \cdot \chi_{ {\ell}_{2} } = \sum_{ s=| {\ell}_1- {\ell}_2 | }^{ {\ell}_1+{\ell}_2 } \, \chi_{s}$. An interesting fact to keep in mind is that the matrices $\tensor{Z(\Gamma) }{ _{\tensor{\chi}{_{\ell_1} }\, \act_{W}    }    }$ and $\tensor{Z(\Gamma) }{ _{\tensor{\chi}{_{\ell_2} }\, \act_{W}    }    }$ commute with each other with arbitrary spins $\ell_1$ and $\ell_2$.

\smallskip

The transition matrix $Z$ can be expressed in terms of the $\{6j\}$ symbols of spin recoupling, where every $6j$-symbol associates with a corresponding vertex along the loop $W$, and is expressed in terms of spins along the loop and the bouquet spin, as illustrated on fig.\ref{fig:BouquetSpins}. The following lemma gives the expression for the matrix $Z$ \cite{Chen:2022rty}. 
\begin{lemma} \label{lemma:Amplitudes-6jsymbols}
Given an oriented loop $W[v_1 \overset{ e_{1} }{\to} \cdots \overset{ e_{p-1} }{\to} v_p \overset{ e_{p} }{\to} v_1]$ on the graph $\Gamma$, the loop holonomy operator $\wh{\chi}_{ {\ell} } \, \act_{W}$ acts on the spin network basis, labeled by the spins $k_{\alpha}$ on the loop edges and the bouquet spins $j_{\alpha}$ on the loop vertices, by following transition matrix:
\begin{align}
\tensor{   {\Big[ Z(\Gamma)_{\tensor{\chi}{_{\ell} }\, \act_{W}   } \Big] }   } {^{ \{ J_{\alpha}, K_{\alpha} \} } _{ \{ j_{\alpha}, k_{\alpha} \} } }
=
(-1)^{\sum_{\alpha=1}^{p} (j_{\alpha} +k_{\alpha} +K_{\alpha} +{\ell} ) }
\,
&
\prod_{\alpha=1}^{p} \sqrt{ (2K_{\alpha}+1)(2k_{\alpha}+1) } \, \delta_{J_{\alpha} j_{\alpha}}
\nn
\\
& \times
\overleftarrow{ \prod_{\alpha=1}^{p} \,
\begin{Bmatrix}
j_{\alpha} & K_{\alpha} & K_{\alpha-1} \\
\ell & k_{\alpha-1} & k_{\alpha}
\end{Bmatrix}
 }
\,.
\label{eq:Amplitudes-6jsymbols}
\end{align}
The $\delta_{J_{\alpha} j_{\alpha}}$ is because loop holonomy operator does not change the bouquet spins.
\end{lemma}

\smallskip

The next question is to relate the action on $\Gamma$ to the action on $\Gamma_{(R)}$.
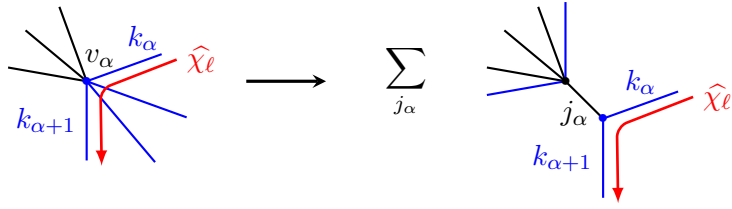
\begin{figure}
\centering
\begin{tikzpicture}[scale=0.7]

\coordinate (O) at (0,0);
\coordinate (O3) at (-30:0.3);

\draw[thick] (O) -- ++(170:1.5);
\draw[thick] (O) -- ++(110:1.5);
\draw[thick] (O) -- ++(140:1.5);

\begin{knot}[
  consider self intersections=true,
  flip crossing=3,
]
\strand[thick,blue] (O) -- ++ (-50:1.5) --  ++ (-50:0.5);
\strand[thick,blue] (O) -- ++(-20:1.5) -- ++ (-20:0.5);
\strand[thick,blue] (O) -- node[midway,left] {$k_{\alpha+1}$} ++(-90:1.5) ++ (-20:0.3) coordinate (O1);
\strand[thick,blue] (O) -- node[near end,above] {$k_{\alpha}$} ++(20:1.5) ++ (-20:0.3) coordinate (O2);
\strand [line width=1pt,red,->,>=latex,rounded corners] (O2) -- (O3) -- (O1);
\end{knot}

\draw (O) node[scale=0.7,blue] {$\bullet$} ++(60:0.5)node {$v_{\alpha}$};

\draw (O2) node[right,red] {$\widehat{\chi_{\ell} }$};

\draw[->,>=stealth,very thick] (O) ++(3,0) -- ++ (1.5,0);

\coordinate (A) at (9,0);

\draw (A)  ++ (-3,0) node {$ \displaystyle{ \sum_{ j_{\alpha} } }$};

\draw[thick] (A) -- ++(170:1.5);
\draw[thick] (A) -- ++(110:1.5);
\draw[thick] (A) -- ++(140:1.5);

\draw[thick,blue] (A) -- ++(90:1.5);
\draw[thick,blue] (A) -- ++(190:1.5);

\draw[thick] (A) -- node[below,near start] {$j_{\alpha}$} ++(-45:1) coordinate (A1);

\path (A1) ++ (-30:0.3) coordinate (A4);

\draw[thick,blue] (A1) -- node[midway,left] {$k_{\alpha+1}$} ++(-90:1.5) ++ (-20:0.3) coordinate (A2);
\draw[thick,blue] (A1) -- node[midway,above] {$k_{\alpha}$} ++(20:1.5) ++ (-20:0.3)coordinate (A3);

\draw [line width=1pt,red,->,>=latex,rounded corners] (A3) -- (A4) -- (A2);

\draw (A) node[scale=0.7] {$\bullet$};
\draw (A1) node[scale=0.7,blue] {$\bullet$};

\draw (A3) node[right,red] {$\widehat{\chi_{\ell} }$};

\end{tikzpicture}

\caption{For a vertex $v_{\alpha}$ living along the $W$, the bouquet spin-$j_{\alpha}$ is the recoupled spin of those spins that are not living along the $W$.
}
\label{fig:BouquetSpins}
\end{figure}
%


\subsection{Coarse-graining of loop holonomy operator}
One of the goals of this part is to show that the action of loop holonomy operator on a graph can be exactly mapped to the action on its coarse-grained graph. 
We present the transformation rule for the two actions. Based on that, we show that the dynamics of spin network entanglement can be exactly studied from the coarse-grained graph.

Coarse-graining the loop $W$ is done by keeping the edges that are not coarse-grained, then by gluing them into loop to be referred to $W_{(R)}$, e.g. Fig.\ref{fig:ReducingGraph-OneLoop}.

Two trivial situations are not to be considered: (i) $W$ is completely isolated in certain $\Gamma_{i}$. (ii) Trivial coarse-graining $W_{(R)}=W$.

According to Lemma \ref{lemma:Amplitudes-6jsymbols}, the evolution on these loopy spin sub-networks $\{ \Upsilon_{i} \}$ is determined by the bouquet spins and the spins on the loop edges. So the transition matrix $Z$ can be obtained from studying the type of particular graphs that the bouquet spins are represented by boundary spins, such as Fig.\ref{fig:ReducingGraph-OneLoop}.

Now consider two evolutions: (a) Evolution $| \psi_{\Gamma} \ra \to | \psi_{\Gamma}(t) \ra$ generated by loop holonomy operator acting on $W \subset \Gamma$. (b) Evolution $| \psi_{\Gamma_{(R)}} \ra \to | \psi_{\Gamma_{(R)}}(t) \ra$ generated by loop holonomy operator action on $W_{(R)} \subset \Gamma_{(R)}$.
Then we arrive the following result:
\begin{figure}
\centering
\begin{tikzpicture}[scale=0.6]

\coordinate (O) at (3,0);
\coordinate (P) at (4.5,0);
\coordinate (A) at (-6,0);

\draw [thick,domain=0:360,blue] plot ({-6+1.75 * cos(\x)}, {1.75 * sin(\x)});
\draw[thick] (A) ++(-0.6,0);
\draw[thick] (A) ++(0:1.75)  node[scale=0.7,blue] {$\bullet$}--++ (0:1);
\draw[thick] (A) ++(60:1.75) node[scale=0.7,blue] {$\bullet$} --++ (60:1) ++(60:0.35);
\draw[thick] (A) ++(120:1.75) node[scale=0.7,blue] {$\bullet$}--++ (120:1) ++(120:0.35);
\draw[thick] (A) ++(180:1.75) node[scale=0.7,blue] {$\bullet$}--++ (180:1) ++(180:0.35);
\draw[thick] (A) ++(240:1.75) node[scale=0.7,blue] {$\bullet$} --++ (240:1);
\draw[thick] (A) ++(300:1.75) node[scale=0.7,blue] {$\bullet$}--++ (300:1);


\draw[dashed] (A) ++(-3,-3) rectangle (-6.2,3);
\draw (A) ++ (-2.5,-2.5) node[above]{$\Gamma_{1}$};
\draw[dashed] (A) ++(0.2,-3) rectangle (-3,3);
\draw (A) ++ (2.5,-2.5) node[above]{$\Gamma_{2}$};

\draw[->,>=stealth,very thick] (-2,0) -- (-0.5,0);

\draw [thick,domain=90:120,blue] plot ({3+1.75 * cos(\x)}, {1.75 * sin(\x)});
\draw [thick,domain=240:270,blue] plot ({3+1.75 * cos(\x)}, {1.75 * sin(\x)});
\draw [thick,domain=120:240,green] plot ({3+1.75 * cos(\x)}, {1.75 * sin(\x)});
\draw [thick,domain=60:90,blue] plot ({4.5+1.75 * cos(\x)}, {1.75 * sin(\x)});
\draw [thick,domain=-60:60,green] plot ({4.5+1.75 * cos(\x)}, {1.75 * sin(\x)});
\draw [thick,domain=-90:-60,blue] plot ({4.5+1.75 * cos(\x)}, {1.75 * sin(\x)});

\draw[thick] (P) ++(0:1.75)  node[scale=0.7,blue] {$\bullet$}--++ (0:1);
\draw[thick] (P) ++(60:1.75) node[scale=0.7,blue] {$\bullet$} --++ (60:1) ++(60:0.35);
\draw[thick] (O) ++(120:1.75) node[scale=0.7,blue] {$\bullet$}--++ (120:1) ++(120:0.35);
\draw[thick] (O) ++(180:1.75) node[scale=0.7,blue] {$\bullet$}--++ (180:1) ++(180:0.35);
\draw[thick] (O) ++(240:1.75) node[scale=0.7,blue] {$\bullet$} --++ (240:1);
\draw[thick] (P) ++(300:1.75) node[scale=0.7,blue] {$\bullet$}--++ (300:1);

\draw[->,>=stealth,very thick] (8,0) -- (9.5,0);

\coordinate (A1) at (11.5,0);
\coordinate (A2) at (13.5,0);

\draw (A1) ++ (0,-0.75) node {$u_1$};
\draw (A2) ++ (0,-0.75) node {$u_2$};

\draw (A1) ++ (1,-2) node {$\Gamma_{(R)}$};

\draw[thick] (A1) -- ++ (135:1);
\draw[thick] (A1) -- ++ (180:1);
\draw[thick] (A1) -- ++ (225:1);

\draw[thick] (A2) -- ++ (45:1);
\draw[thick] (A2) -- ++ (0:1);
\draw[thick] (A2) -- ++ (-45:1);


\draw[blue,thick,in=115,out=65,rotate=0] (A1) to (A2) node[scale=0.7] {$\bullet$} to[out=245,in=-65] (A1) node[scale=0.7] {$\bullet$};

\end{tikzpicture}
%
%
\caption{The sub-networks $\Gamma_{1}$ and $\Gamma_{2}$ are coarse-grained to $\Upsilon_{1}$ and $\Upsilon_{2}$ for $\Gamma_{(R)}$ via gauge-fixing along the maximal trees $T_{1}$ and $T_{2}$ ({\color{green}{green}}). The $W$ and coarse-grained loop $W_{(R)}$ are colored in {\color{blue}{blue}}.
}
\label{fig:ReducingGraph-OneLoop}
\end{figure}
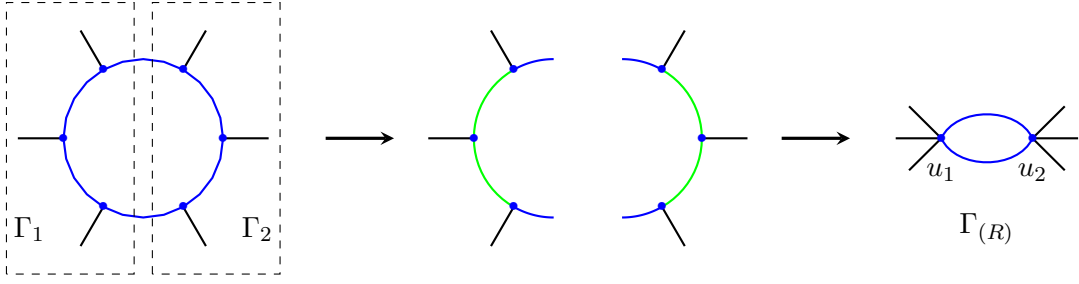

\begin{res} \label{Prop:CG-LoopHolonomyOperator}
Suppose the oriented loop $W[v_1 \overset{ e_{1} }{\to} \cdots \overset{ e_{p-1} }{\to} v_p \overset{ e_{p} }{\to} v_1]$ is partitioned by $n$ sub-networks $\{ \Gamma_{i} \}_{i=1}^{n}$, i.e. $W=\bigsqcup_{i=1}^{n} W_{i}$ where $W_{i} \subset \Gamma_{i}$.
For every $\Gamma_{i}$, we relabel the bouquet spins $j_{\alpha}$ and bulk spins $k_{\alpha}$ along the $W_{i}$ and denote them by
\be
\{ j^{(i)}_{\alpha} \} \equiv \{ j^{(i)}_{1}\,, \cdots \,, j^{(i)}_{p_{i}} \}\,,
\qquad
\{ k_{\alpha}^{(i)} \} \equiv \{ k^{(i)}_{0} \,, k^{(i)}_{1}\,, \cdots \,, k^{(i)}_{p_{i}} \}
=
\{ k_{\alpha}^{o(i)} \} \sqcup
\{ k_{\alpha}^{\pp(i)} \}
\nn
\ee
with the implicit convention $k^{(i)}_{p_{i}}=k^{(i+1)}_{0}$ due to the spin-matching constraint, where
\be
\{ k_{\alpha}^{o(i)} \} \equiv \{ k^{(i)}_{1}\,, \cdots \,, k^{(i)}_{p_{i}-1} \}\,,
\qquad
\{ k_{\alpha}^{\pp(i)} \} \equiv \{ k^{(i)}_{0}\,, k^{(i)}_{p_{i}} \}\,,
\nn
\ee
because $\{ k_{\alpha}^{o(i)} \}$ are bulk spins for $\Gamma_{i}$, and $\{ k_{\alpha}^{\pp(i)} \}$ are part of boundary spins for $\Gamma_{i}$. Spins $\{ j_{\alpha}^{(i)}, k_{\alpha}^{\pp(i)} \}$ define cluster intertwiners
\be
I_{u_{i}} \in\mathrm{Inv}_{\SU(2)}\Big{[}
\cV_{k_{0}^{(i)} } \otimes
\cV_{k_{p_{i}}^{(i)} } \otimes \
\bigotimes_{\alpha=1}^{ p_{i} } \cV_{j_{\alpha} }
\Big{]}
\,.
\ee
Now the cluster intertwiner is labeled by $\{ j_{\alpha}^{(i)}, k_{\alpha}^{(i)} \}$, and recoupling spins $\{ j_{\alpha}^{(i)}, k_{\alpha}^{\pp(i)} \}$ leads to the representation $\cU_{\Upsilon_{i}} : \cH_{\Upsilon_{i}} \longrightarrow \cH_{\Upsilon_{i}}$ for channel transformation:
\be
\tensor{   {\Big[ \, \cU_{\Upsilon_{i} }^{ \{ j_{\alpha}^{(i)}, k_{\alpha}^{\pp(i)} \} } \, \Big] }   } {^{ I_{u_{i}}^{o} }    }_{ \{ k_{\alpha}^{o(i)} \} }
=
\la \Psi_{ \Upsilon_{i}, I_{u_{i}} }
| \Psi_{ \Upsilon_{i}, \{ j_{\alpha }^{(i)}, k_{\alpha}^{(i)} \} } \ra
\,.
\ee
Here $I_{u_{i}}=\{ j_{\alpha }^{(i)}, k_{\alpha}^{\pp(i)} \} \cup I_{u_{i}}^{o} $ equivalently labels the cluster intertwiner $\{ j_{\alpha}^{(i)}, k_{\alpha}^{(i)} \}$ (illustrated as Fig.\ref{fig:ClusterIntertwiner}). 
Tensoring representations leads to
\be
\tensor{   {\Big[ \, \cU_{\Gamma_{(R)} } \, \Big] }   } {^{ \{ I_{u } \} }  _{ \{ j_{\alpha}, k_{\alpha} \}   }    }
=
\prod_{i=1}^{n}
\tensor{   {\Big[ \, \cU_{\Upsilon_{i} }^{ \{ j_{\alpha}^{(i)}, k_{\alpha}^{\pp(i)} \} } \, \Big] }   } {^{ I_{u_{i}}^{o} }    }_{ \{ k_{\alpha}^{o(i)} \} }
\,.
\ee
Then the transition matrices based on $\Gamma$ and $\Gamma_{(R)}$ satisfy 
\beq
\tensor{   {\Big[ Z(\Gamma)_{\tensor{\chi}{_{\ell} }\, \act_{W}   } \Big] }   } {^{ \{ J_{\alpha}, K_{\alpha} \} } _{ \{ j_{\alpha}, k_{\alpha} \} } }
&=&
\sum_{ \{ I_{u_{i}}^{o} \} } \sum_{ \{ \tl{I}_{u_{i}}^{o} \} }
\overline{
\tensor{   {\Big[ \, \cU_{\Gamma_{(R)} } \, \Big] }   } {^{ \{ \tl{I}_{u_{i} } \} }  _{ \{ J_{\alpha}, K_{\alpha} \}   }    }
}
\tensor{   {\Big[ \tensor{ Z( \Gamma_{(R)} ) }{ _{\tensor{\chi}{_{\ell} }\, \act_{W_{(R)} }    }    } \Big] }   } {^{ \{ \tl{I}_{u_{i} } \} } _{ \{ I_{u_{i} } \} } }
\nn \\
&& \phantom{ \sum_{ \{ I_{u_{i}}^{o} \} } \sum_{ \{ \tl{I}_{u_{i}}^{o} \} } }
\tensor{   {\Big[ \, \cU_{\Gamma_{(R)} } \, \Big] }   } {^{ \{ I_{u_{i} } \} }  _{ \{ j_{\alpha}, k_{\alpha} \}   }    }
\prod_{\alpha=1}^{p} \delta_{ J_{\alpha} j_{\alpha} }
\,, \label{eq:LoopHolonomyOperator-CG-GammaLoopy}
\eeq
where $I_{u_{i}}=\{ j_{\alpha }^{(i)}, k_{\alpha}^{\pp(i)} \} \cup I_{u_{i}}^{o} $ and $\tl{I}_{u_{i}}=\{ J_{\alpha }^{(i)}, K_{\alpha}^{\pp(i)} \} \cup \tl{I}_{u_{i}}^{o} $, and the transition matrix $Z( \Gamma_{(R)} ) $ is given by
\beq
&&
\tensor{   {\Big[ \tensor{ Z( \Gamma_{(R)} ) }{ _{\tensor{\chi}{_{\ell} }\, \act_{W_{(R)} }    }    } \Big] }   } {^{ \{ \tl{I}_{u } \} } _{ \{ I_{u } \} } }
\nn
\\
&=&
\int
\prod_{e\in W_{(R)} }\rd g_{e} \,
\underline{
\la \Psi_{\pp\Gamma_{(R)}, \{ \tl{I}_{u} \} } (\{ g_{e} \}_{e\in W_{(R)} })
\vert \chi_{\ell} ( G_{ W_{(R)} } ) \vert
\Psi_{\pp\Gamma_{(R)}, \{ I_{u } \} } (\{ g_{e} \}_{e\in W_{(R)} }) \ra
}
\,.
\eeq
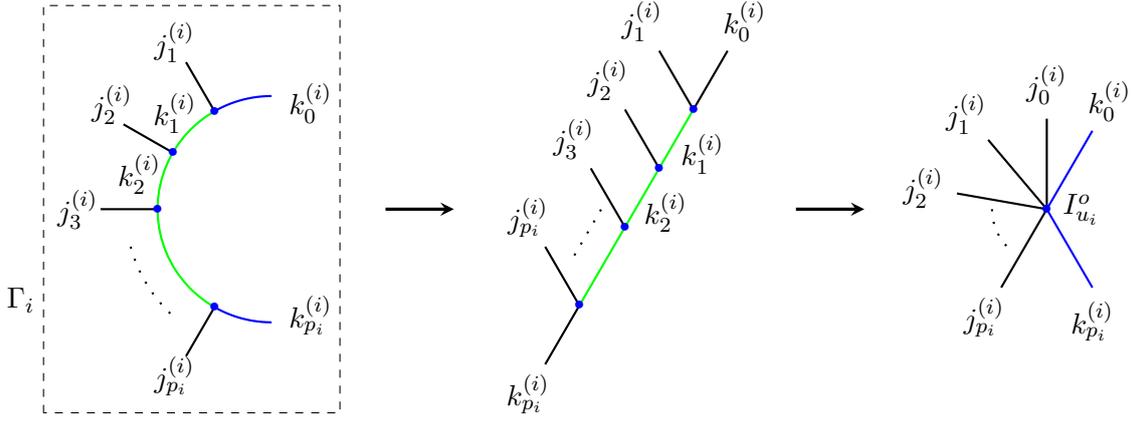
\begin{figure}
\centering
\begin{tikzpicture}[scale=0.6]

\coordinate (O) at (0,0);

\draw[dashed] (O) ++(-5,-4.5) rectangle (1.5,4.5);
\draw (O) ++ (-5.5,-2.5) node[above]{$\Gamma_{i}$};

\draw [thick,domain=90:120,blue] plot ({0+2.5 * cos(\x)}, {2.5 * sin(\x)});
\draw [thick,domain=240:270,blue] plot ({0+2.5 * cos(\x)}, {2.5 * sin(\x)});
\draw [thick,domain=120:240,green] plot ({0+2.5 * cos(\x)}, {2.5 * sin(\x)});

\draw[thick] (O) ++(120:2.5) node[scale=0.7,blue] {$\bullet$}--++ (120:1.25) ++(120:0.5) node {$j_{1}^{(i)}$};
\draw[thick] (O) ++(150:2.5) node[scale=0.7,blue] {$\bullet$}--++ (150:1.25) ++(120:0.5) node {$j_{2}^{(i)}$};
\draw[thick] (O) ++(180:2.5) node[scale=0.7,blue] {$\bullet$}--++ (180:1.25) ++(180:0.5) node {$j_{3}^{(i)}$};
\draw[thick] (O) ++(240:2.5) node[scale=0.7,blue] {$\bullet$} --++ (240:1.25) ++(240:0.5) node {$j_{p_{i}}^{(i)}$};

\draw [thick, loosely dotted,domain=195:230] plot ({0+3.2 * cos(\x)}, {3.2 * sin(\x)});

\draw (O) ++ (70:2.5) node {$k_{0}^{(i)}$};
\draw (O) ++ (290:2.5) node {$k_{p_{i} }^{(i)}$};
\draw (O) ++ (135:3) node {$k_{1}^{(i)}$};
\draw (O) ++ (165:3) node {$k_{2}^{(i)}$};

\draw[->,>=stealth,very thick] (2.5,0) -- (4,0);

\coordinate (P) at (10,3.5);

\path (P) ++ (240:1.5) coordinate (P1) ++ (240:1.5) coordinate (P2) ++ (240:1.5) coordinate (P3) ++ (240:2) coordinate (P4);

\draw[thick] (P1) --++ (120:1.5) ++(120:0.75) node {$j_{1}^{(i)}$};
\draw[thick] (P2) --++ (120:1.5) ++(120:0.75) node {$j_{2}^{(i)}$};
\draw[thick] (P3) --++ (120:1.5) ++(120:0.75) node {$j_{3}^{(i)}$};
\draw[thick] (P4) --++ (120:1.5) ++(120:0.75) node {$j_{p_{i}}^{(i)}$};

\draw[thick] (P) -- (P1);

\draw[thick] (P4) --++(240:1.5) ++ (240:0.75) node {$k_{p_{i}}^{(i)}$};

\path (P3) ++ (120:0.75) ++ (240:0.3) coordinate (O3);
\path (P4) ++ (120:0.75) ++ (60:0.3) coordinate (O4);
\draw [thick, loosely dotted] (O3) -- (O4);

\draw[thick,green] (P1) --node[near end,right=1,black] {$k_{1}^{(i)}$} (P2) --node[near end,right,black] {$k_{2}^{(i)}$} (P3) -- (P4);

\draw (P) ++ (60:0.75) node {$k_{0}^{(i)}$};
\draw (P1) node[scale=0.7,blue] {$\bullet$};
\draw (P2) node[scale=0.7,blue] {$\bullet$};
\draw (P3) node[scale=0.7,blue] {$\bullet$};
\draw (P4) node[scale=0.7,blue] {$\bullet$};

\draw[->,>=stealth,very thick] (11.5,0) -- (13,0);

\coordinate (Q) at (17,0);

\draw[thick] (Q) --++ (90:2) ++(90:0.75) node {$j_{0}^{(i)}$};
\draw[thick] (Q) --++ (130:2) ++(130:0.75) node {$j_{1}^{(i)}$};
\draw[thick] (Q) --++ (170:2) ++(170:0.75) node {$j_{2}^{(i)}$};
\draw[thick] (Q) --++ (240:2) ++(240:0.75) node {$j_{p_{i}}^{(i)}$};

\draw[thick,blue] (Q) --++ (60:2) ++(60:0.75) node[black] {$k_{0}^{(i)}$};
\draw[thick,blue] (Q) --++ (-60:2) ++(270:0.75) node[black] {$k_{p_{i}}^{(i)}$};

\draw [thick, loosely dotted,domain=180:230] plot ({17+1.2 * cos(\x)}, {0+1.2 * sin(\x)});

\draw (Q) node[scale=0.7,blue] {$\bullet$};
\draw (Q) ++ (0:0.75) node {$I_{u_{i}}^{o}$};

\end{tikzpicture}
%
%
\caption{The illustration for recoupling spins to acquire cluster intertwiner $I_{u_{i}}$.
}
\label{fig:ClusterIntertwiner}
\end{figure}
\end{res}
\begin{proof}
The transition matrix $Z$ eq.(\ref{eq:Amplitudes-6jsymbols}) can be computed from the type of graph $\Gamma=W \sqcup \cE^{\pp}$ where $\cE^{\pp}=\pp\Gamma$ is the set of $p$ boundary edges $\{ e^{\pp}_{1}, \cdots , e^{\pp}_{p} \}$, and $\Gamma^{o}=W$.
Let us consider the simplest case $n=2$. The proof is straightforward to generalize to arbitrary $n$.
Suppose bipartition $\Gamma=\Gamma_{1} \sqcup \Gamma_{2}$,
\be
\Gamma_{1}=T_{1} \sqcup \{ e^{\pp}_{1}, \cdots , e^{\pp}_{q} , e_{p}^{(1)}, e_{q}^{(1)} \} 
\,, \qquad
\Gamma_{2}=T_{2} \sqcup \{ e^{\pp}_{q+1}, \cdots , e^{\pp}_{p} , e_{p}^{(2)}, e_{q}^{(2)} \} \,. 
\ee
where $e_{p}^{(1)}$, $e_{q}^{(1)}$, $e_{p}^{(2)}$, $e_{q}^{(2)}$ are due to the partition that splits edges $e_{p}$ and $e_{q}$ with $e_{p}=e_{p}^{(1)} \sqcup e_{p}^{(2)}$ and $e_{q}=e_{q }^{(1)} \sqcup e_{q}^{(2)}$.
The bulks of $\Gamma_{1}$ and $\Gamma_{2}$ are $\Gamma_{1}^{o}=T_{1}[v_1 \overset{ e_{1} }{\to} \cdots \overset{ e_{q-2} }{\to} v_{q-1} \overset{ e_{q-1} }{\to} v_q]$ and $\Gamma_{2}^{o}=T_2[v_{q+1} \overset{ e_{q+1} }{\to} \cdots \overset{ e_{p-2} }{\to} v_{p-1} \overset{ e_{p-1} }{\to} v_p]$, respectively. In particular, $W_{1}=T_{1} \sqcup \{ e_{p}^{(1)}, e_{q}^{(1)} \}$ and $W_{2}=T_{2} \sqcup \{ e_{p}^{(2)}, e_{q}^{(2)} \}$ so $W=W_{1} \sqcup W_{2}$.
An example is illustrated by Fig.\ref{fig:ReducingGraph-OneLoop}.
We start by noticing
\beq
&&
\tensor{   {\Big[ Z(\Gamma)_{\tensor{\chi}{_{\ell} }\, \act_{W}   } \Big] }   } {^{ \{ J_{\alpha}, K_{\alpha} \} } _{ \{ j_{\alpha}, k_{\alpha} \} } }
\nn
\\
&=&
\int
\prod_{e\in W }\rd g_{e} \,
\underline{ \la
\Psi_{\pp\Gamma, \{J_{\alpha},K_{\alpha} \} } (\{ g_{e} \}_{e\in W })
\vert \chi_{\ell} ( G_{W} ) \vert
\Psi_{\pp\Gamma, \{j_{\alpha},k_{\alpha} \} } (\{ g_{e} \}_{e\in W }) \ra
}
\label{eq:TMonGraph-CGHolonomy}
\\
&=&
\int
\prod_{e\in W \setminus (T_{1} \sqcup T_{2}) }\rd h_{e} \,
\underline{
\la \Psi_{\pp\Gamma, \{J_{\alpha},K_{\alpha} \} } (\{h_{e}\}_{e\in W \setminus (T_{1} \sqcup T_{2}) }, \{ \id \}_{e\in T_{1} \sqcup T_{2} }) \vert \chi_{\ell} ( G_{W} )
}
\nn
\\
&&\phantom{ \int \prod_{e\in W \setminus (T_{1} \sqcup T_{2}) }\rd h_{e} \, }
\underline{
\vert
\Psi_{\pp\Gamma, \{j_{\alpha},k_{\alpha} \} } (\{h_{e}\}_{e\in W \setminus (T_{1} \sqcup T_{2}) }, \{ \id \}_{e\in T_{1} \sqcup T_{2} }) \ra
}
\,. \label{eq:TMonGraph-GaugeFixing-CGHolonomy}
\eeq
Let us break down this step: we start from eq.(\ref{eq:TransitionMatrix}). Recall eq.(\ref{eq:DBHSub-decomposition}), we can express the bulk-boundary map $| \Psi_{\pp\Gamma, \{j_{\alpha},k_{\alpha} \} } (\{ g_{e} \}_{e\in W }) \ra$ in terms of 'gluing operation' $\utimes$ for bulk-boundary maps $| \Psi_{\pp\Gamma_{1}, \{j_{\alpha}^{(1)}, k_{\alpha}^{(1)} \} } (\{ g_{e'} \}_{e'\in T_{1} }) \ra$ and $| \Psi_{\pp\Gamma_{2}, \{j_{\alpha}^{(2)}, k_{\alpha}^{(2)} \} } (\{ g_{e''} \}_{e''\in T_{2} }) \ra$. We then gauge-fix the two bulk-boundary maps such that the holonomies along $T_{1}$ and $T_{2}$ are gauge-fixed to $\id$ in line with eq.(\ref{eq:GaugeFixing-BoundaryStates}), i.e.,
\beq
&&| \Psi_{\pp\Gamma, \{j_{\alpha},k_{\alpha} \} } (\{ g_{e} \}_{e\in W }) \ra
\nn
\\
&=&
\left[
h_{ e_{p}^{(1)} } \otimes
h_{e_{q}^{(1)} } \otimes
\left(\bigotimes_{e'\in\pp\Gamma_{1} \setminus ( e_{p}^{(1)} \sqcup e_{q}^{(1)} )  } h_{v(e')}^{\eps_{e'}^{v}}\right)
| \Psi_{ \pp\Gamma_{1}, \{j_{\alpha}^{(1)},k_{\alpha}^{(1)} \} } (\{ \id \}_{e'\in T_{1} }) \ra
\right]
\nn
\\
&&
\bigutimes
\left[
h_{ e_{p}^{(2)} } \otimes
h_{e_{q}^{(2)} } \otimes
\left(\bigotimes_{e''\in\pp\Gamma_{2} \setminus ( e_{p}^{(2)} \sqcup e_{q}^{(2)} ) } h_{v(e'')}^{\eps_{e''}^{v}}\right)
| \Psi_{ \pp\Gamma_{2}, \{j_{\alpha}^{(2)},k_{\alpha}^{(2)} \} } (\{ \id \}_{e''\in T_{2} }) \ra
\right]
\label{eq:TMonGraph-GF-CGHolonomy}
\\
&=&
\left(\bigotimes_{e\in\pp\Gamma } h_{v(e)}^{\eps_{e}^{v}}\right)
\vert \Psi_{\pp\Gamma, \{j_{\alpha},k_{\alpha} \} } (\{h_{e}\}_{e\in W \setminus (T_{1} \sqcup T_{2}) }, \{ \id \}_{e\in T_{1} \sqcup T_{2} }) \ra
\,. \label{eq:TMonGraph-GFTrees-CGHolonomy}
\eeq
The advantage with the gauge-fixing is that the action of loop holonomy operator is only nontrivial along piecewise edges $e_{p}^{(1)}$, $e_{q}^{(1)}$, $e_{q}^{(2)}$, $e_{p}^{(2)}$ (interfacing edges between $\Gamma_{1}$ and $\Gamma_{2}$), and is trivial along other piecewise edges. Notably, gauge-fixing does not change the loop holonomy, i.e.,
\be
G_{W}=h_{e_{p}^{(2)} } \cdot h_{e_{q}^{(2)} } \cdot h_{e_{q}^{(1)} } \cdot h_{e_{p}^{(1)} }=G_{W_{(R)}}
\,.
\ee
We put eq.(\ref{eq:TMonGraph-GFTrees-CGHolonomy}) back to the eq.(\ref{eq:TMonGraph-CGHolonomy}), and note that the boundary holonomies are erased by the scalar product `$\underline{\la \, \ | \, \ \ra}$' for $\cH_{\pp\Gamma}$ due to $h^{\dagger} h=\id$, while the holonomies $h_{ e_{p}^{(1)} }$, $h_{ e_{p}^{(2)} }$, $h_{ e_{q}^{(1)} }$, $h_{ e_{q}^{(2)} }$ are not erased by the `$\underline{\la \, \ | \, \ \ra}$'.

To handle eq.(\ref{eq:TMonGraph-GF-CGHolonomy}), we follow the spirit of eq.(\ref{eq:Correspondence-GaugeFixed,LoopySPN}). We glue boundary edges along the $T_{1}$ and $T_{2}$. The resulting intertwiners allow for decomposition: 
\begin{equation}
\begin{aligned}
| \Psi_{ \pp\Gamma_{1}, \{j_{\alpha}^{(1)},k_{\alpha}^{(1)} \} } (\{ \id \}_{e'\in T_{1} }) \ra
=&
\sum_{I_{u_{1}}^{o} }
\tensor{   {\Big[ \, \cU_{\Upsilon_{1} }^{ \{ j_{\alpha}^{(1)}, k_{\alpha}^{\pp(1)} \} } \, \Big] }   } {^{ I_{u_{1}}^{o} }    }_{ \{ k_{\alpha}^{o(1)} \} }
| \Psi_{ \pp\Upsilon_{1}, I_{u_{1} } } \ra
\,,
\\
| \Psi_{ \pp\Gamma_{2}, \{j_{\alpha}^{(2)},k_{\alpha}^{(2)} \} } (\{ \id \}_{e''\in T_{2} }) \ra
=&
\sum_{I_{u_{2}}^{o} }
\tensor{   {\Big[ \, \cU_{\Upsilon_{2} }^{ \{ j_{\alpha}^{(2)}, k_{\alpha}^{\pp(2)} \} } \, \Big] }   } {^{ I_{u_{2}}^{o} }    }_{ \{ k_{\alpha}^{o(2)} \} }
| \Psi_{ \pp\Upsilon_{2}, I_{u_{2} } } \ra
\,.
\end{aligned} \label{eq:TMonGraph-GF-Gluing-CGHolonomy}
\end{equation}
Here the $\Upsilon_{1}$ and $\Upsilon_{2}$ are coarse-grained graphs of respect $\Gamma_{1}$ and $\Gamma_{2}$. With eq.(\ref{eq:TMonGraph-GF-Gluing-CGHolonomy}), the boundary state in eq.(\ref{eq:TMonGraph-GFTrees-CGHolonomy}) is rewritten as:
\beq
&&
| \Psi_{\pp\Gamma, \{j_{\alpha},k_{\alpha} \} } (\{h_{e}\}_{e\in W \setminus (T_{1} \sqcup T_{2}) }, \{ \id \}_{e\in T_{1} \sqcup T_{2} }) \ra
\nn
\\
&=&
\left[
h_{ e_{p}^{(1)} } \otimes
h_{e_{q}^{(1)} } \otimes
\sum_{I_{u_{1}}^{o} }
\tensor{   {\Big[ \, \cU_{\Upsilon_{1} }^{ \{ j_{\alpha}^{(1)}, k_{\alpha}^{\pp(1)} \} } \, \Big] }   } {^{ I_{u_{1}}^{o} }    }_{ \{ k_{\alpha}^{o(1)} \} }
| \Psi_{ \pp\Upsilon_{1}, I_{u_{1} } } \ra
\right]
\nn
\\
&&
\bigutimes
\left[
h_{ e_{p}^{(2)} } \otimes
h_{e_{q}^{(2)} } \otimes
\sum_{I_{u_{2}}^{o} }
\tensor{   {\Big[ \, \cU_{\Upsilon_{2} }^{ \{ j_{\alpha}^{(2)}, k_{\alpha}^{\pp(2)} \} } \, \Big] }   } {^{ I_{u_{2}}^{o} }    }_{ \{ k_{\alpha}^{o(2)} \} }
| \Psi_{ \pp\Upsilon_{2}, I_{u_{2} } } \ra
\right]
\,. \label{eq:GaugeFixing-CGHolonomy}
\eeq
We put eq.(\ref{eq:GaugeFixing-CGHolonomy}) back to eq.(\ref{eq:TMonGraph-GaugeFixing-CGHolonomy}), obtaining the transition matrix along $W_{(R)}$,
\beq
&&
\int
\bigg[
( h_{ e_{p}^{(1)} } \otimes
h_{e_{q}^{(1)} } ) \act
| \Psi_{ \pp\Upsilon_{1}, \tl{I}_{u_{1} } } \ra
\bigutimes \,
( h_{ e_{p}^{(2)} } \otimes
h_{e_{q}^{(2)} } ) \act
| \Psi_{ \pp\Upsilon_{2}, \tl{I}_{u_{2} } } \ra
\bigg]^{\dagger}
\chi_{\ell} ( h_{ e_{p}^{(2)} } \cdot h_{ e_{q}^{(2)} }  \cdot h_{ e_{q}^{(1)} }  \cdot h_{ e_{p}^{(1)} } )
\nn
\\
&&\quad
\bigg[
( h_{ e_{p}^{(1)} } \otimes
h_{e_{q}^{(1)} } ) \act
| \Psi_{ \pp\Upsilon_{1}, I_{u_{1} } } \ra
\bigutimes \,
( h_{ e_{p}^{(2)} } \otimes
h_{e_{q}^{(2)} } ) \act
| \Psi_{ \pp\Upsilon_{2}, I_{u_{2} } } \ra
\bigg]
\, \rd h_{ e_{p}^{(1)} } \rd h_{e_{q}^{(1)} } \rd h_{ e_{p}^{(2)} } \rd h_{e_{q}^{(2)} }
\nn \\
&=&
\int
\prod_{e\in W_{(R)} }\rd g_{e} \,
\underline{
\la \Psi_{\pp\Gamma_{(R)}, \{ \tl{I}_{u} \} } (\{ g_{e} \}_{e\in W_{(R)} })
\vert \chi_{\ell} ( G_{ W_{(R)} } ) \vert
\Psi_{\pp\Gamma_{(R)}, \{ I_{u } \} } (\{ g_{e} \}_{e\in W_{(R)} }) \ra
}
\nn \\
&=&
\tensor{   {\Big[ \tensor{ Z( \Gamma_{(R)} ) }{ _{\tensor{\chi}{_{\ell} }\, \act_{W_{(R)} }    }    } \Big] }   } {^{ \{ \tl{I}_{u } \} } _{ \{ I_{u } \} } }
\, \prod_{\alpha=1}^{p} \delta_{ J_{\alpha} j_{\alpha} }
\,. \label{eq:TransitionMatrix-CGGraph}
\eeq
Here $I_{u_{i}}=\{ j_{\alpha }^{(i)}, k_{\alpha}^{\pp(i)} \} \cup I_{u_{i}}^{o}$ and $\tl{I}_{u_{i}}=\{ J_{\alpha }^{(i)}, K_{\alpha}^{\pp(i)} \} \cup \tl{I}_{u_{i}}^{o} $.
The $\delta_{J_{\alpha} j_{\alpha} }$ is imposed by scalar product `$\underline{\la \, \ | \, \ \ra}$'.
So eq.(\ref{eq:TransitionMatrix-CGGraph}) actually represents the transition matrix on the coarse-grained graph $\Gamma_{(R)}$. Therefore, eq.(\ref{eq:TMonGraph-GaugeFixing-CGHolonomy}) leads to the transformation for particular case $n=2$,
\beq
&&
\tensor{   {\Big[ Z(\Gamma)_{\tensor{\chi}{_{\ell} }\, \act_{W}   } \Big] }   } {^{ \{ J_{\alpha}, K_{\alpha} \} } _{ \{ j_{\alpha}, k_{\alpha} \} } }
\nn
\\
&=&
\sum_{ \{ I_{u_{i}}^{(o)} \} } \sum_{ \{ \tl{I}_{u_{i}}^{(o)} \} }
\overline{
\tensor{   {\Big[ \, \cU_{\Upsilon_{1} }^{ \{ J_{\alpha}^{(1)}, K_{\alpha}^{\pp(1)} \} } \, \Big] }   } {^{ \tl{I}_{u_{1}}^{o} }    }_{ \{ K_{\alpha}^{o(1)} \} }
}
\overline{
\tensor{   {\Big[ \, \cU_{\Upsilon_{2} }^{ \{ J_{\alpha}^{(2)}, K_{\alpha}^{\pp(2)} \} } \, \Big] }   } {^{ \tl{I}_{u_{2}}^{o} }    }_{ \{ K_{\alpha}^{o(2)} \} }
}
\tensor{   {\Big[ \tensor{ Z( \Gamma_{(R)} ) }{ _{\tensor{\chi}{_{\ell} }\, \act_{W_{(R)} }    }    } \Big] }   } {^{ \{ \tl{I}_{u_{i}} \} } _{ \{ I_{u_{i}} \} } }
\nn
\\
&&
\phantom{ \sum_{ \{ I_{u_{i}}^{(o)} \} } \sum_{ \{ \tl{I}_{u_{i}}^{(o)} \} } }
\tensor{   {\Big[ \, \cU_{\Upsilon_{1} }^{ \{ j_{\alpha}^{(1)}, k_{\alpha}^{\pp(1)} \} } \, \Big] }   } {^{ I_{u_{1}}^{o} }    }_{ \{ k_{\alpha}^{o(1)} \} }
\tensor{   {\Big[ \, \cU_{\Upsilon_{2} }^{ \{ j_{\alpha}^{(2)}, k_{\alpha}^{\pp(2)} \} } \, \Big] }   } {^{ I_{u_{2}}^{o} }    }_{ \{ k_{\alpha}^{o(2)} \} }
\, \prod_{\alpha=1}^{p} \delta_{ J_{\alpha} j_{\alpha} }
\,. \label{eq:LoopHolonomyOperator-Transformation-GammaLoopy-Bipartition}
\eeq
\end{proof}

Fig.\ref{fig:LoopHolonomyOperator-ChannelSwitch} is a snapshot for the proof. The unitary $\cU_{\cB}$ for dual boundary Hilbert space is implemented by first gauge-fixing the holonomy along $e_{\alpha+1}$ such that the holonomy operator acts trivially on $e_{\alpha+1}$, then gluing two bouquet edges and switching channel, thus the $e_{\alpha+1}$ is coarse-grained during the action of holonomy operator.

\smallskip

A directly corollary of the Result \ref{Prop:CG-LoopHolonomyOperator}, is that we can compute expectation $\la \wh{ \chi_{\ell} }\, \act_{W} \ra$ from loopy spin network \cite{Chen:2022rty}.

\smallskip

Now that we can see, the exponential evolution that is generated by loop holonomy operator $\wh{\chi_{\ell} } \, \act_{W}$ on $\Gamma$, is related to the exponential evolution that is generated by loop holonomy operator $\wh{\chi_{\ell} } \, \act_{W(R)}$ on $\Gamma_{(R)}$ via following transformation:
\beq
\exp\left [ - \ri t \, \tensor{Z(\Gamma) }{ _{\tensor{\chi}{_{\ell} }\, \act_{W}    }   } \right]
=
\cU_{\Upsilon_{1} }^{\dagger} \cdots \cU_{\Upsilon_{n} }^{\dagger}
\exp\left [ - \ri t \, \tensor{Z(\Gamma_{(R)}) }{ _{\tensor{\chi}{_{\ell} }\, \act_{W_{(R)} }    }   } \right]
\cU_{\Upsilon_{1} } \cdots \cU_{\Upsilon_{n} }
\,.
\eeq
Every $\cU_{\Upsilon_{i} }$ is a unitary for internal space of intertwiner at $u_{i}$, thus they do not affect the spin-matching constraints between interfacing edges. Instead, they are to be interpreted as local unitaries at $u_{i}$, i.e., these unitaries do not affect spin entanglement. Indeed, following transformation eq.(\ref{eq:DensityMatrix-Transformation}), the density matrix $\rho_{\pp\Gamma^{*}}$'s evolution is now given by
\beq
\rho_{\pp\Gamma^{*}}(t)
&=&
\cU_{\cB_{1} }^{\dagger} \cdots \cU_{\cB_{n} }^{\dagger}
\, \rho_{\pp\Gamma_{(R)}^{*} }(t) \,
\cU_{\cB_{1} } \cdots \cU_{\cB_{n} }
\,,
\\
\rho_{\pp\Gamma_{(R)}^{*} }(t)
&=&
\exp
\left [ - \ri (t-t_{0}) \, \tensor{Z(\Gamma) }{ _{\tensor{\chi}{_{\ell} }\, \act_{W_{(R)} }    }   } \right]
\, \rho_{\pp\Gamma_{(R)}^{*} }(t_0) \,
\exp
\left [ \ri (t-t_{0}) \, \tensor{Z(\Gamma) }{ _{\tensor{\chi}{_{\ell} }\, \act_{W_{(R)} }    }   } \right]
\,.
\eeq
Partial tracing $\rho_{\Gamma}$ over sub-networks $\Gamma \setminus \Gamma_{i}$ is equivalent to partial tracing $\rho_{\pp\Gamma^{*}}$ over $(\cH_{ \cB_{i}^{c} })^{*}$ where $\cB_{i}^{c} \equiv \pp(\Gamma \setminus \Gamma_{i})$. It leads to the reduced density matrix for  $(\cH_{\cB_{i} } )^{*}$,
\beq \label{eq:DensityMatrix-Transformation-inProp}
\rho_{\pp\Gamma_{i}^{*} }(t)
=
\tr_{ (\cH_{ \cB_{i}^{c} })^{*} } \Big[ \rho_{\pp\Gamma^{*}}(t) \Big]
=
\cU_{\cB_{i} }^{\dagger}
\, \rho_{\pp\Upsilon_{i}^{*} }(t) \,
\cU_{\cB_{i} }
 \,.
\eeq
Therefore, the reduced density matrices $\rho_{\pp\Gamma_{i}^{*} }(t)$ and $\rho_{\pp\Upsilon_{i}^{*} }(t)$ are equivalent up to a unitary for dual boundary Hilbert space $(\cH_{ \cB_{i} })^{*}$.

\begin{res} \label{prop:IntertwinerEntanglement-ReducedGraph-Graph}
Given a partition $\Gamma=\bigsqcup_{i=1}^{n}\Gamma_{i}$, and an oriented loop $W$ on $\Gamma$. Let the coarse-grained graph for $\Gamma$ be $\Gamma_{(R)}=\bigsqcup_{i=1}^{n} \Upsilon_{i}$ where $\Upsilon_{i}$ are loopy graphs for respect $\Gamma_{i}$, and the coarse-grained loop $W_{(R)}$ for $W$ on $\Gamma_{(R)}$.
Let $\cE$ be any entanglement measure. Consider evolutions generated by loop holonomy operators $\wh{\chi_{\ell} } \, \act_{W}$ and the corresponding $\wh{\chi_{\ell} } \, \act_{W_{(R)}}$.
Then the dynamics of spin network entanglement between $\Gamma_{i}$ is identical to the dynamics of intertwiner entanglement between $\Upsilon_{i}$, i.e.,
\be
\cE[ \rho_{\Gamma}(t) ]=\cE[ \rho_{\Gamma_{(R)} }(t) ]
\,.
\ee
\end{res}
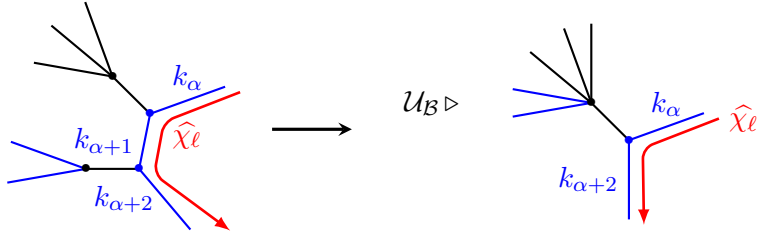
\begin{figure}
\centering
\begin{tikzpicture}[scale=0.7]

\coordinate (O) at (0,0.5);

\coordinate (P) at (-0.5,-1.25);

\draw[thick] (O) -- ++(170:1.5);
\draw[thick] (O) -- ++(110:1.5);
\draw[thick] (O) -- ++(140:1.5);

\draw[thick,blue] (P) -- ++(160:1.5);
\draw[thick,blue] (P) -- ++(190:1.5);

\draw[thick] (O) -- ++(-45:1) coordinate (O1);

\draw[thick] (P) -- ++(0:1) coordinate (P1);

\path (O1) ++ (-30:0.3) coordinate (O4);

\path (P1) ++ (-20:0.3) coordinate (P2);
\draw (P2) ++(-40:1.5) ++ (-20:0.3) coordinate (P3);

\draw[thick,blue] (P1) -- ++(-50:1.5) node[midway,left] {$k_{\alpha+2}$};

\draw[thick,blue] (P1) -- node[left,midway] {$k_{\alpha+1}$} (O1);

\draw[thick,blue] (O1) -- node[midway,above] {$k_{\alpha}$} ++(20:1.5) ++ (-20:0.3)coordinate (O3);

\draw [line width=1pt,red,->,>=latex,rounded corners] (O3) -- (O4) -- (P2) -- (P3);

\draw (O) node[scale=0.7] {$\bullet$};
\draw (O1) node[scale=0.7,blue] {$\bullet$};

\draw (P) node[scale=0.7] {$\bullet$};
\draw (P1) node[scale=0.7,blue] {$\bullet$};

\draw (O4) ++ (-30:0.5) node[red] {$\widehat{\chi_{\ell} }$};

\draw[->,>=stealth,very thick] (O) ++(3,-1) -- ++ (1.5,0);

\coordinate (A) at (9,0);

\draw (A)  ++ (-3,0) node {$\cU_{\cB} \, \act$};

\draw[thick] (A) -- ++(90:1.5);
\draw[thick] (A) -- ++(110:1.5);
\draw[thick] (A) -- ++(140:1.5);

\draw[thick,blue] (A) -- ++(170:1.5);
\draw[thick,blue] (A) -- ++(190:1.5);

\draw[thick] (A) -- ++(-45:1) coordinate (A1);

\path (A1) ++ (-30:0.3) coordinate (A4);

\draw[thick,blue] (A1) -- node[midway,left] {$k_{\alpha+2}$} ++(-90:1.5) ++ (-20:0.3) coordinate (A2);
\draw[thick,blue] (A1) -- node[midway,above] {$k_{\alpha}$} ++(20:1.5) ++ (-20:0.3)coordinate (A3);

\draw [line width=1pt,red,->,>=latex,rounded corners] (A3) -- (A4) -- (A2);

\draw (A) node[scale=0.7] {$\bullet$};
\draw (A1) node[scale=0.7,blue] {$\bullet$};

\draw (A3) node[right,red] {$\widehat{\chi_{\ell} }$};

\end{tikzpicture}

\caption{The illustration for unitary transformation on loop holonomy operator: in the left side one is able to gauge fix the holonomy associated with spin $k_{\alpha+1}$ into $\id$, then implement gluing and channel transformation to coarse-grain $k_{\alpha+1}$.}
\label{fig:LoopHolonomyOperator-ChannelSwitch}
\end{figure}

At the end of the day, we have shown that spin network entanglement allows to be coarse-grained at dynamical level at least for the evolution generated by loop holonomy operator, thus one can study the dynamics of entanglement from coarse-grained graph.


\section{Examples} \label{Section:Examples}
We would like to conclude this paper with explicit examples of spin network entanglement coarse-graining, looking at entanglement excitation by loop holonomy operator.


\subsection{Triangle graph}
We consider the holonomy operator acting on the loop of triangle graph Fig.\ref{fig:triGraph-canGraph}. We compute explicitly the bipartite entanglement between $\cH_{A}$ and $\cH_{B}\otimes\cH_{C}$. Then we show that the bipartite entanglement can be studied from coarse-grained graph.
\begin{figure}[htb]
\centering
\begin{tikzpicture}[scale=0.7]

\coordinate (O) at (3,0);
\coordinate (P) at (4.5,0);
\coordinate (A) at (-3.5,0);

\draw[thick] (A) ++(60:1.2) coordinate(B3) node[scale=0.7,blue] {$\bullet$} node[right]{$C$} --++ (60:1) node[right] {$j_3$};
\draw[thick] (A) ++(180:1.2) coordinate(B1) node[scale=0.7,blue] {$\bullet$} node[above]{$A$}--++ (180:1) node[above] {$j_1$};
\draw[thick] (A) ++(300:1.2) coordinate(B2) node[right]{$B$} --++ (300:1) node[right] {$j_2$};

\draw[blue,thick] (B1) -- node[below] {$k_1$} (B2) node[scale=0.7,blue] {$\bullet$} -- node[right] {$k_2$} (B3) node[scale=0.7,blue] {$\bullet$} --node[above=2,midway] {$k_3$} (B1) node[scale=0.7,blue] {$\bullet$};

\draw[->,>=stealth,very thick] (-1.5,0) -- (0,0);

\draw (1,0) node {$\sum_{j_{23}}$};

\coordinate (O1) at (3,0);
\coordinate (O2) at (5,0);

\draw[thick] (O1) -- ++ (180:1) node[above] {$j_1$};

\draw[thick] (O2) -- ++ (1,0) node[above,midway] {$j_{23}$} -- ++ (45:1) node[above]{$j_3$};
\draw[thick] (O2) ++(1,0) node[scale=0.7] {$\bullet$} -- ++ (-45:1) node[below] {$j_2$};

\draw[blue,thick,in=115,out=65,rotate=0] (O1) to node[above,midway] {$k_3$} (O2) node[scale=0.7] {$\bullet$} to [out=245,in=-65] node[below] {$k_1$} (O1) node[scale=0.7] {$\bullet$};

\draw[->,>=stealth,very thick] (7.25,0) -- (8.75,0);

\coordinate (A1) at (10.5,0);
\coordinate (A2) at (12.5,0);

\draw[thick] (A1) -- ++ (180:1) node[above] {$j_1$};

\draw[thick] (A2) -- ++ (45:1) node[above]{$j_3$};
\draw[thick] (A2) -- ++ (-45:1) node[below] {$j_2$};

\draw[blue,thick,in=115,out=65,rotate=0] (A1) to node[midway,above] {$k_3$} (A2) node[scale=0.7] {$\bullet$} to[out=245,in=-65] node[midway,below] {$k_1$} (A1) node[scale=0.7] {$\bullet$};

\draw (A1) ++ (-90:0.5) node[scale=0.8] {$A$};
\draw (A2) ++ (-90:0.5) node[scale=0.8] {$B'$};

\end{tikzpicture}
\caption{Coarse-graining triangle graph to candy graph.}
\label{fig:triGraph-canGraph}
\end{figure}
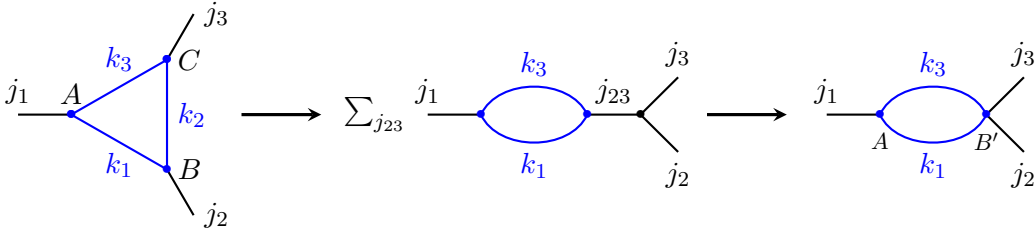

Spin network states can involve superpositions of the bulk spins $K_1$, $K_2$, $K_3$ while keeping boundary spins $j_1$, $j_2$, $j_3$. Such bulk spin superposition naturally induces a superposition of intertwiners. If this superposition carries correlations between the three vertices, this will be reflected in the entanglement between the four vertices.

Starting with an initial spin network basis state $| \Psi_{tri, \{j_i,k_i\} } \ra=| j_1,k_1,k_3 \ra_{A} \otimes | j_2,k_2,k_1 \ra_{B} \otimes | j_3, k_3,k_2\ra_{C}$, we consider the evolution generated by the loop holonomy operator $\wh{\chi_{\ell} }$,
\be
\wh{\chi_{\ell} }
:\bigotimes_{i=1}^{3}
\textrm{Inv}_{\SU(2)} \Big( \cV_{j_i} \otimes \cV_{k_{i} } \otimes \cV_{k_{i-1}} \Big)
\to
\bigoplus_{ \{K_i\} }
\bigotimes_{i=1}^{3}
\textrm{Inv}_{\SU(2)} \Big( \cV_{j_i} \otimes \cV_{K_{i}} \otimes \cV_{K_{i-1} } \Big)
\,.
\nn
\ee
For  infinitesimal time $t \to 0$, the unitarity evolution operator is $e^{ - \ri t \, \wh{\chi_{\ell} } } = \id - \ri  t \, \wh{\chi_{\ell} } + \cO(t^2)$, which acts as
 \beq
e^{ - \ri t \, \widehat{\chi_{\ell} } } | \Psi_{tri,\{j_i,k_i\}} \ra
&=&
e^{ - \ri t \, \widehat{\chi_{\ell} } } | j_1,k_1,k_3 \ra_{A} \otimes | j_2,k_2,k_1 \ra_{B} \otimes | j_3, k_3,k_2\ra_{C}
\nn \\
&=&
 \sum_{K_i=| k_i-{\ell} |}^{k_i+{\ell} }
  \left( \delta^{K_1}_{k_1}\delta^{K_2}_{k_2}\delta^{K_3}_{k_3}  - \ri t
  {[} Z(tri)_{\ell} {]}^{ \{j_{i}, K_{i} \} }_{ \phantom{ \{j_{i}, K_{i} \} } \{j_{i}, k_{i} \} } \, \right)
 | j_1,K_1,K_3 \ra_{A}
\nn \\
&&
\phantom{  \sum_{K_i=| k_i-{\ell} |}^{k_i+{\ell} } }
 \otimes | j_2,K_2,K_1 \ra_{B} \otimes | j_3, K_3,K_2\ra_{C}
 + \cO(t^2)
\,, \label{eq:ShortTime-BasisState-triGraph}
\eeq
where $| j_1,K_1,K_3 \ra_{A} \in \cH_{A}$, $| j_2,K_2,K_1 \ra_{B} \in \cH_{B}$, $| j_3,K_3,K_2 \ra_{C} \in \cH_{C}$, denote the intertwiners living at respect trivalent vertex.
According to the general formula (\ref{eq:Amplitudes-6jsymbols}), the transition matrix $[Z(tri)_{\ell}]$ is expressed in terms of $6j$-symbols for the action of loop holonomy operator $\wh{ \chi_{\ell} }$
\beq
{[} Z(tri)_{\ell} {]}^{ \{j_{i}, K_{i} \} }_{ \phantom{ \{j_{i}, K_{i} \} } \{j_{i}, k_{i} \} }
&=&
(-1)^{\sum_{i=1}^{3}( j_i+k_i+K_i + \ell)} \begin{Bmatrix}
   j_1 & k_1 & k_3 \\
   \ell &  K_3 & K_1
  \end{Bmatrix}\begin{Bmatrix}
   j_2 & k_2 & k_1 \\
   \ell &  K_1 & K_2
  \end{Bmatrix} \begin{Bmatrix}
   j_3 & k_3 & k_2 \\
   \ell & K_2 & K_3
  \end{Bmatrix}
  \nn
  \\
  &&\times
  \prod_{i=1}^{3} \sqrt{ (2k_i+1)(2K_i+1) }
  \,.
\eeq
The matrix elements are all real numbers. We take special care in properly normalizing the truncated state,
\be
| \Psi_{tri,\{j_i,k_i\}} (\ell,t) \ra
=
\f{ | \Psi_{tri,\{j_i,k_i\}} \ra - \ri t
\sum_{ \{K_i\} }
{[} Z(tri)_{\ell} {]}^{ \{j_{i}, K_{i} \} }_{ \phantom{ \{j_{i}, K_{i} \} } \{j_{i}, k_{i} \} }
| \Psi_{tri,\{j_i,K_i\}} \ra }
{ \sqrt{ 1+t^2 \sum_{s=0}^{2\ell} {[} Z(tri)_{s} {]}^{ \{j_{i}, k_{i} \} }_{ \phantom{ \{j_{i}, k_{i} \} } \{j_{i}, k_{i} \} } }   }
\,. \label{eq:TruncatedState-triGraph}
\ee
The normalization factor, at the denominator, can be computed explicitly using the composition rule of (\ref{eq:Composition-Amplitudes}).
The normalization factor can then be evaluated to
\begin{align}
&N_{tri, \{j_{i},k_{i}\} }(\ell,t)
\nn \\
=&
1+t^2 \sum_{s=0}^{2\ell} (-1)^{\sum_{i=1}^{3}( j_i+2k_i + s)} \begin{Bmatrix}
   j_1 & k_1 & k_3 \\
   s &  k_3 & k_1
  \end{Bmatrix}\begin{Bmatrix}
   j_2 & k_2 & k_1 \\
   s &  k_1 & k_2
  \end{Bmatrix} \begin{Bmatrix}
   j_3 & k_3 & k_2 \\
   s & k_2 & k_3
  \end{Bmatrix}
  \prod_{i=1}^{3} (2k_i+1)
  \\
  =&
  1+t^2 \sum_{\{K_{i}\}} \begin{Bmatrix}
   j_1 & k_1 & k_3 \\
   \ell &  K_3 & K_1
  \end{Bmatrix}^2 \begin{Bmatrix}
   j_2 & k_2 & k_1 \\
   \ell &  K_1 & K_2
  \end{Bmatrix}^2 \begin{Bmatrix}
   j_3 & k_3 & k_2 \\
   \ell & K_2 & K_3
  \end{Bmatrix}^2 \prod_{i=1}^{3} (2k_i+1)(2K_i+1)
  \,.
\end{align}

\smallskip

We now look at the entanglement from the truncated state eq.(\ref{eq:TruncatedState-triGraph}). The initial spin network state is a basis state, thus is fully separable with vanishing entanglement. The final state is given by the density matrix $\rho_{tri}(t)=| \Psi_{tri,\{j_i,k_i\}} (\ell,t) \ra\la \Psi_{tri,\{j_i,k_i\}} (\ell,t) |$. Consider bipartition $\cH_A$ and $\cH_B \otimes \cH_C$. The reduced density matrix $\rho_{tri_{A} }(t) \in \textrm{End} (\cH_{A})$ is obtained via partial tracing over $\cH_{B} \otimes \cH_{C}$ with orthonormal basis $ | j_2,K_2,K_1 \ra_{B} \otimes | j_3,K_3,K_2 \ra_{C} \equiv | \{j_2,j_3,K_1,K_2,K_3\} \ra_{BC}$:
\beq
\rho_{tri_{A} }(t)
&=&
\sum_{K_1,K_2,K_3} \la \{j_2,j_3,K_1,K_2,K_3\} | \rho_{tri}(t) | \{j_2,j_3,K_1,K_2,K_3\} \ra_{BC}
\nn \\
&=& \bigg[ t^2 \sum_{K_1, K_2,K_3}
\Big( {[} Z(tri)_{\ell} {]}^{ \{j_{i}, K_{i} \} }_{ \phantom{ \{j_{i}, K_{i} \} } \{j_{i}, k_{i} \} } \Big)^2
| j_1,K_1,K_3 \ra \la j_1,K_1,K_3 |_{A}
\nn \\
&&
\phantom{ \bigg( }
+ | j_1,k_1,k_3 \ra \la j_1,k_1,k_3 |_{A}
\bigg] \times \f{1}{ N_{ tri, \{j_{i},k_{i}\} }(\ell,t) }
\,.
\eeq
The eigenvalues of $\rho_{tri_{A} }(t)$ can be read off directly from this formula since the reduced density matrix is diagonal in the $\big\{| j_1,K_1,K_3 \ra \big\}$ basis for $\cH_{A}$,
\beq
&&\lambda_{ \rho_{tri_A } }[K_1,K_3]
\nn
\\
&=&
\f{\delta^{K_1}_{k_1} \, \delta^{K_3}_{k_3}
+ t^2 \sum_{K_2}
\begin{Bmatrix}
   j_1 & k_1 & k_3 \\
   \ell &  K_3 & K_1
  \end{Bmatrix}^2
  \begin{Bmatrix}
   j_2 & k_2 & k_1 \\
   \ell &  K_1 & K_2
  \end{Bmatrix}^2
  \begin{Bmatrix}
   j_3 & k_3 & k_2 \\
   \ell & K_2 & K_3
  \end{Bmatrix}^2 \, \prod_{i=1}^{3} (2K_i+1)(2k_i+1)
  }{
  1+t^2 \sum_{s=0}^{2\ell} (-1)^{\sum_{i=1}^{3}( j_i+2k_i + s)} \begin{Bmatrix}
   j_1 & k_1 & k_3 \\
   s &  k_3 & k_1
  \end{Bmatrix}\begin{Bmatrix}
   j_2 & k_2 & k_1 \\
   s &  k_1 & k_2
  \end{Bmatrix} \begin{Bmatrix}
   j_3 & k_3 & k_2 \\
   s & k_2 & k_3
  \end{Bmatrix}
  \prod_{i=1}^{3} (2k_i+1)
  }
\,. \nn
\\
\label{eq:eigenvalues-Truncated-triGraph}
\eeq

\smallskip

Now we look at the entanglement excitation on the coarse-grained graph. For bipartition $\cH_A$ and $\cH_B \otimes \cH_C$, we gauge fix $g_2 \to \id$, then contract vertices $B,C$ along $e_{2}$, acquiring a 4-valent vertices with spins $j_2,j_3,k_1,k_3$ as Fig.\ref{fig:triGraph-canGraph}. The spin network coarse-grained state from the triangle graph to the candy graph is given by
\be
| \psi_{ tc, \{ j_i,k_i\} } \ra
=
| j_1,k_1,k_3 \ra_{A}
\otimes
\sum_{j_{23} }
\cU_{ \{j,k\}_{B'}}^{k_2,j_{23} } | \{ j_2, j_3, k_3, k_1 \}, j_{23} \ra_{B'}
\,.
\label{eq:triGraph-canGraph}
\ee
As initial state based on triangle graph, the initial state $| \psi_{ tc, \{ j_i,k_i\} } \ra$ based on the coarse-grained graph is a product state. The intertwiner at $B'$ is a superposition of basis states $| \{ j_2, j_3, k_3, k_1 \}, j_{23} \ra_{B'}$ with respect to internal spin $j_{23}$. The unitary matrix
\be
\cU_{ \{j,k\}_{B'}}^{k_2,j_{23} } 
=
(-1)^{2k_2+k_3-k_1+j_2+j_3} \sqrt{ (2k_2+1)(2j_{23}+1) }
\begin{Bmatrix}
j_2 & j_3 & j_{23} \\
k_3 & k_1 & k_2 
\end{Bmatrix}
\ee
is account for the channel transformation, and the matrix elements are all real numbers, and the unitarity is verified by the orthogonality of $6j$-symbols,
\be
\sum_{j_{23} } \cU_{ \{j,k\}_{B'}}^{k_2,j_{23} }  \cU_{ \{j,k\}_{B'}}^{\tk_2,j_{23} } 
=
\sum_{j_{23} } \sqrt{ (2k_2+1) (2\tk_2+1) }(2j_{23}+1)
\begin{Bmatrix}
j_2 & j_3 & j_{23} \\
k_3 & k_1 & k_2 
\end{Bmatrix}
\begin{Bmatrix}
j_2 & j_3 & j_{23} \\
k_3 & k_1 & \tk_2 
\end{Bmatrix}
=\delta_{k_2 \tk_2}
\,.
\ee
The unitarity requires the two $\cU$ having same $\{ j , k \}$ labeling since they are external labels for the intertwiner. On the other hand, the $k_2,j_{23}$ are internal labels for the intertwiner, so we also have
\be \label{eq:unitary-tc}
\sum_{k_2 } \cU_{ \{j,k\}_{B'}}^{k_2,j_{23} }  \cU_{ \{j,k\}_{B'}}^{k_2,\tl{j}_{23} } 
=\delta_{j_{23} \tl{j}_{23} }
\,.
\ee
The unitary $\cU_{B'}$ links the transition matrices $Z(tri)$ and $Z(tc)$,
\be
{[} Z(tri)_{\ell} {]}^{ \{j_{i}, K_{i} \} }_{ \phantom{ \{j_{i}, K_{i} \} } \{j_{i}, k_{i} \} }
=
\sum_{  j_{23}, J_{23} }
\cU_{ \{j,K\}_{B'}}^{K_2,J_{23} } 
[ Z(tc)_{\ell} ]^{ j_1,J_{23}, K_1,K_3 }_{ \phantom{ J_{1}, J_{23}, K_1,K_3 } j_{1}, j_{23}, k_1, k_3 }
\cU_{ \{j,k\}_{B'}}^{k_2,j_{23} } 
  \,, \label{eq:Amplitudes-tri,tc}
\ee
where the transition matrix for $Z(tc)$ is given by
\begin{align}
[ Z(tc)_{\ell} ]^{ J_1,J_{23}, K_1,K_3 }_{ \phantom{ J_{1}, J_{23}, K_1,K_3 } j_{1}, j_{23}, k_1, k_3 }
=&
(-1)^{j_{1}+j_{23}+K_1+k_1+K_3+k_3+2\ell}
\begin{Bmatrix}
   j_{1} & k_1 & k_3 \\
   \ell &  K_3 & K_1
  \end{Bmatrix}
\begin{Bmatrix}
   j_{23} & k_1 & k_3 \\
   \ell &  K_3 & K_1
  \end{Bmatrix}
\nn  \\
  & \times
  \sqrt{(2k_1+1)(2K_1+1)(2k_3+1)(2K_3+1)} \delta_{j_1 J_1} \delta_{j_{23} J_{23}}
  \,. \label{eq:Amplitudes-tcGraph}
\end{align}
Explicitly, the transformation is expressed in terms of $6j$-symbols:
\begin{align}
&
(-1)^{\sum_{i=1}^{3}( j_i+k_i+K_i + \ell)} \begin{Bmatrix}
   j_1 & k_1 & k_3 \\
   \ell &  K_3 & K_1
  \end{Bmatrix}\begin{Bmatrix}
   j_2 & k_2 & k_1 \\
   \ell &  K_1 & K_2
  \end{Bmatrix} \begin{Bmatrix}
   j_3 & k_3 & k_2 \\
   \ell & K_2 & K_3
  \end{Bmatrix}
\nn
\\
=&
\sum_{  j_{23} }
(-1)^{2\ell + j_{1}+j_{23} + 2K_1 + 2k_3 +2k_2-2K_2 } (2j_{23}+1)
\begin{Bmatrix}
j_2 & j_3 & j_{23} \\
K_3 & K_1 & K_2 
\end{Bmatrix}
\begin{Bmatrix}
   j_{1} & k_1 & k_3 \\
   \ell &  K_3 & K_1
  \end{Bmatrix}
  \nn
\\
&\phantom{ \sum_{  j_{23} } }
\times
\begin{Bmatrix}
   j_{23} & k_1 & k_3 \\
   \ell &  K_3 & K_1
  \end{Bmatrix}
\begin{Bmatrix}
j_2 & j_3 & j_{23} \\
k_3 & k_1 & k_2 
\end{Bmatrix}
\,. \label{eq:Amplitudes-tri,tc-6j}
\end{align}
As done in \cite{Chen:2022rty}, this identity is related to Biedenharn-Elliot identity \cite{Bonzom:2009zd}.

\smallskip

Let us compute the Schmidt eigenvalues from $| \psi_{ tc, \{ j_i,k_i\} } \ra$. Similarly, starting with initial coarse-grained spin network state $| \psi_{ tc, \{ j_i,k_i\} } \ra$ as eq.(\ref{eq:triGraph-canGraph}), we consider the evolution generated by the loop holonomy operator $\wh{\chi_{\ell} }$,
\beq
\wh{\chi_{\ell} } :
&&
\textrm{Inv}_{\SU(2)} \Big( \cV_{j_1} \otimes \cV_{k_{1} } \otimes \cV_{k_{3}} \Big)
\otimes
\textrm{Inv}_{\SU(2)} \Big( \cV_{j_2} \otimes \cV_{j_3} \otimes \cV_{k_{3} } \otimes \cV_{k_{1}} \Big)
\nn \\
&&\to
\bigoplus_{ K_1,K_3 }
\textrm{Inv}_{\SU(2)} \Big( \cV_{j_1} \otimes \cV_{K_{1} } \otimes \cV_{K_{3}} \Big)
\otimes
\textrm{Inv}_{\SU(2)} \Big( \cV_{j_2} \otimes \cV_{j_3} \otimes \cV_{K_{3} } \otimes \cV_{K_{1}} \Big)
\,.
\nn
\eeq
Following the same logic as with the triangle graph, we compute the evolution of the state, truncated to leading order and properly normalized,
\be
| \Psi_{tc,\{j_i,k_i\}} (\ell,t) \ra
=
\f{ | \Psi_{tc,\{j_i,k_i\}} \ra - \ri t
\sum_{ K_1,K_3 }\sum_{ j_{23} }
[ Z(tc)_{\ell} ]^{ j_1,j_{23}, K_1,K_3 }_{ \phantom{ j_{1}, J_{23}, K_1,K_3 } j_{1}, j_{23}, k_1, k_3 }
| \Psi_{tc,\{j_i,K_i\}} \ra }
{ \sqrt{ 1+t^2 \sum_{s=0}^{2\ell} \sum_{ j_{23} } [ Z(tc)_{s} ]^{ j_1,j_{23}, k_1,k_3 }_{ \phantom{ j_{1}, j_{23}, k_1,k_3 } j_{1}, j_{23}, k_1, k_3 } (\cU_{ \{j,k\}_{B'}}^{k_2,j_{23} } )^2 }   }
\,. \label{eq:TruncatedState-triGraph}
\ee
Readers should notice that notation $| \Psi_{tc,\{j_i,K_i\}} \ra \equiv | \Psi_{sc,\{j_i,K_1,k_2,K_3\}} \ra$, i.e. the spin $k_2$ is now fixed under the action of $\wh{\chi_{\ell} }$ on the candy graph.
The transition matrix is given by eq.(\ref{eq:Amplitudes-tcGraph}), and the normalization factor $N_{tc, \{j_{i},k_{i}\} }(\ell,t)=N_{tri, \{j_{i},k_{i}\} }(\ell,t)$ due to eq.(\ref{eq:Amplitudes-tri,tc-6j}).

\smallskip

Repeat the same procedure for entanglement $| \Psi_{tc,\{j_i,k_i\}} (\ell,t) \ra$. The density matrix for final state on the coarse-grained graph is $\rho_{tc}(t)=| \Psi_{tc,\{j_i,k_i\}} (\ell,t) \ra\la \Psi_{tc,\{j_i,k_i\}} (\ell,t) |$. The reduced density matrix $\rho_{tc_{A} }(t) \in \textrm{End} (\cH_{A})$ is obtained via partial tracing over $\cH_{B'}$, which can be done via choosing $\big\{ | \{ j_2, j_3, k_3, k_1 \}, j_{23} \ra_{B'} \big\}$ orthonormal basis. 
So the reduced density matrix $\rho_{tc_{A} }(t)$ reads:
\beq
\rho_{tc_{A} }(t)
&=&
\sum_{K_1,K_3,j_{23} } \la \{ j_2, j_3, K_3, K_1 \}, j_{23} | \rho_{tc}(t)  | \{ j_2, j_3, K_3, K_1 \}, j_{23} \ra_{B'}
\nn \\
&=&
\bigg(
t^2 \sum_{K_1, K_3,j_{23} }
\Big([ Z(tc)_{\ell} ]^{ j_1,j_{23}, K_1,K_3 }_{ \phantom{ j_{1}, J_{23}, K_1,K_3 } j_{1}, j_{23}, k_1, k_3 } \cU_{ \{j,k\}_{B'}}^{k_2,j_{23} }  \Big)^2
| j_1,K_1,K_3 \ra \la j_1,K_1,K_3 |_{A}
\nn
\\
&&\phantom{ \bigg( }
+| j_1,k_1,k_3 \ra \la j_1,k_1,k_3 |_{A}
\bigg) \times \f{1}{ N_{ tc, \{j_{i},k_{i}\} }(\ell,t) }
\label{eq:ReducedDM-tcGraph}
\eeq
where the $\Big( \cU_{ \{j,k\}_{B'}}^{k_2,j_{23} } \Big)^2$ is the probability distribution for $j_{23}$ seen from another channel to the intertwiner $| \{ j_2,j_3,k_3,k_1 \}, k_2 \ra_{B'}$,
\be
p(j_{23})=
\Big( \cU_{ \{j,k\}_{B'}}^{k_2,j_{23} } \Big)^2
=
(2k_2+1)(2j_{23}+1)
\begin{Bmatrix}
j_2 & j_3 & j_{23} \\
k_3 & k_1 & k_2 
\end{Bmatrix}^2
\,.
\ee
Again, reduced density matrix (\ref{eq:ReducedDM-tcGraph}) is a diagonal matrix with respect to $\{ | j_1,K_1,K_3 \ra\}$ basis for $\cH_{A}$, thus the eigenvalues of $\rho_{tc_{A} }(t)$ are read
\beq
\lambda_{ \rho_{tc_A } }[K_1,K_3]
&=&
\bigg(
t^2 \sum_{ j_{23} }
  \begin{Bmatrix}
   j_{23} & j_2 & j_3 \\
   k_2 &  k_3 & k_1
  \end{Bmatrix}^2
\begin{Bmatrix}
   j_1 & k_1 & k_3 \\
   \ell &  K_3 & K_1
  \end{Bmatrix}^2
  \begin{Bmatrix}
   j_{23} & k_1 & k_3 \\
   \ell &  K_3 & K_1
  \end{Bmatrix}^2 \prod_{i=1}^{3}(2k_i+1)
  \nn \\
  &&
  \phantom{\bigg(}\times
   (2K_1+1)(2K_3+1)(2j_{23}+1)
   +\delta^{K_1}_{k_1} \, \delta^{K_3}_{k_3}
   \bigg) \times \f{ 1 }{ N_{tc, \{j_{i}\}, \{k_{i}\} }(\ell,t) }
\,. \nn
\\
\label{eq:eigenvalues-Truncated-tcGraph}
\eeq
One can show $\lambda_{ \rho_{tri_A } }[K_1,K_3]=\lambda_{ \rho_{tc_A } }[K_1,K_3]$ with eq.(\ref{eq:Amplitudes-tri,tc-6j}). Hence the entanglement excitation is preserved under the coarse-graining.


\subsection{Square graph}
We consider the holonomy operator acting on the loop of square graph Fig.\ref{fig:squGraph-canGraph}. We compute explicitly the bipartite entanglement between $\Gamma_{1}$ and $\Gamma_{2}$ where $\Gamma_{1}$ is made of vertices $A,D$ and their adjacent edges, $\Gamma_{2}$ is made of vertices $B,C$ and their adjacent edges. Then we show that the bipartite entanglement can be studied by its coarse-grained graph as Fig.\ref{fig:squGraph-canGraph}.
\begin{figure}[htb]
\centering
\begin{tikzpicture}[scale=0.7]

\coordinate (A) at (-0.75,0.75);
\coordinate (B) at (0.75,0.75);
\coordinate (C) at (0.75,-0.75);
\coordinate (D) at (-0.75,-0.75);

\draw[thick] (A) --++(135:0.7) ++ (135:0.35) node {$j_1$};
\draw[thick] (B) --++(45:0.7) ++ (45:0.35) node {$j_2$};
\draw[thick] (C) --++(-45:0.7) ++ (-45:0.35) node {$j_3$};
\draw[thick] (D) --++(-135:0.7) ++ (-135:0.35) node {$j_4$};

\draw[blue,thick] (A) -- node[above=1,midway] {$k_1$} (B) node[scale=0.7,blue] {$\bullet$} -- node[right=1,midway] {$k_2$} (C) node[scale=0.7,blue] {$\bullet$} --node[below=1,midway] {$k_3$} (D) node[scale=0.7,blue] {$\bullet$} -- node[left=1,midway] {$k_4$} (A) node[scale=0.7,blue] {$\bullet$};

\draw[->,>=stealth,very thick] (2,0) -- (3.5,0);

\draw (4.75,0) node {$\sum_{j_{14},j_{23}}$};

\coordinate (O1) at (7.5,0);
\coordinate (O2) at (9.5,0);

\draw[thick] (O1) -- ++ (-1,0) node[above,midway] {$j_{14}$} -- ++ (135:1) node[above]{$j_1$};
\draw[thick] (O1) ++(-1,0) node[scale=0.7] {$\bullet$} -- ++ (-135:1) node[below] {$j_4$};

\draw[thick] (O2) -- ++ (1,0) node[above,midway] {$j_{23}$} -- ++ (45:1) node[above]{$j_2$};
\draw[thick] (O2) ++(1,0) node[scale=0.7] {$\bullet$} -- ++ (-45:1) node[below] {$j_3$};

\draw[blue,thick,in=115,out=65,rotate=0] (O1) to node[above,midway] {$k_1$} (O2) node[scale=0.7] {$\bullet$} to [out=245,in=-65] node[below] {$k_3$} (O1) node[scale=0.7] {$\bullet$};

\draw[->,>=stealth,very thick] (11.5,0) -- (13,0);

\coordinate (A1) at (14,0);
\coordinate (A2) at (16,0);

\draw[thick] (A1) -- ++ (135:1) node[above]{$j_1$};
\draw[thick] (A1) -- ++ (-135:1) node[below] {$j_4$};

\draw[thick] (A2) -- ++ (45:1) node[above]{$j_2$};
\draw[thick] (A2) -- ++ (-45:1) node[below] {$j_3$};

\draw[blue,thick,in=115,out=65,rotate=0] (A1) to node[midway,above] {$k_1$} (A2) node[scale=0.7] {$\bullet$} to[out=245,in=-65] node[midway,below] {$k_3$} (A1) node[scale=0.7] {$\bullet$};

\draw (A1) ++ (-90:0.5) node[scale=0.8] {$A'$};
\draw (A2) ++ (-90:0.5) node[scale=0.8] {$B'$};

\end{tikzpicture}
\caption{Coarse-graining square graph to candy graph.}
\label{fig:squGraph-canGraph}
\end{figure}
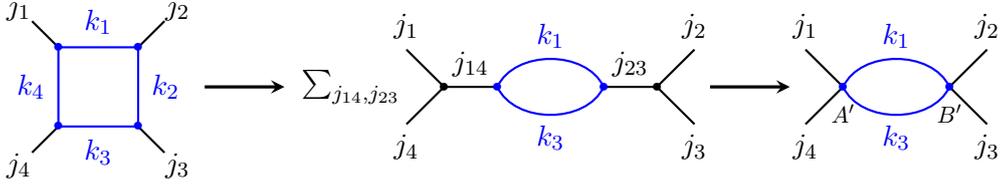

Starting with an initial spin network basis state $| \Psi_{squ, \{j_i,k_i\} } \ra=| j_i,k_4,k_1 \ra_{A} \otimes | j_2,k_1,k_2 \ra_{B} \otimes | j_3, k_2,k_3\ra_{C} \otimes | j_4,k_3,k_4\ra_{D}$, we consider the evolution generated by the loop holonomy operator $\wh{\chi_{\ell} }$,
\be
\wh{\chi_{\ell} }
:\bigotimes_{i=1}^{4}
\textrm{Inv}_{\SU(2)} \Big( \cV_{j_i} \otimes \cV_{k_{i} } \otimes \cV_{k_{i-1} } \Big)
\to
\bigoplus_{ \{K_i\} }
\bigotimes_{i=1}^{4}
\textrm{Inv}_{\SU(2)} \Big( \cV_{j_i} \otimes \cV_{K_{i}} \otimes \cV_{K_{i-1} } \Big)
\,.
\nn
\ee
For  infinitesimal time $t \to 0$, the unitarity evolution operator is $e^{ - \ri t \, \wh{\chi_{\ell} } } = \id - \ri  t \, \wh{\chi_{\ell} } + \cO(t^2)$, which acts as
 \beq
&&
e^{ - \ri t \, \widehat{\chi_{\ell} } } | \Psi_{squ,\{j_i,k_i\}} \ra
=
e^{ - \ri t \, \widehat{\chi_{\ell} } } | j_1,k_1,k_4 \ra_{A} \otimes | j_2,k_2,k_1 \ra_{B} \otimes | j_3,k_3,k_2 \ra_{C} \otimes | j_4,k_4,k_3 \ra_{D}
\nn \\
&=&
 \sum_{K_i=| k_i-{\ell} |}^{k_i+{\ell} }
 \left( \delta^{K_1}_{k_1}\delta^{K_2}_{k_2}\delta^{K_3}_{k_3}\delta^{K_4}_{k_4}  - \ri t {[} Z(squ)_{\ell} {]}^{ \{j_{i}, K_{i} \} }_{ \phantom{ \{j_{i}, K_{i} \} } \{j_{i}, k_{i} \} } \right) | j_1,K_1,K_4 \ra_{A} \otimes | j_2,K_2,K_1 \ra_{B}
\nn \\
&&
\phantom{  \sum_{K_i=| k_i-{\ell} |}^{k_i+{\ell} } }
\otimes | j_3,K_3,K_2 \ra_{C} \otimes | j_4,K_4,K_3 \ra_{D}
+ \cO(t^2)
\,, \label{eq:ShortTime-BasisState-squGraph}
\eeq
where $| j_i,K_1,K_4 \ra_{A} \in \cH_{A}$, $| j_2,K_2,K_1 \ra_{B} \in \cH_{B}$, $| j_3,K_3,K_2 \ra_{C} \in \cH_{C}$, $| j_4,K_4,K_3 \ra_{D} \in \cH_{D}$ denote the intertwiners living at respect trivalent vertex.
The transition matrix $[Z(squ)_{\ell}]$ is expressed in terms of $6j$-symbols according to the general formula (\ref{eq:Amplitudes-6jsymbols}) for the action of loop holonomy operator $\wh{ \chi_{\ell} }$
\begin{align}
{[} Z(squ)_{\ell} {]}^{ \{j_{i}, K_{i} \} }_{ \phantom{ \{j_{i}, K_{i} \} } \{j_{i}, k_{i} \} }
=&
(-1)^{\sum_{i=1}^{4}( j_i+k_i+K_i + \ell)} \begin{Bmatrix}
   j_1 & k_1 & k_4 \\
   \ell &  K_4 & K_1
  \end{Bmatrix}\begin{Bmatrix}
   j_2 & k_2 & k_1 \\
   \ell &  K_1 & K_2
  \end{Bmatrix}
  \nn
  \\
  &\times
  \begin{Bmatrix}
   j_3 & k_3 & k_2 \\
   \ell & K_2 & K_3
  \end{Bmatrix}
  \begin{Bmatrix}
   j_4 & k_4 & k_3 \\
   \ell & K_3 & K_4
  \end{Bmatrix}
  \prod_{i=1}^{4} \sqrt{ (2k_i+1)(2K_i+1) }
  \,.
\end{align}
Again, the matrix elements are all real numbers. Following the same logic as with the triangle graph, we compute the evolution of the state, truncated to leading order and properly normalized,
\be
| \Psi_{squ,\{j_i,k_i\}} (\ell,t) \ra
=
\f{ | \Psi_{squ,\{j_i,k_i\}} \ra - \ri t
\sum_{ \{K_i\} }
{[} Z(squ)_{\ell} {]}^{ \{j_{i}, K_{i} \} }_{ \phantom{ \{j_{i}, K_{i} \} } \{j_{i}, k_{i} \} }
| \Psi_{squ,\{j_i,K_i\}} \ra }
{ \sqrt{ 1+t^2 \sum_{s=0}^{2\ell} {[} Z(squ)_{s} {]}^{ \{j_{i}, k_{i} \} }_{ \phantom{ \{j_{i}, k_{i} \} } \{j_{i}, k_{i} \} } }   }
\,, \label{eq:NormalizedShortTimeState-candygraph}
\ee
and the normalization factor can then be evaluated to
\begin{align}
N_{ squ, \{j_{i}\}, \{k_{i}\} }(\ell,t)
=&
1+t^2 \sum_{s=0}^{2\ell} (-1)^{\sum_{i=1}^{4}( j_i+2k_i + s)} \begin{Bmatrix}
   j_1 & k_4 & k_1 \\
   s &  k_1 & k_4
  \end{Bmatrix}\begin{Bmatrix}
   j_2 & k_1 & k_2 \\
   s &  k_2 & k_1
  \end{Bmatrix} \begin{Bmatrix}
   j_3 & k_2 & k_3 \\
   s & k_3 & k_2
  \end{Bmatrix}
  \nn
  \\
  &\phantom{1+}\times
  \begin{Bmatrix}
   j_4 & k_3 & k_4 \\
   s & k_4 & k_3
  \end{Bmatrix}
  \prod_{i=1}^{4} (2k_i+1)
  \,.
\end{align}

\smallskip

We now look at the reduced state from the pure state (\ref{eq:NormalizedShortTimeState-candygraph}). The initial spin network state is a basis state, thus is fully separable with vanishing entanglement. The final state is given by density matrix $\rho_{squ}(t)=| \Psi_{squ,\{j_i,k_i\}} (\ell,t) \ra\la \Psi_{squ,\{j_i,k_i\}} (\ell,t) |$. The reduced density matrix $\rho_{squ_{AD} } \in \textrm{End} (\cH_{A} \otimes \cH_{D} )$ is obtained via partial tracing over $\cH_{B} \otimes \cH_{C}$, which is done via choosing orthonormal basis $| \{j_2,j_3,K_1,K_2,K_3\} \ra_{BC} \equiv | j_2,K_1,K_2 \ra_{B} \otimes | j_3,K_2,K_3 \ra_{C}$ to implement $\sum_{K_1,K_2,K_3} \la \{j_2,j_3,K_1,K_2,K_3\} | \rho_{squ}(t) | \{j_2,j_3,K_1,K_2,K_3\} \ra_{BC}$, so the reduced density matrix $\rho_{squ_{AD}  }(t)$ reads:
\begin{align}
\rho_{squ_{AD}  }(t)
=&
\f{1}{ N_{squ, \{j_{i}\}, \{k_{i}\} }(\ell,t) }
\bigg(
| j_1,k_1,k_4 \ra \la j_1,k_1,k_4 |_{A} \otimes | j_4,k_4,k_3 \ra \la j_4,k_4,k_3 |_{D}
\nn
\\
&-\ri t \sum_{K_4}
{[} Z(squ)_{\ell} {]}^{ \{j_{i}, k_{r}, K_4 \} }_{ \phantom{ \{j_{i}, k_{r}, K_4 \} } \{j_{i}, k_{r}, k_4 \} }
| j_1,k_1,K_4 \ra \la j_1,k_1,k_4 |_{A} \otimes | j_4,K_4,k_3 \ra \la j_4,k_4,k_3 |_{D}
\nn
\\
&+\ri t \sum_{K_4}
{[} Z(squ)_{\ell} {]}^{ \{j_{i}, k_{r}, K_4 \} }_{ \phantom{ \{j_{i}, k_{r}, K_4 \} } \{j_{i}, k_{r}, k_4 \} }
| j_1,k_1,k_4 \ra \la j_1,k_1,K_4 |_{A} \otimes | j_4,k_4,k_3 \ra \la j_4,K_4,k_3 |_{D}
\nn
\\
&+t^2 \sum_{K_4, K'_4} \sum_{ K_1,K_2,K_3 }
{[} Z(squ)_{\ell} {]}^{ \{j_{i}, K_{r}, K_4 \} }_{ \phantom{ \{j_{i}, K_{r}, K_4 \} } \{j_{i}, k_{r}, k_4 \} }
| j_1,K_1,K_4 \ra \la j_1,K_1,K'_4 |_{A} 
\nn \\
&
\phantom{ +t^2 \sum_{K_4, K'_4} \sum_{ K_1,K_2,K_3 } }
\otimes | j_4,K_4,K_3 \ra \la j_4,K'_4,K_3 |_{D}
{[} Z(squ)_{\ell} {]}^{ \{j_{i}, K_{r}, K'_4 \} }_{ \phantom{ \{j_{i}, K_{r}, K'_4 \} } \{j_{i}, k_{r}, k_4 \} }
\bigg)
\,,
\end{align}
where $r=1,2,3$.

\smallskip

Now let us look at the reduced state from the point of view of coarse-graining. Considering bipartition $\cH_A \otimes \cH_D$ and $\cH_B \otimes \cH_C$, we gauge fix $g_2 \to \id$ and $g_4 \to \id$, then contract vertices $A,D$ along $e_{4}$ and vertices $B,C$ along $e_{2}$, to acquire a two 4-valent vertices with respect spins $j_1,j_4,k_1,k_3$ and $j_2,j_3,k_1,k_3$ as Fig.\ref{fig:squGraph-canGraph}. The spin network coarse-grained state from the square graph to the candy graph is given by
\be
| \psi_{ sc, \{ j_i,k_i\} } \ra
=
\sum_{j_{14} } \cU_{ \{j,k\}_{A'}}^{k_4,j_{14} }  | \{j_1, j_4,k_1, k_3 \}, j_{14} \ra_{A'}
\otimes
\sum_{j_{23} }
\cU_{ \{j,k\}_{B'}}^{k_2,j_{23} }  | \{ j_2, j_3, k_2, k_3 \}, j_{23} \ra_{B'}
\,.
\label{eq:squGraph-canGraph}
\ee
It has intertwiner superposition for $\cH_{A'}$ and $\cH_{B'}$ with respect to internal spins $j_{14}$ and $j_{23}$.
The $| \{j_1, j_4,k_1, k_3 \}, j_{14} \ra_{A'}$ and $| \{ j_2, j_3, k_2, k_3 \}, j_{23} \ra_{B'}$ are intertwiners for the two respect 4-valent vertices $A'$ and $B'$, and $\cU_{ \{j,k\}_{A'} }^{k_4,j_{14} }$ and $\cU_{ \{j,k\}_{B'}}^{k_2,j_{23} }$ are unitaries for respect $\cH_{A'}$ and $\cH_{B'}$,
\be
\begin{aligned}
\cU_{ \{j,k\}_{A'} }^{k_4,j_{14} } 
=&
(-1)^{2k_4+k_1-k_3+j_1+j_4} \sqrt{ (2k_4+1)(2j_{14}+1) }
\begin{Bmatrix}
j_1 & j_4 & j_{14} \\
k_3 & k_1 & k_4
\end{Bmatrix}
\,,
\\
\cU_{ \{j,k\}_{B'}}^{k_2,j_{23} } 
=&
(-1)^{2k_2+k_3-k_1+j_2+j_3} \sqrt{ (2k_2+1)(2j_{23}+1) }
\begin{Bmatrix}
j_2 & j_3 & j_{23} \\
k_3 & k_1 & k_2 
\end{Bmatrix}
\,.
\end{aligned}
\ee
They are real numbers. The unitarity of $\cU_{ \{j,k\}_{A'} }^{k_4,j_{14} } $ and $\cU_{ \{j,k\}_{B'}}^{k_2,j_{23} } $ are verified by the orthogonality of $6j$-symbols,
\be
\begin{aligned}
\sum_{j_{14} } \cU_{ \{j,k\}_{A'} }^{k_4,j_{14} }  \cU_{ \{j,k\}_{A'} }^{\tk_4,j_{14} }
=&
\sum_{j_{14} } \sqrt{ (2k_4+1) (2\tk_4+1) } (2j_{14}+1)
\begin{Bmatrix}
j_1 & j_4 & j_{14} \\
k_3 & k_1 & k_4
\end{Bmatrix}
\begin{Bmatrix}
j_1 & j_4 & j_{14} \\
k_3 & k_1 & \tk_4
\end{Bmatrix}
=\delta_{k_4 \tk_4}
\,,
\\
\sum_{j_{23} } \cU_{ \{j,k\}_{B'}}^{k_2,j_{23} }  \cU_{ \{j,k\}_{B'}}^{\tk_2,j_{23} }
=&
\sum_{j_{23} } \sqrt{ (2k_2+1) (2\tk_2+1) }(2j_{23}+1)
\begin{Bmatrix}
j_2 & j_3 & j_{23} \\
k_3 & k_1 & k_2 
\end{Bmatrix}
\begin{Bmatrix}
j_2 & j_3 & j_{23} \\
k_3 & k_1 & \tk_2 
\end{Bmatrix}
=\delta_{k_2 \tk_2}
\,.
\end{aligned}
\ee
As eq.(\ref{eq:unitary-tc}), the unitarity requires the two $\cU$ having same $\{ j , k \}$ and for the internal labels,
\be
\sum_{k_4 } \cU_{ \{j,k\}_{A'}}^{k_4,j_{14} }  \cU_{ \{j,k\}_{A'}}^{k_4,\tl{j}_{14} } 
=\delta_{j_{14} \tl{j}_{14} }
\,, \qquad
\sum_{k_2 } \cU_{ \{j,k\}_{B'}}^{k_2,j_{23} }  \cU_{ \{j,k\}_{B'}}^{k_2,\tl{j}_{23} } 
=\delta_{j_{23} \tl{j}_{23} }
\,.
\ee
The unitaries $\cU_{A'}$ and $\cU_{B'}$ link the transition matrices $Z(squ)$ and $Z(sc)$,
\be
{[} Z(squ)_{\ell} {]}^{ \{j_{i}, K_{i} \} }_{ \phantom{ \{j_{i}, K_{i}\} } \{j_{i}, k_{i} \} }
=
\sum_{ j_{14} } \sum_{  j_{23} }
\cU_{ \{j,K\}_{A'} }^{K_4,j_{14} }
\cU_{ \{j,K\}_{B'}}^{K_2,j_{23} }
[ Z(sc)_{\ell} ]^{ j_{14},j_{23}, K_1,K_3 }_{ \phantom{ j_{14}, j_{23}, K_1,K_3 } j_{14}, j_{23}, k_1, k_3 }
\cU_{ \{j,k\}_{A'} }^{k_4,j_{14} } 
\cU_{ \{j,k\}_{B'}}^{k_2,j_{23} } 
  \,, \label{eq:Amplitudes-SquCan}
\ee
where the transition matrix for $Z(sc)$ is given by
\begin{align}
[ Z(sc)_{\ell} ]^{ J_{14},J_{23}, K_1,K_3 }_{ \phantom{ J_{14}, J_{23}, K_1,K_3 } j_{14}, j_{23}, k_1, k_3 }
=&
(-1)^{j_{14}+j_{23}+K_1+k_1+K_3+k_3+2\ell}
\begin{Bmatrix}
   j_{14} & k_1 & k_3 \\
   \ell &  K_3 & K_1
  \end{Bmatrix}
\begin{Bmatrix}
   j_{23} & k_1 & k_3 \\
   \ell &  K_3 & K_1
  \end{Bmatrix}
\nn  \\
  & \times
  \sqrt{(2k_1+1)(2K_1+1)(2k_3+1)(2K_3+1)}
  \delta_{J_{14}j_{14}}\delta_{J_{23}j_{23}}
  \,, \label{eq:Amplitudes-scGraph}
\end{align}
or explicitly, the transformation is expressed as below identity in terms of $6j$-symbols:
\begin{align}
&
(-1)^{\sum_{i=1}^{4}( j_i+k_i+K_i + \ell)} \begin{Bmatrix}
   j_1 & k_1 & k_4 \\
   \ell &  K_4 & K_1
  \end{Bmatrix}\begin{Bmatrix}
   j_2 & k_2 & k_1 \\
   \ell &  K_1 & K_2
  \end{Bmatrix} \begin{Bmatrix}
   j_3 & k_3 & k_2 \\
   \ell & K_2 & K_3
  \end{Bmatrix} \begin{Bmatrix}
   j_4 & k_4 & k_3 \\
   \ell & K_3 & K_4
  \end{Bmatrix}
\nn
\\
=&
\sum_{ j_{14} } \sum_{  j_{23} }
(-1)^{j_{14}+j_{23} + k_1 + K_1 + k_3 + K_3 +2\ell +2k_2+2k_4-2K_2-2K_4 } (2j_{14}+1)(2j_{23}+1)
\begin{Bmatrix}
j_1 & j_4 & j_{14} \\
K_3 & K_1 & K_4
\end{Bmatrix}
\nn
\\
&\times
\begin{Bmatrix}
j_2 & j_3 & j_{23} \\
K_3 & K_1 & K_2 
\end{Bmatrix}
\begin{Bmatrix}
   j_{14} & k_1 & k_3 \\
   \ell &  K_3 & K_1
  \end{Bmatrix}
\begin{Bmatrix}
   j_{23} & k_1 & k_3 \\
   \ell &  K_3 & K_1
  \end{Bmatrix}
  \begin{Bmatrix}
j_1 & j_4 & j_{14} \\
k_3 & k_1 & k_4
\end{Bmatrix}
\begin{Bmatrix}
j_2 & j_3 & j_{23} \\
k_3 & k_1 & k_2 
\end{Bmatrix}
\,.
\end{align}

\smallskip

Let us look at the reduced state from the point of view of coarse-graining. For the coarse-grained state, the evolution is generated by the loop holonomy operator $\wh{\chi_{\ell} }$ along path $[A' \overset{ e_{3} }{\to} B' \overset{ e_{1} }{\to} A']$,
\beq
\wh{\chi_{\ell} } :
&&
\textrm{Inv}_{\SU(2)} \Big( \cV_{j_1} \otimes \cV_{j_4} \otimes \cV_{k_{3} } \otimes \cV_{k_{1}} \Big)
\otimes
\textrm{Inv}_{\SU(2)} \Big( \cV_{j_2} \otimes \cV_{j_3} \otimes \cV_{k_{3} } \otimes \cV_{k_{1}} \Big)
\nn \\
&&\to
\bigoplus_{ K_1,K_3 }
\textrm{Inv}_{\SU(2)} \Big( \cV_{j_1} \otimes \cV_{j_4} \otimes \cV_{K_{3} } \otimes \cV_{K_{1}} \Big)
\otimes
\textrm{Inv}_{\SU(2)} \Big( \cV_{j_2} \otimes \cV_{j_3} \otimes \cV_{K_{3} } \otimes \cV_{K_{1}} \Big)
\,.
\nn
\eeq
Following the same logic, the evolution of the state is truncated to leading order and properly normalized,
\be
| \Psi_{sc,\{j_i,k_i\}} (\ell,t) \ra
=
\f{ | \Psi_{sc,\{j_i,k_i\}} \ra - \ri t
\sum_{ K_1,K_3 }\sum_{ j_{14}, j_{23} }
[ Z(sc)_{\ell} ]^{ j_{14},j_{23}, K_1,K_3 }_{ \phantom{ j_{14}, j_{23}, K_1,K_3 } j_{14}, j_{23}, k_1, k_3 }
| \Psi_{sc,\{j_i,K_i\}} \ra }
{ \sqrt{ N_{sc, \{j_{i},k_{i}\} }(\ell,t) }   }
\,. \label{eq:TruncatedState-triGraph}
\ee
Note that here notation $| \Psi_{sc,\{j_i,K_i\}} \ra \equiv | \Psi_{sc,\{j_i,K_1,k_2,K_3,k_4\}} \ra$, because the spins $k_2$ and $k_4$ are fixed under the action of $\wh{\chi_{\ell} }$ on the candy graph.
The transition matrix is given by eq.(\ref{eq:Amplitudes-scGraph}). The normalization factor $N_{sc, \{j_{i},k_{i}\} }(\ell,t)=N_{squ, \{j_{i},k_{i}\} }(\ell,t)$ is again due to eq.(\ref{eq:Amplitudes-scGraph}).

Repeat the same procedure. The density matrix for final state on the coarse-grained graph is $\rho_{sc}(t)=| \Psi_{sc,\{j_i,k_i\}} (\ell,t) \ra\la \Psi_{sc,\{j_i,k_i\}} (\ell,t) |$. The reduced density matrix $\rho_{sc_{A'} }(t) \in \textrm{End} (\cH_{A'})$ is obtained via partial tracing over $\cH_{B'}$, which is done via choosing orthonormal basis $\big\{ | \{ j_2, j_3, k_3, k_1 \}, j_{23} \ra_{B'} \big\}$, so the reduced density matrix $\rho_{sc_{A'} }(t)$ reads:
\beq
\rho_{sc_{A'} }(t)
&=&
\sum_{K_1,K_3,j_{23} } \la \{ j_2, j_3, K_3, K_1 \}, j_{23} | \rho_{sc}(t)  | \{ j_2, j_3, K_3, K_1 \}, j_{23} \ra_{B'}
\nn \\
&=&
\f{1}{ N_{ sc, \{j_{i},k_{i}\} }(\ell,t) }
\bigg(
\sum_{j_{14}, \tl{j}_{14} } \cU_{ \{j,k\}_{A'}}^{k_4,j_{14} }
| \{ j_1,j_4,k_3,k_1\}, j_{14} \ra \la \{ j_1,j_4,k_3,k_1\}, \tl{j}_{14} |_{A'} \cU_{ \{j,k\}_{A'}}^{k_4,\tl{j}_{14} }
\nn
\\
&&-\ri t \sum_{j_{14}, \tl{j}_{14} } \sum_{j_{23} }
\Big( \cU_{ \{j,k\}_{B'}}^{k_2,j_{23} } \Big)^2
\cU_{ \{j,k\}_{A'}}^{k_4,j_{14} }
[ Z(sc)_{\ell} ]^{ j_{14},j_{23}, k_1,k_3 }_{ \phantom{ j_{14}, j_{23}, k_1,k_3 } j_{14}, j_{23}, k_1, k_3 }
\cU_{ \{j,k\}_{A'}}^{k_4,\tl{j}_{14} }
\nn
\\
&&\phantom{ -\ri t \sum_{j_{14}, \tl{j}_{14} } \sum_{j_{23} } }
| \{ j_1,j_4,k_3,k_1\}, j_{14} \ra \la \{ j_1,j_4,k_3,k_1\}, \tl{j}_{14} |_{A'}
\nn
\\
&&+\ri t \sum_{j_{14}, \tl{j}_{14} } \sum_{j_{23} }
\Big( \cU_{ \{j,k\}_{B'}}^{k_2,j_{23} } \Big)^2
\cU_{ \{j,k\}_{A'}}^{k_4,j_{14} }
[ Z(sc)_{\ell} ]^{ \tl{j}_{14},j_{23}, k_1,k_3 }_{ \phantom{ \tl{j}_{14}, j_{23}, K_1,K_3 } \tl{j}_{14}, j_{23}, k_1, k_3 }
\cU_{ \{j,k\}_{A'}}^{k_4,\tl{j}_{14} }
\nn
\\
&&\phantom{ -\ri t \sum_{j_{14}, \tl{j}_{14} } \sum_{j_{23} } }
| \{ j_1,j_4,k_3,k_1\}, j_{14} \ra \la \{ j_1,j_4,k_3,k_1\}, \tl{j}_{14} |_{A'}
\nn
\\
&&+ t^2 \sum_{ j_{14},\tl{j}_{14} }\sum_{K_1, K_3,j_{23} }
\Big( \cU_{ \{j,k\}_{B'}}^{k_2,j_{23} } \Big)^2
[ Z(sc)_{\ell} ]^{ j_{14},j_{23}, K_1,K_3 }_{ \phantom{ j_{14}, j_{23}, K_1,K_3 } j_{14}, j_{23}, k_1, k_3 }
[ Z(sc)_{\ell} ]^{ \tl{j}_{14},j_{23}, K_1,K_3 }_{ \phantom{ \tl{j}_{14}, j_{23}, K_1,K_3 } \tl{j}_{14}, j_{23}, k_1, k_3 }
\nn \\
&&
\phantom{ + t^2 \sum_{ j_{14},\tl{j}_{14} }\sum_{K_1, K_3,j_{23} } }
\cU_{ \{j,k\}_{A'}}^{k_4,j_{14} }
| \{ j_1,j_4,K_3,K_1\}, j_{14} \ra \la \{ j_1,j_4,K_3,K_1\}, \tl{j}_{14} |_{A'}
\cU_{ \{j,k\}_{A'}}^{k_4,\tl{j}_{14} }
\bigg) \,. 
\nn 
\eeq
Now we show the eigenvalues of $\rho_{squ_{AD} }(t)$ and $\rho_{sc_{A'} }(t)$ are identical.
With below alteration eq.(\ref{eq:Amplitudes-SquCan})
\be
[ Z(sc)_{\ell} ]^{ j_{14},j_{23}, K_1,K_3 }_{ \phantom{ j_{14}, j_{23}, K_1,K_3 } j_{14}, j_{23}, k_1, k_3 }
\cU_{ \{j,k\}_{A'} }^{k_4,j_{14} } 
\cU_{ \{j,k\}_{B'}}^{k_2,j_{23} }
=
\sum_{ K_2,K_4 }
\cU_{ \{j,K\}_{A'} }^{K_4,j_{14} }
\cU_{ \{j,K\}_{B'}}^{K_2,j_{23} }
{[} Z(squ)_{\ell} {]}^{ \{j_{i}, K_{i} \} }_{ \phantom{ \{j_{i}, K_{i}\} } \{j_{i}, k_{i} \} }
\,,
\ee
one can relate $\rho_{squ_{AD} }(t)$ to $\rho_{sc_{A'} }(t)$ by
\be
\rho_{squ_{AD} }(t)
=
M \rho_{sc_{A'} }(t) M^{\dagger}
\ee
where $M$ is a unitary map defined by
\beq
M&:&
\textrm{Inv}_{\SU(2)} \Big( \cV_{j_1} \otimes \cV_{j_{4}} \otimes \cV_{k_{3} } \otimes \cV_{k_{1} } \Big)
\nn \\
&&\to
\bigoplus_{K_4} \textrm{Inv}_{\SU(2)} \Big( \cV_{j_1} \otimes \cV_{K_1} \otimes \cV_{K_4} \Big) \otimes
\textrm{Inv}_{\SU(2)} \Big( \cV_{j_4} \otimes \cV_{K_{4} } \otimes \cV_{k_{3} } \Big)
\,, \nn \\
M
&=&
\sum_{ K_4} \sum_{ j_{14} }
\cU_{ \{j,k\}_{A'}}^{K_4,j_{14} }
\Big( | j_1,k_1,K_4 \ra_{A} \otimes | j_4,K_4,k_3 \ra_{D} \Big)
\la \{ j_1,j_4,k_3,k_1\}, j_{14} |_{A'}
\,.
\eeq
%
%
Hence, the entanglement excitation between $A,D$ and $B,C$ is exactly reflected in the entanglement excitation between $A'$ and $B'$.


\subsection{Path-dependency on simplest two-loop graph}
We consider the holonomy operator acting on the graph Fig.\ref{fig:PathChoices} either along path $j_1 \to k_a \to j_2$ or path $j_1 \to k_b \to  j_2$. We look at the bipartite entanglement between $\cH_{A}$ and $\cH_{B}\otimes\cH_{C}$. We show the entanglement excitation's dependency on the choices of path along which the holonomy operator acts. We also study the path-dependency from coarse-grained graph, converting the path-dependency to the spin-dependency of self-loop.

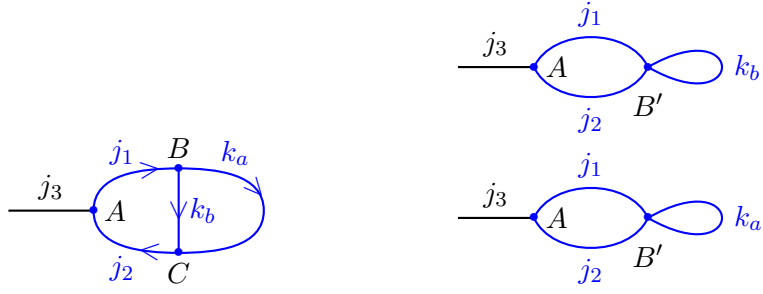
\begin{figure}
\centering
\begin{subfigure}[t]{0.4\linewidth}
\centering
	\begin{tikzpicture} [scale=0.7]
\coordinate (O) at (0,0);

\draw[blue,thick] (O) ++ (0,-1.6) coordinate (O1);
\draw[blue,thick] (O) ++ (-1.6,-0.8) coordinate (P);
\draw[blue,thick] (O) ++ (1.6,-0.8) coordinate (Q);
\draw [blue,thick] (O) to[out=180,in=90] node[near start,sloped] {$>$} node[midway,above] {$j_1$} (P) to[out=-90,in=180] node[near end,sloped] {$<$} node[midway,below] {$j_2$} (O1);
\draw[blue,thick] (O) -- node[midway,sloped] {$>$} node[midway,right] {$k_b$} (O1);

\draw [blue,thick] (O) to[out=0,in=90] node[near end,sloped] {$>$} node[midway,above] {$k_a$} (Q) to[out=-90,in=0] (O1);

\draw[thick] (P) -- node[midway,above] {$j_3$} ++(-1.6,0);

\draw[thick] (O) node[scale=0.7,blue] {$\bullet$} node[above] {$B$};
\draw[thick] (O1) node[scale=0.7,blue] {$\bullet$} node[below] {$C$};
\draw[thick] (P) node[scale=0.7,blue] {$\bullet$} node[right] {$A$};

\end{tikzpicture}
\caption{Holonomy operator acts along path $j_1 \to k_a \to j_2$ or path $j_1 \to k_b \to j_2$.
}
\label{fig:PathChoices}
\end{subfigure}
\begin{subfigure}[t]{0.4\linewidth}
\centering
\begin{tikzpicture} [scale=1]
\coordinate (A1) at (0,0);
\coordinate (A2) at (1.5,0);

\draw[thick] (A1) -- node[midway,above] {$j_3$} ++(-1,0);

\draw[blue,thick,in=115,out=65,rotate=0] (A1) to node[above,midway] {$j_1$} (A2) node[scale=0.7] {$\bullet$} to[out=245,in=-65] node[below,midway] {$j_2$} (A1) node[scale=0.7] {$\bullet$};

\draw (A1) ++ (0:0.3) node[scale=1] {$A$};
\draw (A2) ++ (-90:0.5) node[scale=1] {$B'$};

\draw[thick,in=+30,out=-30,scale=3,rotate=0,blue] (A2)  to [loop] node[midway,right=0.5,rotate=0] {$k_b$} (A2);

\coordinate (B1) at (0,-2);
\coordinate (B2) at (1.5,-2);

\draw[thick] (B1) -- node[midway,above] {$j_3$} ++(-1,0);

\draw[blue,thick,in=115,out=65,rotate=0] (B1) to node[above,midway] {$j_1$} (B2) node[scale=0.7] {$\bullet$} to[out=245,in=-65] node[below,midway] {$j_2$} (B1) node[scale=0.7] {$\bullet$};

\draw (B1) ++ (0:0.3) node[scale=1] {$A$};
\draw (B2) ++ (-90:0.5) node[scale=1] {$B'$};

\draw[thick,in=+30,out=-30,scale=3,rotate=0,blue] (B2)  to [loop] node[midway,right=0.5,rotate=0] {$k_a$} (B2);

\end{tikzpicture}
\caption{The coarse-grained graphs for path choices. Now the path dependency is reflected in the spin of self-loop.}
\label{fig:PathDependency-Loopy}
\end{subfigure}
\caption{The illustration for holonomy operator's path-dependency. The path-dependency can be converted to spin-dependency via coarse-graining method.
}
\label{fig:PathDependency-Example}
\end{figure}
For the sake of simplicity, we consider a simple spin network with spins $k_a=1,k_b=\f12$ and $j_1=j_2=\f12,j_3=1$, and set loop holonomy spin $\ell=\f12$.

Look at the action acting along the path $j_1 \to k_a \to j_2$. The spin network state can be labeled by spin-shifting $| j_1 j_2 K_1 \ra_a$. The initial state is $| \f12 \f12 1 \ra_a$, and following same logic, the final truncated state reads,
\be
| \psi_{a}(t) \ra
=
\f{ | \f12 \f12 1 \ra_{a}
+ \f{ \ri t }{2\sqrt{3} } | 0 1 \f12 \ra_{a}
+ \f{ \ri t }{2\sqrt{3} } | 1 0 \f12 \ra _{a}
+ \f{ \ri t }{3\sqrt{6} } | 1 1 \f12 \ra_{a}
+ \f{ 4 \ri t }{3\sqrt{3} } | 1 1 \f32 \ra_{a}
}{ \sqrt{ 1+\f{7 t^2}{9} } }
\,.
\ee
Likewise,  look at the action acting along the path $j_1 \to k_b \to j_2$. The spin network state can be labeled by spin-shifting $| j_1 j_2 K_2 \ra_{b}$. The initial state is $| \f12 \f12 \f12 \ra_{b}$, and the final truncated state reads,
\be
| \psi_{b}(t) \ra
=
\f{ | \f12 \f12 \f12 \ra_{b}
+ \f{ \sqrt{2} \ri t }{3 } | 0 1 1 \ra_{b}
+ \f{ \sqrt{2} \ri t }{3 } | 1 0 1 \ra _{b}
+ \f{ \sqrt{2} \ri t }{3 } | 1 1 0 \ra_{b}
+ \f{ 2\sqrt{2} \ri t }{ 3\sqrt{3} } | 1 1 1 \ra_{b}
}{ \sqrt{ 1+\f{26 t^2}{27} } }
\,.
\ee
Partial tracing over $\cH_{B} \otimes \cH_{C}$, the reduced states ${[} \rho_{a}(t) {]}_{A}$ and ${[} \rho_{b}(t) {]}_{A}$ are read,
\begin{align}
{[} \rho_{a}(t) {]}_{A}
=&
\f{ | \f12 \f12 \ra\la \f12 \f12 | + \f{ t^2 }{ 12 } | 0 1\ra\la 01 | + \f{ t^2 }{ 12 } | 1 0 \ra\la 1 0 | + \f{ 11 t^2 }{ 18 } | 1 1\ra\la 1 1 |
}{ 1+\f{7 t^2}{9} }
\,,
\\
{[} \rho_{b}(t) {]}_{A}
=&
\f{ | \f12 \f12 \ra\la \f12 \f12 | + \f{ 2 t^2 }{ 9 } | 0 1\ra\la 01 | + \f{ 2 t^2 }{ 9 } | 1 0 \ra\la 1 0 | + \f{ 14 t^2 }{ 27 } | 1 1\ra\la 1 1 |
}{ 1+\f{26 t^2}{27} }
\,.
\end{align}
The reduced states can be also derived from coarse-grained graph as Fig.\ref{fig:PathDependency-Loopy}. In this scenario, a loopy vertex affects the transition matrix through bouquet spin. In this sense, the loopy spin network in the problem amounts to being a 4-valent vertex. Based on that point, one can computes the eigenvalues (\ref{eq:eigenvalues-Truncated-tcGraph}) by setting the bouquet spin at $B'$.
As expected, working on the coarse-grained graph leads to same reduced states. Interestingly, the path-dependency is translated into the spin-dependency on the self-loop. The different entanglement excitations are interpreted by the different transition matrices, and the different transition matrices are reflected in the different bouquet spins.


\section{Conclusion \& outlook}
The present paper is dedicated to the study of coarse-graining spin entanglement at both kinematical and dynamical level in loop quantum gravity.
More precisely, we looked into the spin network entanglement on graph and coarse-grained graph, showing that arbitrary spin network state and the corresponding coarse-grained spin network state carry identical spin network entanglement. Furthermore, we show that the identical entanglement relation can be pushed towards the dynamical level as long as the evolution is generated by loop holonomy operator. That is, the dynamics of spin network entanglement can be also exactly reflected in the coarse-grained graph.

At the technical level, we utilize the point of view of bulk-boundary maps, understanding coarse-graining as unitary for dual boundary Hilbert space. This manner builds the relation between the spin network states based on graph and coarse-grained graph, guiding us to find the transformation between transition matrices for holonomy operator and the coarse-grained holonomy operator. Moreover, these unitaries are interpreted as local unitaries that preserves the entanglement amongst sub-networks. The conclusion is universal for any entanglement measure, according to the requirement that any entanglement measure should be invariant under local unitary.

\smallskip

We wish to shed light on the question about the quantum gravity degrees of freedom. In some sense, the coarse-graining feature of spin network entanglement implies that the graphical degrees of freedom can be considerably reduced.

On the other hand, the coarse-graining feature emphasizes the degrees of freedom of interfacing edges and self-loops. More precisely, the entanglement between spin sub-networks is exactly reflected in the entanglement between loopy spin sub-networks to which coarse-graining gives rise. Every loopy sub-network is made up by a single vertex plus boundary edges and self-loops (due to gauge-fixing). The bulk topology is reflected in the number of self-loops. When comes to spin network entanglement, the intertwiner at the single vertex encodes all the information about the sub-graph, associating with boundary spins and self-loop spins via recoupling.

\smallskip

Following the coarse-graining feature of spin network entanglement, we can consider a simple situation: staring with a closed spin network, partitioning it into a bipartite system, we then ask the maximal entanglement entropy between the subsystems. In line with our result, the entanglement entropy can be studied on the coarse-grained graph with two loopy vertices, and the maximal entanglement entropy is bounded from upper by the minor dimension of two loopy intertwiner spaces. Say $N$ the number of interfacing edges, and $u$ the loopy vertex with the number $L$ of self-loops. Suppose that the intertwiner space associated with $u$ has minor dimension, then the maximal entanglement entropy could be converted to a $\U(N+2L)$ problem for computing dimension \cite{Freidel:2009ck}. In particular, if one further assumes $BF$ dynamics as considered in \cite{Livine:2017xww}, then the self-loops could be removed at all, thus it becomes as 2-vertex model with $L=0$ \cite{Borja:2010gn,Aranguren:2022nzn}. In this case, the maximal entanglement entropy could meet area-entropy law at certain limit \cite{Freidel:2009ck}.

\smallskip

Since holonomy operators are the basic building blocks of the Hamiltonian dynamics of loop quantum gravity, this work gives a hint of non-local degrees of freedom description for the coarse-graining and dynamics on the quantum information carried by spin network states. We hope that further characterizing the various operators of loop quantum gravity through coarse-graining feature on the correlation and entanglement will allow to reformulate the precise mathematical framework for holography.

\acknowledgments

Q.C. especially thanks Etera Livine for many discussions, inspirations, encouragements and some of the key insights.
Q.C also thanks Xiangjing Liu, Shang-qiang Ning for discussions. 
Q.C. is financially supported by the China Scholarship Council.


\bibliographystyle{JHEP}
\bibliography{LQG}


%
%
%




\end{document}